\documentclass[letterpaper, DIV=16]{scrartcl}  	
\usepackage{amssymb, amsthm, mathtools,bbm,color}
\usepackage[square,sort,comma,numbers]{natbib}
\usepackage{enumerate}
\usepackage{dsfont,times,tikz} 
\usetikzlibrary{matrix,arrows.meta,shapes.geometric,calc}
\definecolor{Plum}{rgb}{.5,0,1}

\newcommand{\ran}{\operatorname{ran}}
\newcommand{\dist}{\operatorname{dist}}
\newcommand{\gap}{\operatorname{gap}}
\newcommand{\supp}{\operatorname{supp}}

\newcommand{\spa}{\operatorname{span}}

 \def\1{{\mathchoice {\mathrm{1\mskip-4mu l}} {\mathrm{1\mskip-4mu l}} %
		{\mathrm{1\mskip-4.5mu l}} {\mathrm{1\mskip-5mu l}}}}

\newcommand{\bC}{{\mathbb C}}
\newcommand{\bN}{{\mathbb N}}
\newcommand{\bZ}{{\mathbb Z}}
\newcommand{\cA}{{\mathcal A}}
\newcommand{\cH}{{\mathcal H}}
\newcommand{\braket}[2]{\langle {#1} \ | \ {#2}\rangle}
\newcommand{\ket}[1]{|{#1}\rangle}
\newcommand{\ketbra}[1]{\vert #1\rangle\langle #1\vert}
\newcommand{\vp}{\varphi}
\newcommand{\caG}{\mathcal{G}}

\newcommand{\caC}{\mathcal{C}}
\newcommand{\cP}{\mathcal{P}}
\newcommand{\cQ}{\mathcal{Q}}
\newcommand{\cG}{\mathcal{G}}
\newcommand{\cB}{\mathcal{B}}
\newcommand{\caB}{\mathcal{B}}
\newcommand{\cD}{\mathcal{D}}

\newcommand{\cS}{\mathcal{S}}
\newcommand{\cR}{\mathcal{R}}
\newcommand{\per}{\mathrm{per}}

\renewcommand{\mod}{\,\mathrm{mod}}
\renewcommand{\subset}{\subseteq}

\newcommand{\be}{\begin{equation}}
\newcommand{\ee}{\end{equation}}
\newcommand{\eq}[1]{(\ref{#1})}
  
\newtheorem{theorem}{Theorem}
\newtheorem{lem}[theorem]{Lemma}
\newtheorem{cor}[theorem]{Corollary}
\newtheorem{dfn}[theorem]{Definition}
\newtheorem{proposition}[theorem]{Proposition}
\newtheorem{assumption}[theorem]{Assumption}

\numberwithin{equation}{section}
\numberwithin{theorem}{section}
\usepackage{graphicx}				
\usepackage{amssymb}

\title{Spectral Gaps and Incompressibility in a $\pmb{ \nu = 1/3} $\\ Fractional Quantum Hall System}
\author{Bruno Nachtergaele, Simone Warzel, and Amanda Young}
\date{\small\today}							

\begin{document}
\maketitle

\minisec{Abstract} 
We study an effective Hamiltonian for the standard $\nu=1/3$ fractional quantum Hall system
in the thin cylinder regime. We give a complete description of its ground state space in terms
of what we call Fragmented Matrix Product States, which are labeled by a certain family of tilings of the
one-dimensional lattice. We then prove that the model has a spectral gap above the ground states for a 
range of coupling constants that includes physical values.  As a consequence of the gap we establish the 
incompressibility of the fractional quantum Hall states. We also show that all the ground states labeled by
a tiling have a finite correlation length, for which we give an upper bound. We demonstrate by example, however,
that not all superpositions of tiling states have exponential decay of correlations.  \\[2ex]
%
\bigskip

\tableofcontents
\bigskip
\section{Introduction}\label{sec:intro}
	
The fractional quantum Hall effect (FQHE) is a result of the collective behavior of interacting charge degrees of freedom in a two-dimensional geometry with 
perpendicular magnetic field \cite{PhysRevLett.48.1559}. Two hallmarks characterize the remarkable properties of this quantum state of matter: the 
incompressibility of the liquid into which the charge carriers condense and the existence of an energy gap to excitations with fractional charge. 

Theoretical models start from Laughlin's famous ansatz for the many-body correlated ground-state  wave function~\cite{PhysRevLett.50.1395}. 
An effective description of the observed features of excitations above the ground-state is based on Haldane pseudo-potentials, i.e.~short-range 
repulsive interactions projected onto the lowest Landau level \cite{PhysRevLett.51.605}. The same approach can also be used to study partially-filled 
higher Landau levels. These model Hamiltonians are tailored to a maximal filling fraction of the ground-state  and are expected to provide 
a faithful effective description of the gap in the excitation spectrum, the incompressibility of the state, as well as the hierarchical nature of related 
filling factors~\cite{,PhysRevLett.54.237,Pokrovsky_1985,Trugman:1985lv}  (see also~\cite{QHEOxford2003} and references therein). 
Despite confirmation of these properties  in many numerical works~\cite{Duncan1990}, which naturally deal with finite systems, it is still a major 
open problem to provide a proof for the existence of such a gap in the thermodynamic limit. The importance of this issue is carefully argued in a 
recent overview of mathematical results and challenges by Rougerie~\cite{Roug19}.

One of the beautiful mathematical aspects of Haldane's approach is the fact that the features of the Hamiltonians are expected to be robust with 
respect to the particular choice of two-dimensional geometry~\cite{Haldane:2018pi}. One may therefore start by studying the pseudo-potential 
corresponding to $ 1/3 $-filling of the Landau levels in a cylinder geometry. In this case, an orthonormal basis $ (\psi_n) $ of one-particle eigenstates 
of the lowest Landau level is indexed by integers $ n \in \mathbb{Z} $, and determined by the magnetic length $ \ell = \sqrt{\hbar/(eB)} $ and the 
dimensionless parameter $ \alpha := \ell/R $ where $R$ is the
radius of the cylinder; namely, 
\begin{equation}
\psi_n(x,y) =\sqrt{\frac{\alpha}{4\pi^{3/2}  \ell^2}} \exp\left( i  n \frac{\alpha y }{\ell}\right)  \exp\left( - \frac{1}{2} \left[ \frac{x}{\ell}  -  n  \alpha \right]^2  \right)
\end{equation}
where $ x \in \mathbb{R} $ and $ y \in [0,2\pi R)  $ corresponds to the angular direction, cf.~Figure~\ref{fig:cylinder}. 
\begin{figure}[ht]
	\begin{center}
		\scalebox{1}
		{\begin{tikzpicture}
			\newcommand \w {.5};
			\newcommand \h {.75};
			\newcommand \n {.5};
			\newcommand \m {10};
			
			\draw (\n,0) ellipse ({\w} and {\h});
			\draw (\n,-\h) -- (\m,-\h);
			\draw (\n,\h) -- (\m,\h);
			\draw (\m,-\h) arc (270:450:{\w} and {\h});
			\draw[dashed] (\m,-\h) arc (270:90:{\w} and {\h});
			\draw[thick, ->] (\n,0)--(\n,\h) node[midway, right, xshift=-1mm] {$R$};
			\draw[thick, ->] (5.3,.5)--(5.5,1.5) node[above, right] {$B$};
			\draw(5.33,.65)--(5.48,.65)--(5.45,.5);
			\draw[thick, ->] (4.5,-.25) -- (6,-.25) node[below, yshift=-1mm] {$x$};
			\draw[thick, ->] (7,0) arc (0:45:0.5 and .75) node[right, xshift=1mm] {$y$}; 
			\draw[thick] (7,0) arc (0:-45:0.5 and .75);

			
			\draw[yshift=-2.5cm, domain=-.25:1.75, smooth, variable=\x] plot ({\x},{1/(.5*sqrt(2*pi))*exp(-8*(\x-1)^2)});
			
			\foreach \y in {2,3,4,9,10}
			{
				\draw[yshift=-2.5cm, domain=\y-.75:\y.75, dashed, smooth, variable=\x] plot ({\x},{1/(.5*sqrt(2*pi))*exp(-8*(\x-\y)^2)});
			}
			\draw[yshift=-2.5cm, ->] (0,-.5) -- (0,1.5) node[left, yshift=-2.5mm] {$|\psi_n(x,y)|^2$};
			\draw[yshift=-2.5cm, ->](-.5,0) -- (11.5,0) node[below, yshift=-1mm] {$x$}; 
			\foreach \x in {1,2,3,4,9,10,11}
			{
				\draw[yshift=-2.5cm] (\x,.25)--(\x,-.25);
			}
			\draw[yshift=-2.6cm, <->] (3,-.25)--(4,-.25) node[midway, below, yshift=-1mm] {$\alpha\ell$};

			\end{tikzpicture}
		}
\caption{Landau orbitals in a cylinder geometry. The magnetic field is perpendicular to the cylinder.  The one-particle eigenstates have a Gaussian shape and are lined up along the cylinder at a spacing given by the magnetic length $ \ell $. If an additional magnetic flux $ 2\pi \beta $ is applied along the axis of the cylinder, the centers of the orbitals would be shifted by an amount $\alpha\beta\ell$ (not shown). }\label{fig:cylinder}
	\end{center}
\end{figure}
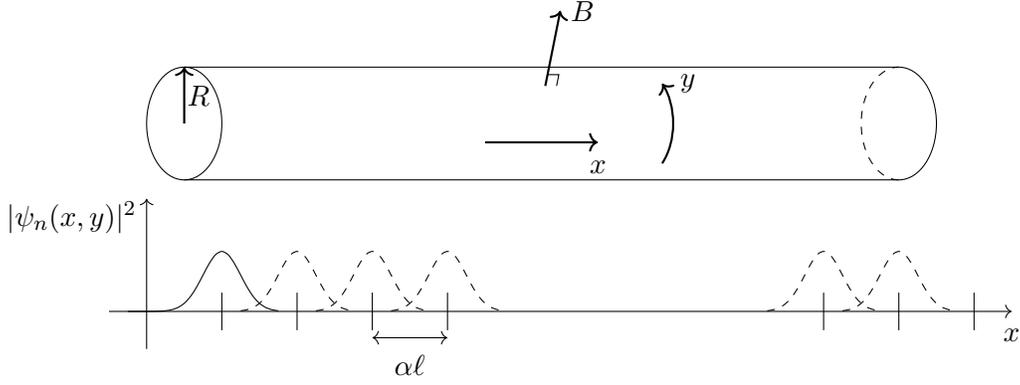
Denoting the corresponding fermionic creation and annihilation operators by $ c_n^* $ and $ c_n $, Haldane's Hamiltonian takes the form
\begin{equation}\label{eq:Haldane}
\textrm{constant}\times  \sum_{s\in \mathbb{Z}/2  } B_s^* B_s , \qquad \mbox{with} \quad B_s = \sum_{k}{\vphantom{\sum}}' F\left(2k \alpha \right) \, c_{s-k} c_{s+k}  \, ,
\end{equation} 
where the primed sum is over integers $ k \in \mathbb{Z} $ when $ s $ is an integer, and over half-integers $ k \in \mathbb{Z} + \tfrac{1}{2}$ in the case $ s $ is a half-integer. 
In order to model a filling fraction $ \nu = 1/3 $, one picks $ F(t) = t e^{-t^2/4} $.  Other filling fractions correspond to different choices of $ F $, cf.~\cite{,PhysRevLett.54.237,Pokrovsky_1985,Trugman:1985lv}.  

When $ \alpha = \ell/R$ is large, i.e.~in  the case that the cylinder radius $ R $ is small relative to the magnetic length $\ell$, it is reasonable to take into account 
only the first few terms in the sum $ B_s $. In particular, restricting to the case $ | 2 k | \leq 3 $, we arrive at a Hamiltonian for 1D lattice fermions of the form:
\begin{equation}\label{eq:Ham}
H =  \sum_{x\in\mathbb{Z}}  \left( n_x n_{x+2} + \kappa \ q_x^* q_x  \right) , 
\end{equation}
where 
\begin{equation}
n_x := c_x^* c_x    \, ,  \quad q_x := c_{x+1}  c_{x+2} -  \lambda \ c_{x}  c_{x+3}, \notag 
\end{equation}
$ \kappa = |F(\alpha)|^2/ |F(2\alpha)|^2 = e^{3\alpha^2/2} / 4 $ and $ \lambda = F(3\alpha) /F(\alpha) = 3 e^{-2\alpha^2} $. 
In what follows, we take model parameters $ \kappa > 0 $ and $ \lambda \in \mathbb{C}$ which are arbitrary unless otherwise stated. 
In Figures \ref{fig:numerics_openbc} and \ref{fig:gapsper}, we show the spectrum for chains of $9$ sites and with the physical choice of parameters
 $\lambda \in (0,3]$ and $ \kappa(\lambda) = (3^{3/4}/4) \lambda^{-3/4}$.

Truncating the sum in~\eqref{eq:Haldane} even further at $ |2k| \leq 2 $ corresponds to setting $ \lambda = 0 $. In this case, $ H $ coincides with the Tao-Thouless limit~\cite{Tao:1983cd}, in which only electrostatic interactions survive and the ground state is described by a classical particle configuration.

Model Hamiltonians of the above form as well as the original model~\eqref{eq:Haldane} conserve the  total particle number $ N = \sum_x n_x $ as 
well as the center of mass  $ \sum_x x\,  n_{x} $. On the periodic system with $L$ sites, we can define the following two unitary operators that each commute with the Hamiltonian,
\begin{equation}\label{symmetries}
U = e^{2\pi i L^{-1}N}, \quad V= e^{2\pi i L^{-1} \sum_{x=1}^L x n_x},
\end{equation}
as well as the translation operator $T$ for which $T^* n_x T = n_{(x-1)\mod L}$. These operators satisfy the relations
$$
VT = UTV, \quad UT = TU.
$$
Since $V$ commutes with $H$, there exists a $ \psi $ that is both a ground state of $H$ and eigenstate of $V$, i.e. $V\psi = \lambda \psi$. Moreover, the above relations imply that $\{T\psi, \, T^2\psi\}$ are also ground states of $H$ and eigenvectors of $V$. In particular, 
\[VT\psi = UTV\psi = \lambda T U \psi =  \lambda e^{2\pi i (N/L) } T\psi,\] 
and a similar calculation shows that $VT^2\psi=\lambda e^{4\pi i (N/L) }T^2\psi$. If $\psi$ is a ground state with $(N/L) = \nu = 1/3$ filling, the factors $e^{2\pi i \nu}$ and $e^{4\pi i \nu}$ are not equal to one, and hence there are at least 3 distinct three-periodic ground states. There are many more ground states with particle number less than $L/3$, which will be evident when we provide a description of the full ground state space in Section~\ref{sec:VMD}.

In addition to having been proposed as a ``solvable" model for the FQHE~\cite{Bergholtz:2005pl,Seidel:2005ld}, center-of-mass preserving operators such as~\eqref{eq:Ham}, whose interaction terms are restricted to 4 (or more) lattice sites,  have become an object of independent interest. In particular, 
this type of Hamiltonian appears in the context of spin liquids~\cite{Lee:2004bs} and more recently in models of scarring in many-body quantum systems, see e.g.~\cite{Pollmann,Bernevig10.19,ChoiTurneretal19,Pollmann19b}.

\subsection{Main results}
The Hamiltonian~\eqref{eq:Ham} is frustration free and its ground state at zero energy is highly degenerate. Special states in the ground state space, which arise from squeezing the Tao-Thouless state  $| \Psi_\textrm{TT} \rangle = | 100 \, 100 \, 100 \, \dots \rangle $ of strictly $1/3$-occupancy, such as
\begin{equation}\label{eq:TTstate}
\prod_{k\in \mathbb{N}_0 } \left( 1 + \lambda c_{3k+2}^* c_{3k+3}^* \, c_{3k+4}  c_{3k+1}  \right)  \big| \Psi_\textrm{TT} \rangle ,
\end{equation}
have been previously identified. Since this state can be mapped to a matrix-product state (MPS) of a spin-$ 1 $ system, its correlations can be analyzed fairly explicitly, which was accomplished in~\cite{Jansen:2009gv, Jansen:2012da,Nakamura:2012bu}. 
Aside from obvious translates, the squeezed states are by far not the only zero-energy eigenstates of $ H $. As we will discuss in Section~\ref{sec:VMD} below, there is a zoo of other states corresponding to lower fillings. One of the results of this paper is a complete description of this zero-energy subspace, which turns out to be exponentially large in the system size. 
Drawing inspiration from~\cite{Jansen:2009gv,Jansen:2012da}, we show that this space can be identified with void-monomer tilings on the line, which give rise to fragmented versions of~\eqref{eq:TTstate}. Properties of the associated novel class of fragmented matrix product states, introduced in this paper, are worked out in~Section~\ref{sec:VMD}. 

The complete classification of the ground-state space is vital for the main goal in this paper: a lower bound on the spectral gap of any finite-volume version of~\eqref{eq:Ham}, which is uniform in both volume and filling fraction. With open boundary conditions (OBC) on 
an interval $ \Lambda = [a,b] $, the operator reads
\begin{equation}\label{eq:fvHam}
H_\Lambda = \sum_{x=a}^{b-2}  n_x n_{x+2} +  \kappa  \sum_{x=a}^{b-3}  q_x^* q_x \, .
\end{equation}
It acts on the fermionic Fock space $ \mathcal{F}_\Lambda $ over the one-particle space  $ \spa \{ \psi_n  | n \in  \Lambda \} $. 
The spectral gap above its zero-energy  eigenspace $ \ker H_\Lambda = \{ \psi \in   \mathcal{F}_\Lambda \, | \, H_\Lambda \psi = 0 \} $ 
is denoted by
\begin{equation}\label{eq:unifgap}
\gap H_\Lambda := \inf \left\{ \langle \psi , H_\Lambda \psi \rangle \, | \, \psi \in  \mathcal{F}_\Lambda \cap ( \ker H_\Lambda)^\perp \, \wedge \, \| \psi \| = 1\right\} 
\end{equation}
where $(\cdot)^\perp$ denotes the orthogonal complement. Our main result is the following lower bound.
\begin{theorem}[Uniform spectral gap]\label{thm:gap}
	There is a monotone increasing function $ f : [0,\infty) \to  [0,\infty) $ such that for all $ \lambda \in \mathbb{C} $ with the property $  f\left(|\lambda|^2\right)  < 1/3  $ and all $ \kappa > 0 $:
	\begin{equation}\label{eq:gaplower}
	\inf_{\Lambda : | \Lambda | \geq 9}  \gap H_\Lambda \geq \left( \min_{L\in \{8,9,10\}} \gap H_{[1,L]} \right)   \frac{\left(1-  \sqrt{3 f\left(|\lambda|^2\right)}\right)^2}{3} \, . 
	\end{equation}
\end{theorem}

The proof of this theorem is based on an adaptation of the martingale method and can be found in Section~\ref{sec:MM}. 
Its proof yields the explicit expression~\eqref{def:f}  for $ f $, which is plotted in Figure~\ref{fig:fplot}.
From the discussion  in Section~\ref{sec:MM}  and Appendix~\ref{app:f_estimates}, we conclude that any $ |\lambda | \leq 5.3 $ satisfies the condition 
$   f\left(|\lambda|^2\right)  < 1/3 $. 
In particular, this applies to the choice $ \lambda = 3 e^{-2\alpha^2}  $ for any $ \alpha > 0 $, which corresponds to the parameters in the truncation of the original Haldane pseudo-potential. 
\begin{figure}[ht]
	\begin{center}
		\includegraphics [width=.6\textwidth]{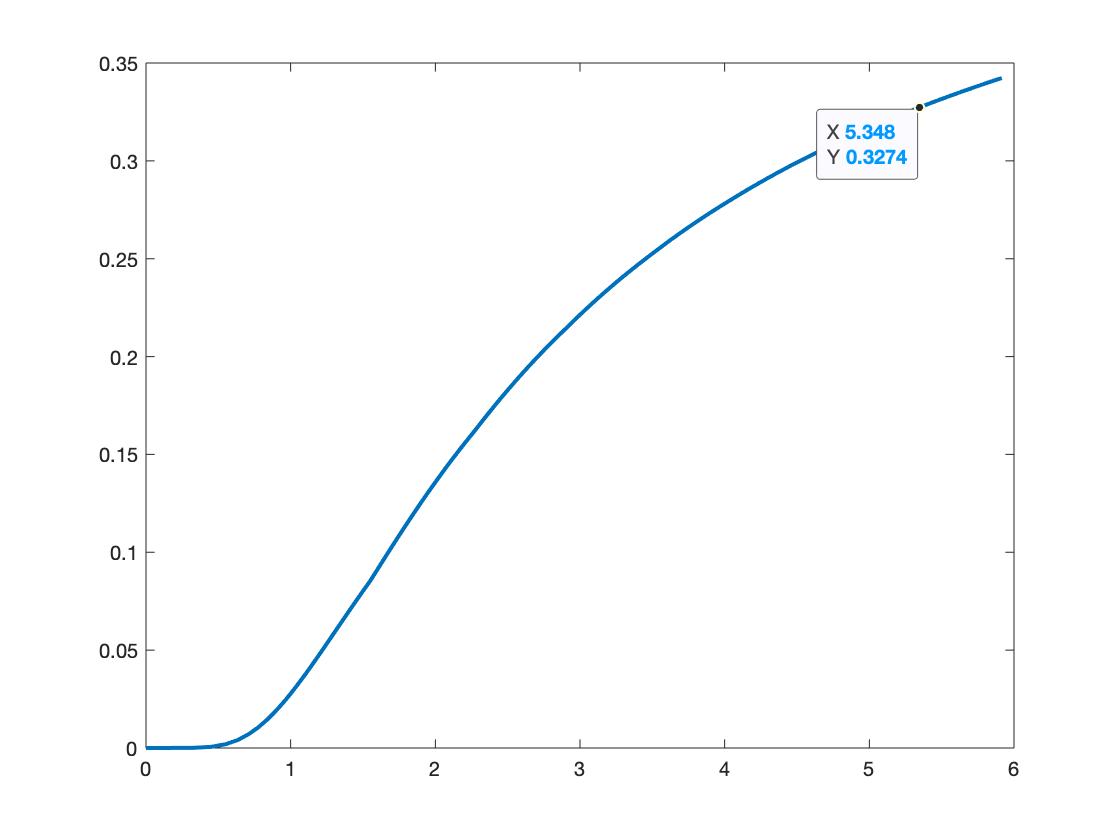}
		\caption{A plot of the function $|\lambda|\mapsto f^{(73)}(|\lambda|^2)$, defined in \eq{f70} and which approximates $f$ up to an error less than $.0052$.} 
		\label{fig:fplot}
	\end{center}
\end{figure}

The gap of an arbitrarily large system  is estimated in terms of the minimum of the gap of  $ H_{[1,L]}  $ for $ L \in \{8,9,10\}  $ lattice sites. This gap depends on $ \lambda $ and $ \kappa $ in a non-trivial way.  Figure~\ref{fig:numerics_openbc} shows a plot of the energy spectrum as a function of  $\lambda\in [0,3]$ with $\kappa(\lambda) = (3^{3/4}/4) \lambda^{-3/4}$ as in the physical case.  In the Tao-Thouless limit $ \lambda = 0 $, the interaction terms are mutually commuting and the spectral gap trivially reduces to $ \gap H_\Lambda = \min\{ 1, \kappa \} $ for any interval $ \Lambda $ larger than four sites. As is seen from the numerical data in~Figure~\ref{fig:numerics_openbc}, for open boundary conditions the gap at $ \lambda = 0 $ is unstable for small $ \lambda $. The low-energy eigenvalues are in fact caused by boundary modes which can be written down fairly explicitly. 
For example, using $\ket{\emptyset}$ to denote the vacuum state, the cyclic (or Krylov) subspace corresponding to $ H_{[a,b]} $ (with $ b \geq a+4 $) and the vector 
\begin{equation} 
| \Psi_1 \rangle \equiv c_a^* c_{a+1}^* c_{a+4}^* | \emptyset \rangle
\end{equation}
is two-dimensional and spanned by the orthogonal basis $ \Psi_1 $ and $ \Psi_2 \equiv c_a^* c_{a+2}^* c_{a+3}^* | \emptyset \rangle $. In this subspace, the action of $ H_{[a,b]} $ is given by the $ 2\times 2$-matrix
$$
\left( \begin{matrix} \kappa |\lambda|^2 & - \kappa \overline{\lambda} \\  - \kappa \lambda & 1+ \kappa (1+|\lambda|^2)   \end{matrix} \right) .
$$
Its eigenvalues are $ \kappa |\lambda|^2 + \left[ (1+\kappa) \pm \sqrt{ (1+\kappa)^2 + 4 \kappa^2 |\lambda|^2} \right] /2 $. For small $ | \lambda | $, the smaller eigenvalue is of the order $ \frac{\kappa}{1+\kappa} |\lambda|^2 + \mathcal{O}(|\lambda|^4)  $. More generally, a (left) boundary mode of $  H_{[a,b]}  $ with energy $ \mathcal{O}(|\lambda|^2) $ can be constructed starting from any vector of the form $ c_a^* c_{a+1}^* \prod_{k=1}^n c_{a+3k+1}^* | \emptyset \rangle $, and likewise for the right boundary. 

\begin{figure}[ht]
	\begin{center}
		\includegraphics [width=.32\textwidth]{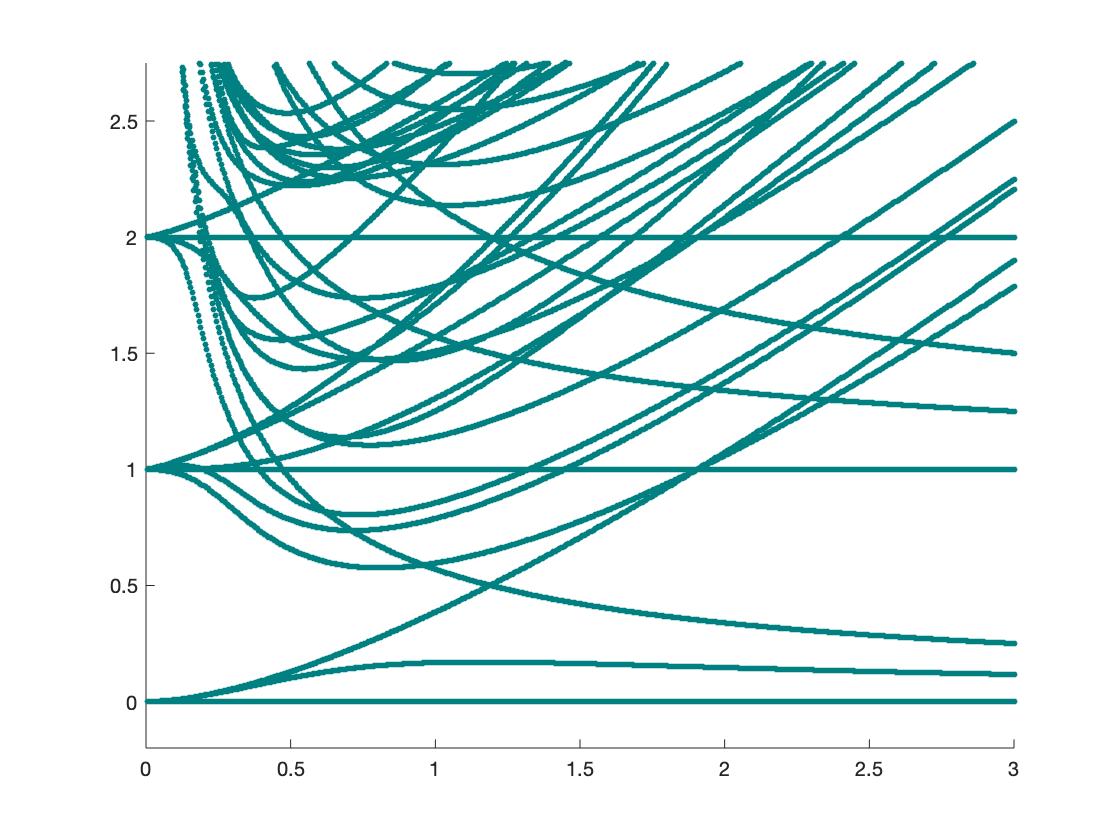}\hfill 
		\includegraphics [width=.32\textwidth]{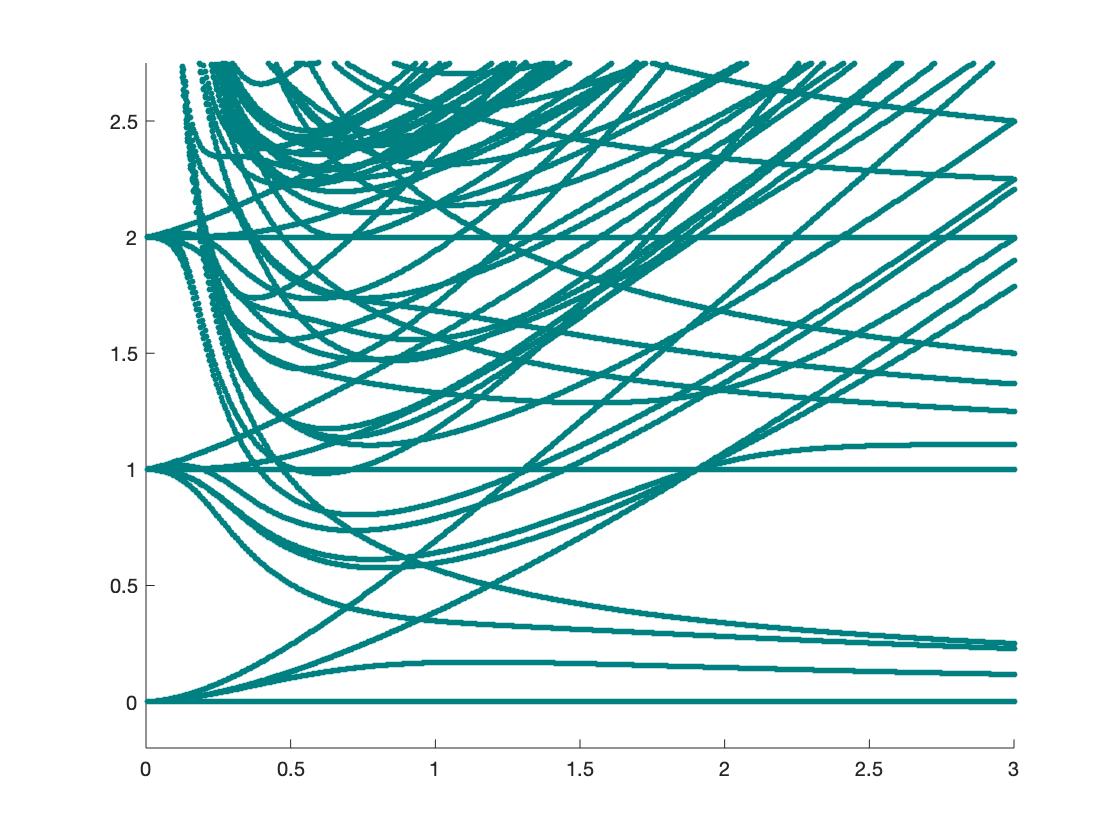}\hfill 
		\includegraphics [width=.32\textwidth]{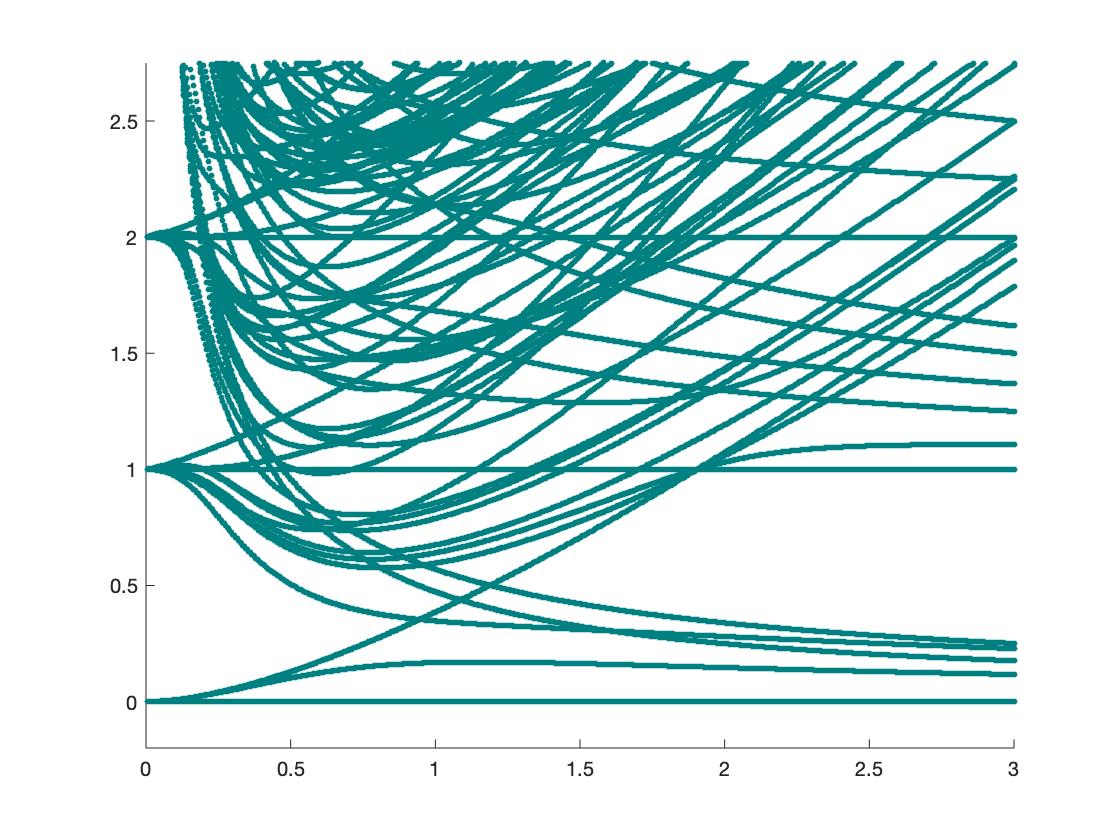}
		\caption{A plot of the spectrum of $ H_{[1,L]} $, with $L=8,9,10$, as a function of $ \lambda \in [0,3]$ for $ \kappa(\lambda) = (3^{3/4}/4) \lambda^{-3/4}$. }
		\label{fig:numerics_openbc}
	\end{center}
\end{figure}

The existence of boundary modes for open boundary conditions, which are responsible for a vanishing gap in the limit $ |\lambda |\downarrow 0 $, 
prompts the question about the existence of such modes in the bulk. For this, it is natural to look at the Hamiltonian with different boundary conditions, e.g.\ the periodic Hamiltonian
\begin{equation}\label{eq:perHam}
H^{\textrm{per}}_{[a,b]} := \sum_{x=a}^{b}  \left( n_x n_{x+2} + \kappa \ q_x^* q_x  \right) 
\end{equation}
where we identify $ b+k \equiv a + k-1$ for $ k \in \mathbb{N} $. Alternatively, and motivated by the above explicit form of the boundary modes above, one may look at (soft) Dirichlet-type boundary conditions
\begin{equation}\label{eq:DHam}
H^{\textrm{D}}_{[a,b]} := 	H_{[a,b]} + \left( n_a n_{a+1} +  n_{b-1} n_{b}\right) . 
\end{equation}
In comparison to $ H_{[a,b]} $, the kernels of $ H^{\#}_{[a,b]} $ are slightly depleted for $ \# \in \{ \textrm{per}, \textrm{D} \} $ and, in particular, the boundary states constructed above are not ground states for either system at $\lambda=0$.  Moreover, numerical data for small system sizes suggest that the gap of $ H^{\#}_{[a,b]} $ for both $ \# \in \{ \textrm{per}, \textrm{D} \} $ is non-vanishing uniformly
for all $ |\lambda | \geq 0$ in compact intervals, cf.~Figure~\ref{fig:gapsper}. This suggests that the instability of the gap for finite open systems at $|\lambda|=$0 is in fact due to boundary states, and is not a feature of the system in the thermodynamic limit $\Lambda \to \bZ$.

The instability of the gap on finite intervals with open boundary conditions in the limit $|\lambda|\to 0$, suggests that although the model with 
$\lambda=0$ is a function of the particle numbers $n_x$ and easily seen to be gapped, perturbative arguments such as given in 
\cite{datta:1996,datta:1996a,de-roeck:2019,frohlich:2018} to prove 
a gap for $\lambda\neq 0$ will not work. This is also consistent with the observation that commuting Hamiltonians such as the model under 
consideration with $\lambda =0$ cannot describe quantum Hall effects \cite{kapustin:2020,bachmann:2020}. If the model had a stable gap at 
$\lambda=0$, the anyons describing the excitations and the basis for the FQHE, would also be stable \cite{hastings:2015,haah:2016,cha:2020}.

In our approach here, we take advantage of the frustration free property of the model. Starting with the game-changing work of Affleck, Kennedy, Lieb, and Tasaki
\cite{affleck:1988}, frustration-freeness has been exploited to obtain lower bounds on the ground state gap for an increasing variety of models.
These techniques often yield mathematical proofs (e.g.~in \cite{fannes:1992,nachtergaele:1996,bravyi:2015,bishop:2016a,gosset:2016,young:2016,lemm:2018,abdul-rahman:2020}). In other instances, they are used in combination with the results of numerical simulations (e.g.\ in \cite{knabe:1988, lemm:2019a,pomata:2019a}). 
Unfortunately, in their present form all these approaches produce lower bounds that contain as a factor the gap of the Hamiltonian with open boundary conditions $H_{[a,b]}$ for some finite interval.
In the model at hand these gaps vanish as $\lambda\to 0$ and these methods give unsatisfactory results if, as we expect to be the case here, the gap for finite systems
with periodic boundary conditions (and hence the bulk gap in the thermodynamic limit) does not vanish in that limit. Aside from this deficiency, useful lower bounds 
can also be obtained for the system with periodic boundary conditions.

In Section~\ref{sec:MM} we provide the following gap bound for the periodic system using a version of Knabe's finite-size criterion \cite{knabe:1988} adapted for longer range interactions. In contrast to the typical situation, here we do not need to rely on numerical calculations to verify this finite size criterion because we already obtained sufficiently good rigorous estimates for open boundary conditions.

\begin{theorem}[Spectral gap - periodic case] \label{thm:periodic_gap}
	Let $n\geq 2$ be an integer such that 
	\[
	g_n := \inf_{0 \leq r \leq 5} \gap\left(H_{[1,3(n+1)+r]}\right) > \frac{\Gamma}{ n}
	\]
	where $ \gamma := \min_{m \in \{ 6,7 \} } \gap(H_{[1,m]}) $ and $ \Gamma := \max_{m \in \{ 6,7 \} }  \|H_{[1,m]} \|$. Then for all $N > n$ and $ r \in \{0,1,2, \ldots, 5 \} $,
	\begin{equation}\label{eq:gapper}
	\gap\left(H_{[1,6N+r]}^\per\right) \geq \frac{\gamma\cdot n}{2\Gamma (n-1)}\left[g_n-\frac{\Gamma}{n}\right].
	\end{equation}
	In particular, for any $\lambda $ such that  $  f\left(|\lambda|^2\right)  < 1/3  $,
	\begin{equation}\label{eq:gapperexp}
\liminf_{L\to \infty} \gap\left(H_{[1,L]}^\per\right) \geq \frac{\gamma}{6\Gamma }\left( \min_{L\in \{8,9,10\}} \gap H_{[1,L]} \right)\left(1-  \sqrt{3 f\left(|\lambda|^2\right)}\right)^2  . 
	\end{equation}
\end{theorem}
The proof of~\eqref{eq:gapper} is spelled out in Subsection~\ref{subsec:periodic_gap}. The explicit lower bound~\eqref{eq:gapperexp} results from inserting~\eqref{eq:gaplower} into~\eqref{eq:gapper}.

%
An important question not addressed in this paper is the persistence of the spectral gap under a class of perturbations that 
include the non-truncated pseudo-potential Hamiltonian (\ref{eq:Haldane}). The large degeneracy of the ground state makes this a subtle question 
for which the available stability theorems (see, e.g.~\cite{Michalakis:2013}) do not apply.
\vskip12pt

\begin{figure}[ht]
	\begin{center}
		\includegraphics [width=.49\textwidth]{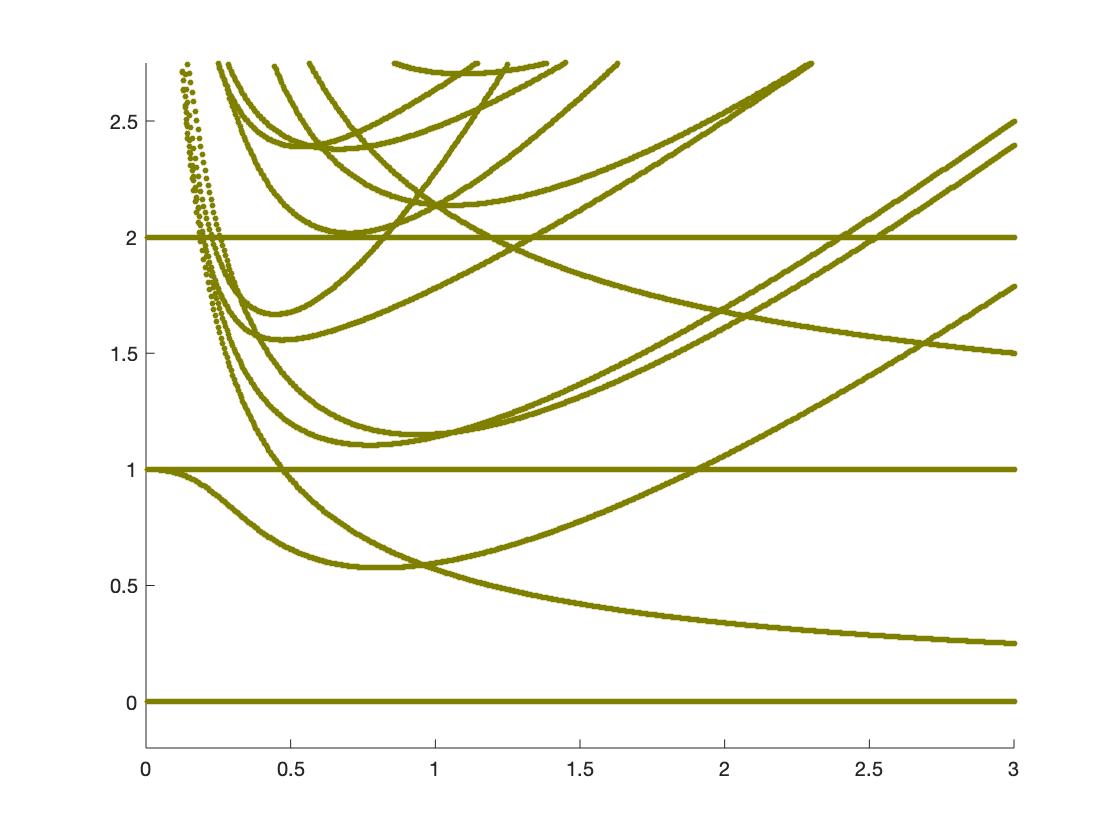}\hfill \includegraphics [width=.49\textwidth]{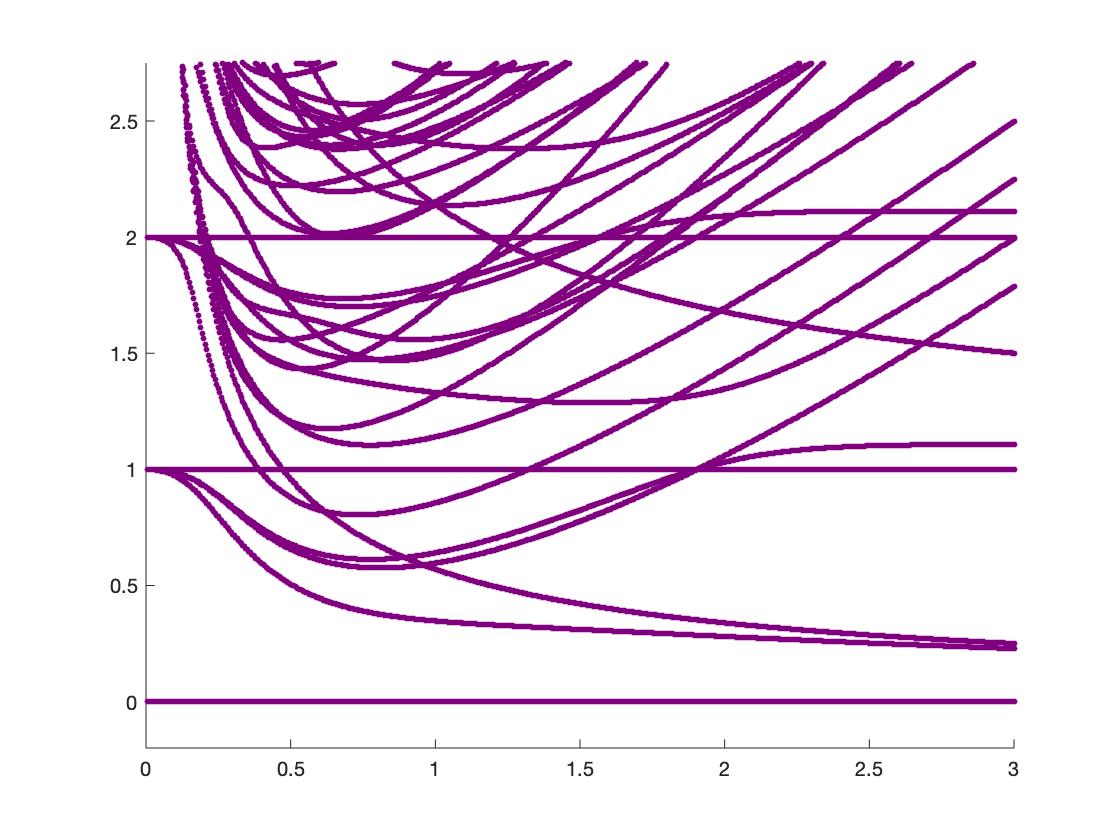}
		\caption{A plot of the spectrum of $ H_{[1,9]}^\textrm{per} $ and $H^D_{[0,9]}$, as a function of $ \lambda \in [0,3]$ for $ \kappa(\lambda) = (3^{3/4}/4) \lambda^{-3/4}$.} \label{fig:gapsper}
	\end{center}
\end{figure}

A consequence of our ground state description and spectral gap result is the incompressibility of the FQH system in the thermodynamic limit. While the ground state of $ H_\Lambda $ on the whole Fock space $ \mathcal{F}_\Lambda $ has zero energy, this may cease to be the case when restricting to a subspace with fixed particle number $ N $, i.e.
\begin{equation}\label{eq:gap2}
E_\Lambda(N) := \inf \left\{ \langle \psi , H_\Lambda \psi \rangle \, | \, \psi \in  \mathcal{F}_\Lambda \, \wedge  \| \psi \| = 1 \, \wedge \,  N_\Lambda \psi = N \psi  \right\} \, ,
\end{equation}
where $ N_\Lambda  = \sum_{x\in \Lambda} n_x $ is the number operator corresponding to $ \Lambda $ which commutes with $ H_\Lambda $.
The structure of $\ker(H_\Lambda)$ described in Section~\ref{sec:VMD} will imply the following behavior of the $N$-particle ground-state energy. 
\begin{theorem}[Maximal filling]\label{thm:gsenergy}
	For any interval $\Lambda$ of length $ |\Lambda | \geq 8 $, there is a maximal particle number $ N_\Lambda^m \in \mathbb{N} $ with
	\begin{equation}\label{eq:maxNumber}
	 \frac{1}{3} \leq \frac{ N_\Lambda^m}{|\Lambda|} \leq  \frac{1}{3} + \frac{4}{3|\Lambda|} 
	 \end{equation}
	such that the ground-state energy obeys: 
	\begin{equation}\label{eq:gsenergy}
	E_\Lambda(N)  \begin{cases}     = 0 & \mbox{if $ N  \leq N_\Lambda^m $,} \\
	> 0 &  \mbox{else.} \end{cases}
	\end{equation} 
	In the thermodynamic limit, the maximal filling fraction equals $ \lim_{|\Lambda|\to \infty }  N_\Lambda^m\big/|\Lambda|  = 1/3 $.
\end{theorem}
A proof of this theorem is found at the end of Section~\ref{sec:VMD}. 
Note that the gap proven in Theorem~\ref{thm:gap} immediately implies a uniform lower bound on the $N$-particle ground-state energy for the second case in~\eqref{eq:gsenergy}.
From this we conclude that the compressibility vanishes in the thermodynamic limit. More precisely, consider $ N $ particles on an interval $ \Lambda $ and let $ \Lambda_+ $ and $\Lambda_- $ stand for the intervals which arise from $ \Lambda $ by adding or subtracting a lattice site.  The inverse compressibility  (at zero temperature) is defined in terms of the second derivative of the ground-state energy: 
\begin{equation}
\frac{1}{\kappa_{\Lambda}(N)}  :=   |\Lambda|  \,  \frac{E_{\Lambda_{+}}(N) + E_{\Lambda_{-}}(N)- 2 E_{\Lambda}(N)}{(2\pi \ell^2 )^2} \, . 
\end{equation}
In the situation where $ 0 = E_{\Lambda_+}(N) = E_{\Lambda}(N) < E_{\Lambda_{-}}(N)  $, an immediate implication of the uniform lower bound on $ E_{\Lambda_{-}}(N)  $ from Theorem~\ref{thm:gap} is as follows.

\begin{cor}[Incompressibility]\label{cor:incompress}
	At zero temperature and critical filling factor, i.e.~$ N =  N_\Lambda^m $, the compressibility $ \kappa_{\Lambda}(N) $ vanishes in the thermodynamic limit $  |\Lambda| \to \infty $.
\end{cor}

On the level of related classical 2D Coulomb problems, incompressibility in the sense of an upper bound on the charge density corresponding to perturbations of Laughlin's wave function was established in~\cite{Lieb:2019vl}. In contrast, our Corollary~\ref{cor:incompress} starts from a microscopic (i.e.~Hamiltonian) description of the system. \\

Spectral gaps are usually associated with exponential clustering in the thermodynamic limit of the ground state~\cite{hastings:2006,nachtergaele:2006a}. The caveat here is that the ground state space $\ker(H_\Lambda)$ grows exponentially in the system size $|\Lambda|$, for which the previous results do not apply. In Section~\ref{sec:correlation_decay}, we show that while the fragmented matrix product states exhibit exponential clustering, there are other pure ground states with arbitrarily slow decay. 

We denote by $\omega_R^\Lambda:\cA_\Lambda\to \bC$ the state associated with the fragmented matrix product state
\begin{equation}\label{eq:tilstate}
\psi_\Lambda(R) = \sum_{\pmb{D}  \in  \mathcal{D}_\Lambda(R) } C_{D_1}^* \dots C_{D_N}^* |\emptyset \rangle  .
\end{equation}
Here, $R\in\cR_\Lambda$ is a tiling of the interval $\Lambda$ by voids and monomers, and $ \pmb{D} = (D_1,\dots, D_N) \in  \mathcal{D}_\Lambda(R) $ is any void-monomer-dimer (VMD) tiling of $ \Lambda $ obtained from $R$ by replacing neighboring monomers by a dimer. The operator associated with an individual domino $ D $ in the tiling is 
\begin{equation} 
C_D^* = \begin{cases} 1 &  \mbox{if $ D $ is a void,} \\  c_x^* & \mbox{if $ D $ is a monomer starting at $ x $,} \\
\lambda \, c_{x+1}^* c_{x+2}^* & \mbox{if $ D $ is a dimer starting at $ x $.}  \end{cases} 
\end{equation} 
For intervals $|\Lambda|\geq 8$ the set $\{\psi_\Lambda(R) \, \big| \, R\in\cR_\Lambda\}$ forms an orthogonal basis for the ground states space. For more details on these states and precise definitions, see Section~\ref{sec:VMD} and Section~\ref{susec:Ref} below.

For the exponential clustering result, recall that the algebra of observables $ \mathcal{A}_\Lambda $ on the fermionic Fock space  $ \mathcal{F}_\Lambda $  decomposes into even and odd parts under parity, denoted by $\mathcal{A}_\Lambda^e $ and $\mathcal{A}_\Lambda^o $, respectively. The even observables are generated by even monomials of the creation and annihilation operators associated with $ \Lambda $. An analogous construction holds for the odd observables. We prove exponential decay for the even and odd observables. To describe the correlation length for $ \lambda \in \mathbb{C} \backslash \{ 0\} $, set
\begin{equation}\label{decay_rate}
c(\lambda ) := \frac{1}{3} \ln \frac{ \sqrt{4|\lambda|^2 +1}+1}{\sqrt{4|\lambda|^2 +1}-1} . 
\end{equation}
While $c(\lambda)$ is undefined for $\lambda=0$, it is easy to check that $\lim_{\lambda \to 0} \exp(-c(\lambda) d ) \to 0$ for any $ d >0 $. 

\begin{theorem}[Exponential clustering] \label{thm:expcluster}
	Consider an interval $ \Lambda \subset \mathbb{Z} $ and subintervals $ X , Y \subset \Lambda $ of distance $ d(X,Y) \geq 20 $ as well as a state $ \omega_{R}^\Lambda $ characterized by $ R \in \mathcal{R}_\Lambda  $. For any $ \lambda \in \mathbb{C} \backslash \{ 0\} $, $ \kappa > 0 $ and any pair of even observables $ A_1 \in \mathcal{A}_X^e $ and $ A_2 \in \mathcal{A}_Y^e $ supported on $ X $ respectively $ Y $,  
	\begin{equation}\label{eq:ABexpdecay}
	\left| \omega_{R}^\Lambda(A_1A_2)- \omega_{R}^\Lambda(A_1)\, \omega_{R}^\Lambda(A_2) \right| \leq 8  \|A_1\|\|A_2\|e^{-c(\lambda)( \dist(X,Y) -20) /2} . 
	\end{equation}
	For any odd observable $ A_1 \in \mathcal{A}_X^o $ and any other  (even or odd) observable $ A_2 \in \mathcal{A}_Y $, we have 
	\begin{equation}\label{eq:zeroodd}
	\omega_{R}^\Lambda(A_1A_2) =  \omega_{R}^\Lambda(A_1) = 0 . 
	\end{equation}
\end{theorem}
The exponential decay of the correlations of $\omega_{R}^\Lambda$ is uniform with respect to the root tiling $R$ and interval $\Lambda$. 
As a consequence, any infinite-volume ground state obtained as a weak-$*$ limit of such finite-volume states will also have exponential clustering.

When $\lambda=0$, every state~\eqref{eq:tilstate}  is a product state, which trivially satisfies
$
\omega_{R}^\Lambda(A_1A_2)- \omega_{R}^\Lambda(A_1)\, \omega_{R}^\Lambda(A_2)=0
$
for all observables $A_1$ and $A_2$ as in~Theorem~\ref{thm:expcluster}.  The result~\eqref{eq:ABexpdecay}, which is initially only formulated for $ \lambda \neq 0 $, thus extends continuously to $\lambda=0$.

As was mentioned above, due to the fact that the VMD states form a subspace $ \mathcal{G}_\Lambda $ which increases exponentially in~$|\Lambda | $, one cannot expect~\eqref{eq:ABexpdecay} to carry over to arbitrary linear combinations of tiling states. In fact, we show in Subsection~\ref{subsec:noexpdecay} that there are  pure ground states of the infinite chain, for which the correlations do not decay exponentially.

In case the interval has length $ |\Lambda| \in 3 \mathbb{N} $, the a fragmented matrix product state corresponding to the choice $ R = R_M $ of a pure monomer tiling is the squeezed Tao-Thouless state~\eqref{eq:TTstate}. Exponential decay was already established for this state in~\cite{Jansen:2009gv,Jansen:2012da} through a slightly different analysis than the one employed here. In~\cite{Nakamura:2012bu} the explicit dependence on~\eqref{decay_rate} 
was worked out for the density correlations in the infinite-volume limit of the squeezed state, i.e. $ \omega_{R_M} := \lim_{|\Lambda| \to \infty} \omega_{R_M}^\Lambda $. Specifically, it was shown that
\begin{equation}
\lim_{|x-y|\to \infty} \,\frac{\ln \left|  \omega_{R_{M}}( n_{x} n_{y})  - \omega_{R_M}( n_{x} ) \omega_{R_M}(n_{y})  \right| }{  |x-y|}= - c(\lambda) .
\end{equation}
We believe that this decay rate should be a bound for the other states $ \omega_{R}^{\Lambda} $, and the factor of two in \eqref{eq:ABexpdecay} is 
a result of our proof method. In general, \eqref{eq:ABexpdecay} is not sharp since the class of fragmented matrix product states includes product states for any
$ \lambda\in\bC$. 

The limiting case $ \lambda = \infty $ corresponds to a situation in which the kernel of $ H_\Lambda $ at maximal filling is five-periodic and not three-periodic. Therefore, one might wonder whether the phases at $ \lambda = 0 $ and $ \lambda = \infty $  are separated by a closing of the gap. It is quite possible, however, that the change of the 
period only occurs in the limit $ \lambda = \infty$, where the correlation length diverges, $ \lim_{\lambda \to \infty} c(\lambda) = 0 $. 


As was proven in \cite{Koma2004}, the non-uniqueness of the ground state and the existence of a gapped excitations is generally tied to rational fillings and the breaking of the translation symmetry. For the squeezed state~\eqref{eq:TTstate}, which corresponds to $ \omega_{R_M}$, translation symmetry breaking was proved in~\cite{Jansen:2009gv}. It also follows from the argument given below (\ref{symmetries}).
For any $ |\lambda| \neq \sqrt{2} $ the one-particle density of the squeezed Tao-Thouless state explicitly shows this translation symmetry breaking as can be seen from the computation  in  \cite{Nakamura:2012bu}, i.e.~for any $ k \in \mathbb{N}$:
\begin{equation}
\omega_{R_M}( n_{3k+1} ) = \frac{1}{\sqrt{4|\lambda|^2 +1}} , \qquad 
\omega_{R_M}( n_{3k + 1 \pm 1} ) =  \frac{1}{2}  \left( 1 -  \frac{1}{\sqrt{4|\lambda|^2 +1}} \right) .
\end{equation}
The squeezed Tao-Thouless state $  \omega_{R_M} $ also exhibits string-order of two kinds \cite{Nakamura:2012bu}. The string order parameters are calculated using the observables $O^z_{3k,3\ell} $ and $\overline{O}^z_{3k,3\ell} $ defined by
$$
O^z_{3k,3\ell} 
=  -(n_{3k+2} - n_{3k}) e^{i\pi \sum_{j=k+1}^{\ell -1} (n_{3j+2} - n_{3j})}(n_{3\ell+2} - n_{3\ell}),
\quad \overline{O}^z _{3k,3\ell} = e^{i\pi \sum_{j=k+1}^{\ell -1} (n_{3j+2} - n_{3j})}.
$$
In the squeezed Tao-Thouless state the expectations converge as $\ell-k\to \infty$ to non-zero values:
$$
\lim_{\ell- k\to \infty} \omega_{R_M}(O^z_{3k,3\ell} ) = \frac{(\sqrt{4|\lambda|^2 +1} -1)^2}{4|\lambda|^2 +1},
\quad \lim_{\ell- k\to \infty} \omega_{R_M}(\overline{O}^z_{3k,3\ell} ) =\frac{1}{4|\lambda|^2 +1}.
$$
It is worth noting that the string observables stemming from the Jordan-Wigner transformation of a fermionic correlation function of the 
type $c^*_k c_\ell$ do {\em not} show string order. Such correlations decay exponentially as shown in Section \ref{subsec:Proofexp}.



\subsection{Reformulation}\label{susec:Ref}
While the operators $ H_\Lambda^\#  $ with $ \# \in \{ \cdot , \textrm{per}, \textrm{D} \} $ act on the fermionic Fock space $ \mathcal{F}_\Lambda $, for the analysis in this work we find it convenient to rewrite the fermionic system as a spin-$ \tfrac{1}{2} $ chain via the Jordan-Wigner transformation on $ \Lambda = [a,b]  $. For the spin-$\tfrac{1}{2}$ chain, the canonical orthonormal basis for $ \mathbb{C}^2 \equiv \spa\{ | 1 \rangle , | 0 \rangle \} $ represents if the site $x\in\Lambda$ is occupied or vacant. The algebra of observables, $\cA_\Lambda$, is the set of all bounded operators acting on the total Hilbert space, $\cH_\Lambda$, defined by the tensor product of the on-site spaces, i.e.
\begin{equation}\label{eq:algebra}
\mathcal{H}_\Lambda := \bigotimes_{x\in \Lambda } \mathbb{C}^2, \quad
\cA_{\Lambda} = \caB(\cH_\Lambda) .
\end{equation}
For any pair of finite volumes $\Lambda_1\subset \Lambda_2$ we can identify an observable $A\in\cA_{\Lambda_1}$ as acting on $\cH_{\Lambda_2}$ via the identification $A\mapsto A\otimes \1_{\Lambda_2\setminus \Lambda_1}\in \cA_{\Lambda_2}$. For simplicity, we will typically suppress the identity in our notation.

Given the three Pauli matrices $ \sigma_x^1 , \sigma_x^2 , \sigma_x^3 $ and the corresponding lowering and raising operators 
$$ \sigma_x^+ = \tfrac{1}{2} ( \sigma_x^1 + i \sigma_x^2 )  \equiv \left( \begin{matrix} 0 & 1 \\ 0 & 0 \end{matrix} \right)   ,  \quad \sigma_x^- = \tfrac{1}{2} ( \sigma_x^1 - i \sigma_x^2 )  \equiv \left( \begin{matrix} 0 & 0 \\ 1 & 0 \end{matrix} \right)   $$ 
the canonical anticommutation relations are implemented on $\cH_\Lambda$ by the operators
\begin{align}\label{eq:JW}
c_x =  \left( \prod_{k=a}^{x-1} \sigma_k^3 \, \right) 	\sigma_x^-  , \qquad 
c_x^*  =\left( \prod_{k=a}^{x-1} \sigma_k^3 \, \right) 	 	\sigma_x^+.
\end{align} 
for all $x\in\Lambda$,~\cite{Jordan:1928sd}. Under the Jordan-Wigner transformation, the Hamiltonian~\eqref{eq:fvHam} with open boundary conditions is unitarily equivalent to the following 
spin-$\tfrac{1}{2} $ system:
\begin{align}\label{def:Hspin}
H_\Lambda & =  \sum_{x=a}^{b-2}  n_x n_{x+2} +  \kappa  \sum_{x=a}^{b-3}  q_x^* q_x  \notag \\
\mbox{with} \quad  n_x = \tfrac{1}{2} &(\sigma_x^3 +1 ) , \qquad q_x = \sigma_{x+1}^- \sigma_{x+2}^- - \lambda \  \sigma_{x}^- \sigma_{x+3}^- .  
\end{align}
The expression for $q_x$ above is obtained by applying the Jordan-Wigner transformation to $q_x^*q_x$ from \eqref{eq:Ham} and factoring the resulting quantity. The case of periodic or (soft) Dirichlet boundary conditions~\eqref{eq:perHam} and~\eqref{eq:DHam} can be rewritten similarly.
More generally, even fermionic observables $ \mathcal{A}^e_X $ transform under Jordan-Wigner to observables in $  \caB(\cH_X) $. This is not the case for odd fermionic observables $ \mathcal{A}^o_X $, which are multiplied by a string of  $ \sigma^3 $-operators extending to the left.\\

Section~\ref{sec:VMD} contains the definition and properties of the VMD states in the spin language, see~Definition~\ref{def:VMDstate}.  For the reformulation of the results in from this introduction, which are stated in the fermionic language, it is 
useful to recall that under the Jordan-Wigner transformation, the vacuum state on the fermionic Fock space is identified with the state  of all spins down,  i.e.
$$
\mathcal{F}_\Lambda \ni  |\emptyset \rangle  \; \equiv \; | \pmb{0} \rangle \in  \mathcal{H}_\Lambda  \, . 
$$
More generally, the occupation basis vectors $ c_{x_1}^* c_{x_2}^* \cdots c_{x_N}^* |\emptyset \rangle  $ indicating a fixed particle configuration at $ x_1 < x_2 < \dots < x_N $ are identified -- up to a sign -- with the spin basis vectors $ | \pmb{\sigma} \rangle $ in which the spins are up at the particle locations, $ \sigma_{x_j} = 1 $ for all $ j \in \{ 1, \dots , N \} $, and all other spins are down,  $ \sigma_{y} = 0 $ for any $ y\neq x_j $.  For the corresponding states, this sign is of course irrelevant. However,  the relative phases matter when considering linear combinations as in the VMD states.
To prove that the fragmented matrix product states defined in~\eqref{eq:tilstate} and \eqref{eq:VMDstate} below indeed result from each other through the Jordan Wigner transformation, it is useful to express them in the occupation basis in Fock space (respectively, spin space). In either representation, for any two particle configurations that differ only by the replacement of two adjacent monomers with a dimer, the phase \emph{and} weight change is given by~\eqref{eq:weightchar} below. As a consequence, the fermionic and spin VMD state defined by the same root tiling $R\in\cR_\Lambda$ are related by the Jordan-Wigner transformation up to a sign, i.e 
\begin{equation}\label{eq:JW_VMD}
\psi_\Lambda^{\mathrm{s}}(R) = \pm U_{\mathrm{JW}}\psi_\Lambda^{\mathrm{f}}(R).
\end{equation}
For the definition of the spin VMD states, see Section~\ref{subsec:VMD_definition}.

\section{Fragmented MPS}\label{sec:VMD}

We introduce a novel class of fragmented matrix product states which are composed of concatenated products of matrix product states (MPS) of arbitrary length. In this paper, we focus on a special subclass, namely those whose fragments are the squeezed Tao-Thouless state~\eqref{eq:TTstate}.
The collection of these states turn out to form an orthogonal basis for the ground state of the Hamiltonian introduced in \eqref{eq:fvHam}. The states are formulated in spin language and associated with tilings of voids (V), monomers (M) and dimers (D) on a finite interval $ \Lambda \subset \mathbb{Z} $. The dimension of the space spanned by these states will turn out to be exponential in $ |\Lambda | $. Therefore, as a whole, this subspace is not amenable to a MPS representation with small matrix size. However, much of the MPS technology can be transplanted to our fragmented MPS.

Given a finite interval $\Lambda$, our fragmented MPS belong to the Hilbert space 
$
\mathcal{H}_\Lambda = \bigotimes_{x\in \Lambda } \mathbb{C}^2 $.
The canonical orthonormal basis of $ \mathbb{C}^2 $ is denoted by $ | 1 \rangle $ (= spin up),  $ | 0 \rangle $ (= spin down). Accordingly, configurations $ \pmb{\sigma} \in \{0,1\}^\Lambda $ label 
the standard orthonormal basis vectors, $ | \pmb{\sigma}  \rangle  $,  of the tensor product $ \mathcal{H}_\Lambda $. Connecting to the motivation presented in the introduction, we interpret spin-up states as being occupied by particles and spin-down states as being vacant.

\subsection{Definition of VMD tilings}\label{subsec:tilings}

We consider states that are supported on configurations described by domino tilings of the lattice. For tilings of $\bZ$ there are three kinds of dominoes, which are characterized by their length and particle content, see~Figure~\ref{fig:dominos}:
\begin{enumerate}
	\itemsep0pt
	\item A \emph{void} covers one lattice site. Its particle content is empty.
	\item A \emph{monomer} covers three lattice sites. It has a particle at its first site.
	\item A \emph{dimer} covers six lattice sites. It has particles at its second and third sites.
\end{enumerate}
\begin{figure}[ht]
	\begin{center}
		\includegraphics [width=.5\textwidth]{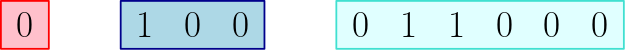}
		\caption{The three basic dominoes: voids, monomers and dimers. The location of particles are denoted by $1$'s and vacancies are marked by $0$'s. }\label{fig:dominos}
	\end{center}
\end{figure}

A \emph{VMD tiling} of $ \mathbb{Z}$ is any tiling of the lattice with voids, monomers and dimers. A \emph{root tiling} of $\bZ$ is any tiling of the lattice with just voids and monomers. Using the replacement rule that any two neighboring monomers can be substituted for a dimer, the entire set of VMD tilings of $ \mathbb{Z}$ is obtained from the set of root tilings, $\cR_\bZ$. Using the replacement rule in reverse, we see that every VMD tiling is associated with a unique root tiling. 


Cutting  a \emph{VMD tiling} of $ \mathbb{Z} $ at the two ends of some finite interval $ \Lambda \subset \mathbb{Z} $ gives rise to a set of \emph{boundary dominoes}~$B=(B^l,B^r)$. The particle content of some of these boundary dominoes can be obtained by a tiling of voids and monomers, and thus do not require the introduction of new dominoes. For intervals $|\Lambda|\geq 5$, the cases that cannot be obtained by voids and monomers gives rise to six boundary tiles.\footnote{For a complete description of the segments of $ \mathbb{Z} $-tilings on smaller intervals, i.e. $|\Lambda|\leq 4$, one would need different boundary dominoes. E.g.~for $ |\Lambda| = 4 $ one would need to remove the last site of the left-boundary domino $ B_d^l $.} Four of these, which we introduce now, are used to define root tilings of $\Lambda$:
\begin{enumerate}
	\itemsep0pt
	\item At the left boundary $B^l$, we introduce a \emph{left dimer} $ B^l_d $ which has length five. It carries two particles at its first and second site, cf.~Figure~\ref{fig:dominos2}.
	\begin{figure}[ht]
		\begin{center}
			\includegraphics [width=.2\textwidth]{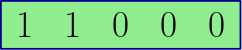}
			\caption{The dimer domino $ B^l_d $ which may be placed at the left boundary.}\label{fig:dominos2}
		\end{center}
	\end{figure}
	\item At the right boundary $ B^r $ we introduce, cf.~Figure~\ref{fig:rightdominos}:
	\begin{enumerate}
		\itemsep0pt
		\item A \emph{right dimer} $ B^r_d $, which has length three and carries two particles at its second and third site.
		\item A \emph{right $ 1 $-monomer} $ B^r_{1m} $, which has length one and carries a particle.
		\item A \emph{right $ 2 $-monomer}  $ B^r_{2m} $, which has length two and carries a particle at its first site. 
	\end{enumerate} 
	\begin{figure}[ht]
		\begin{center}
			\includegraphics [width=.35\textwidth]{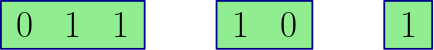}
			\caption{The three right boundary dominoes: the right dimer $ B^r_d $, the right $ 2 $-monomer $ B^r_{2m} $ and the right $ 1 $-monomer $ B^r_{1m} $}\label{fig:rightdominos}
		\end{center}
	\end{figure}
\end{enumerate} 

Basic (i.e. untruncated) dominoes are also allowed at either boundary of a finite interval $\Lambda$. In this case we say the associated boundary condition is empty. Thus, the sets of possible boundary conditions $B=(B^l,B^r)$ for a root tiling on $\Lambda$ are, respectively,
\begin{equation}\label{eq:bdy_cond}
B^l \in \{\emptyset, \, B_d^l\}, \quad B^r\in \{\emptyset, \, B_d^r, \, B_{1m}^r, \, B_{2m}^r\}.
\end{equation}
We define a \emph{root tiling} $ R=(B,V,M)$ of a finite interval $\Lambda$ as a tiling defined by a set of boundary conditions $B=(B^l,B^r)$, a set of voids $V$, and a set of basic monomers $M$. Since each void has length one we often identify $V$ as a subset of $\Lambda$. The set of all root tilings of $\Lambda$ is denoted by $ \mathcal{R}_\Lambda $. It will sometimes be useful to consider the \emph{ordered root domino tiling} (or for short, in a slight abuse of language: root tiling) defined by $R\in\cR_\Lambda$ which we denote by $\pmb{D}_R = (D_1, \, \ldots, \, D_n)$. Here, each $D_i$ is either a void, monomer or boundary tile.

The set of all VMD tilings will once again be defined using a replacement rule. As such, we introduce two additional \emph{right truncated dimers}, which arise from replacing a basic monomer and neighboring right-boundary monomer with a dimer. These complete the list of nontrivial boundary tiles and are defined as follows, cf.~Figure~\ref{fig:Leftdimers}:
\begin{enumerate}
	\itemsep0pt
	\item  $B_{1d}^r$ is the \emph{truncated $ 1 $-dimer}. It has length four and two particles on its second and third site.
	\item $B_{2d}^r$ is the \emph{truncated $ 2 $-dimer}. It has length five and two particles on its second and third site.
\end{enumerate}
\begin{figure}[ht]
	\begin{center}
		\includegraphics [width=.41\textwidth]{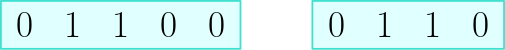}
		\caption{The two truncated dimers used to substitute a basic monomer and neighboring right boundary monomer, see Figure~\ref{fig:rightdominos}. These belong to the set of right boundary dimers, but are not used to define root tilings.}\label{fig:Leftdimers}
	\end{center}
\end{figure}

A \emph{VMD tiling} $\pmb{D} = (D_1, \, \ldots, \, D_n)$ of $ \Lambda $, where  $D_i$ is either a void, monomer, dimer or boundary tile, is any tiling obtained from a root tiling $R\in\cR_\Lambda$ using the following two substitution rules, cf.~Figure~\ref{fig:subst}:
\begin{enumerate}
	\itemsep0pt
	\item Two adjacent basic monomers can be replaced by a basic dimer.
	\item A basic monomer and neighboring 1- or 2-monomer can be replaced by the appropriate truncated dimer.
\end{enumerate}
We note that the replacement rules do not apply to the boundary dimers $B_d^l$ and $B_d^r$. The collection of all VMD tilings derived from a root tiling $R\in\cR_{\Lambda}$ is denoted by  $\mathcal{D}_\Lambda(R) $. Every VMD tiling $\pmb{D}$ is derived from a unique root tiling $ R $. As such, the collection of all \emph{VMD tilings}
\begin{equation}\label{eq:VMD_tilings}
\mathcal{D}_\Lambda := \left\{  \pmb{D}  \, \big| \, \pmb{D} \in  \mathcal{D}_\Lambda(R) \; \mbox{for some $R\in \mathcal{R}_\Lambda $}   \right\} ,
\end{equation}
is equal to the disjoint union of all $\cD_\Lambda(R)$.\\

\begin{figure}[ht]
	\begin{center}
		\includegraphics [width=.65\textwidth]{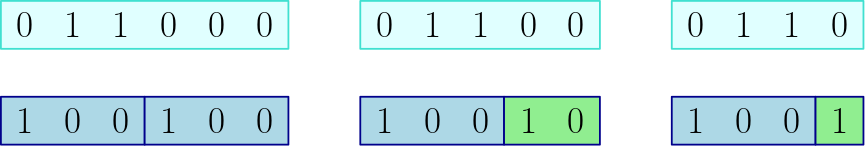}
		\caption{In a root tiling two subsequent monomers may be replaced by a dimer. The three possibilities, including the two boundary cases, are shown in this picture.}\label{fig:subst}
	\end{center}
\end{figure}

\subsection{Definition of VMD states} \label{subsec:VMD_definition}

The specific class of fragmented matrix product states we consider in this work, called \emph{VMD states}, are defined using the VMD tilings introduced in Section~\ref{subsec:tilings}. We now define these states and identify the specific VMD state associated with the squeezed Tao-Thouless state from the fermionic representation. We leave the proof that the VMD states are fragmented matrix product states for Section~\ref{subsec:fragmentation}.

\begin{dfn}\label{def:VMDstate}
	To any   root tiling  $ R \in \mathcal{R}_\Lambda$ of a finite interval $ \Lambda \subset \mathbb{Z}  $, we associate the \emph{VMD state} 
	\begin{equation}\label{eq:VMDstate}
	\psi_\Lambda(R) := \sum_{\pmb{D}  \in  \mathcal{D}_\Lambda(R) } \sigma_{D_1}^+ \dots \sigma_{D_N}^+ | \pmb{0} \rangle  .
	\end{equation}
	The normalized vector $  | \pmb{0} \rangle  \in \mathcal{H}_\Lambda $ is the tensor-product vector of all spins down.
	The sum extends over the collection of VMD tilings  $ \pmb{D} = (D_1,\dots, D_N)  \in\cD_\Lambda(R)$. Each domino tile $ D $ is associated with an operator, which depends on the particle content on that domino
	\begin{equation}\label{eq:diomincrea}
	\sigma_D^+ := \begin{cases} \1 &  \mbox{if $ D $ is a void,} \\  \sigma_x^+ & \mbox{if $ D $ is a monomer starting at $ x $,} \\
	\lambda \, \sigma_{x+1}^+ \sigma_{x+2}^+ & \mbox{if $ D $ is a dimer starting at $ x $,}.  \end{cases} 
	\end{equation} 
	For the operators defined above, we use the convention that the left-boundary dimer $B_d^l$ ``starts'' at $x=\min(\Lambda)-1$. The \emph{VMD subspace} on $ \Lambda $ is
	\begin{equation}\label{eq:vmd_subspace}
	\mathcal{G}_\Lambda := \spa\left\{ \psi_\Lambda(R) \big| \, \mbox{$ R \in \mathcal{R}_\Lambda $} \right\} .  
	\end{equation}
\end{dfn}
We will often consider VMD states on a sub-interval $\Lambda'\subset \Lambda$ for which we use the convention that 
\[\psi_{\Lambda'}(R)=1 \quad \text{if} \quad \Lambda'=\emptyset.\]

A special VMD state on $ \Lambda = [1, 3L ]$ is defined from the root configuration $R_M=(B,V,M)$ where $ B = V = \emptyset $, i.e. the root tiling which only consists of monomers:
\begin{equation}\label{eq:md_state}
\varphi_L := \psi_{ [1, 3L ] }( R_M ) = \sum_{\pmb{D}  \in  \mathcal{D}_{ [1, 3L ] }( R_M ) }  \sigma_{D_1}^+ \dots \sigma_{D_N}^+ | \pmb{0} \rangle . 
\end{equation}
Here the sum extends over all monomer-dimer tilings of $  [1, 3L ] $. In the fermionic language this state coincides with the squeezed Tao-Thouless state~\eqref{eq:TTstate}. As we will see below (cf.~Theorem~\ref{thm:product1}), every VMD state fragments into a product of Tao-Thouless type states and void states. To describe all fragmented MPS on any finite interval, we will also be concerned with cropped versions of \eqref{eq:md_state}.  Appending to $ L $ monomers a right $1$- or $ 2 $-monomer gives rise to the vectors
\begin{align}
\varphi_{L+1}^{(1)}  & := \psi_{ [1, 3L +1] }\left( R_M^1\right)  \notag \\
\varphi_{L+1} ^{(2)}  & := \psi_{ [1, 3L +2] }\left( R_M^2 \right), \label{eq:md_12}
\end{align}
where $R_M^j = \big((\emptyset, B_{jm}^r), \, \emptyset, \, M\big)$. We can extend this definition to $j=3$ by appending a regular length-3 monomer, i.e.
\begin{equation}\label{eq:md_3}
\varphi_{L+1}^{(3)} := \varphi_{L+1}.
\end{equation}

Thus, for $j\in\{1,2,3\}$ and $L\in\bN$ we let 
$ \varphi_{L}^{(j)} \in \mathcal{H}_{[1,3(L-1)+j]} $
denote the VMD state generated by $L$ monomers where the last monomer has length $j$. Notice that any monomer-dimer tiling associated with the state $\vp_L^{(j)}$ always ends in $j-1$ zeros. As a result, the following factorization property holds for $1\leq j\leq k \leq 3$:
\begin{equation}\label{eq:factorize}
\vp_L^{(k)}  = \vp_{L}^{(j)}\otimes\ket{0}^{\otimes k-j}.
\end{equation}

%
%

\subsection{Basic properties of VMD tilings}\label{subsec:tiling_props}

In the obvious way, any VMD tiling is associated with a particle configuration  on $ \Lambda $,
\begin{equation}\label{eq:mapconf}
\pmb{\sigma}_\Lambda : \mathcal{D}_\Lambda  \ \to   \{ 0, 1 \}^\Lambda, \quad 
\pmb{D} \  \mapsto  \pmb{\sigma}_\Lambda(\pmb{D}) .
\end{equation}
The VMD space is a particular subspace contained in the span of all VMD tilings, $\spa \{\ket{\pmb{\sigma}_\Lambda(\pmb{D})} \, \big| \,\pmb{D}\in\cD_\Lambda\}$. The latter subspace also plays an important role throughout this work. In this section we establish several key properties related to VMD tilings. We first show that any two different domino tilings $\pmb{D},\, \pmb{D}'\in\cD_\Lambda$, cannot produce the same particle configuration, i.e. $ \pmb{\sigma}_\Lambda(\pmb{D})\neq\pmb{\sigma}_\Lambda(\pmb{D}')$.
\begin{lem}[Injectivity]\label{lem:injective}
	The map $  \pmb{\sigma}_\Lambda:   \mathcal{D}_\Lambda  \ \to   \{ 0, 1 \}^\Lambda $ is injective for any interval $\Lambda$.
\end{lem}
\begin{proof}
	Suppose $ \pmb{D} \neq \pmb{D}' $ are two distinct VMD tilings, and define $j$ to be the smallest index so that $D_j \neq D_j'$. Both $D_j$ and $D_j'$ begin at the same site $x\in \Lambda$. We go through the possible cases.
%
	\begin{enumerate}[(i)]
		\itemsep0pt
		\item If $ D_j $ is a monomer, then $ \pmb{\sigma}_\Lambda(\pmb{D})_x = 1 \neq 0  = \pmb{\sigma}_\Lambda(\pmb{D}')_x $ unless $  D_j' $ is a left-boundary dimer for which $ \pmb{\sigma}_\Lambda(\pmb{D}')_{x+1} = 1 \neq 0 =  \pmb{\sigma}_\Lambda(\pmb{D})_{x+1}  $.
		\item If $ D_j $ is a left-boundary dimer, then  $ (\pmb{\sigma}_\Lambda(\pmb{D})_{x} , \pmb{\sigma}_\Lambda(\pmb{D})_{x+1} )= (1,1) \neq  (\pmb{\sigma}_\Lambda(\pmb{D}')_{x} , \pmb{\sigma}_\Lambda(\pmb{D}')_{x+1} )$.
		\item If $ D_j $ is a dimer, then $  \pmb{\sigma}_\Lambda(\pmb{D})_x = 0 \neq   \pmb{\sigma}_\Lambda(\pmb{D}')_x $ unless $ D_j' $ is a void, in which case $  (\pmb{\sigma}_\Lambda(\pmb{D}')_{x+1} , \pmb{\sigma}_\Lambda(\pmb{D}')_{x+2} ) \neq (1,1) =  (\pmb{\sigma}_\Lambda(\pmb{D})_{x+1} , \pmb{\sigma}_\Lambda(\pmb{D})_{x+2} ) $. This follows from the fact, that the only domino starting with two particles is a left-boundary dimer.
		\item If $D_j$ is a void, then $ \pmb{\sigma}_\Lambda(\pmb{D})_x = 0 $ and hence $ D_j' $ has to be a dimer. This is excluded by the previous argument, since the domino $ D_{j+1} $ cannot start with two particles.
	\end{enumerate}
	In all of these cases $\sigma_{\Lambda}(\pmb{D})\neq \sigma_{\Lambda}(\pmb{D}')$. Thus, $\sigma_\Lambda$ is injective.
\end{proof}

Thus, for any interval $ \Lambda $ there is a bijection between the set of domino tilings~$ \cD_\Lambda $ and the set of particle configurations~$  \ran\pmb{\sigma}_\Lambda $. The next question to consider is how tilings on nested intervals $\Lambda'\subseteq \Lambda$ are related. It is easy to see that there can be many tilings $\pmb{D}\in\cD_{\Lambda}$ whose particle content agrees with a given $\pmb{D}'\in\cD_{\Lambda'}$. For example, any monomer or dimer on $\Lambda\setminus \Lambda'$ can always be replaced with the appropriate number of voids. However, uniqueness does hold in the opposite direction. Assume $|\Lambda'|\geq 5$, and recall that the boundary tiles were defined so that given any tiling $\pmb{D}=(D_1,\ldots, D_n)\in\cD_{\Lambda}$ there is a domino tiling $\pmb{D}'\in\cD_{\Lambda'}$ for which
\begin{equation}\label{eq:induced_tiling}
\pmb{\sigma}_{\Lambda}(\pmb{D})\restriction_{\Lambda'} = \pmb{\sigma}_{\Lambda'}(\pmb{D}').
\end{equation}
The previous result guarantees that $\pmb{D}'$ is unique and so we call $\pmb{D}'$ the tiling \emph{induced} by $\pmb{D}$ on $\Lambda'$. In case $ \pmb{D}_R $ is a root tiling corresponding to $ R $, the induced tiling on $ \Lambda' $ is also a root tiling and we call the corresponding root $ R' \in \cR_{\Lambda'} $ the \emph{induced root}. 

While the particle content of $\pmb{D}$ and the induced tiling $\pmb{D}'$ agree, it is not guaranteed the tilings are compatible, meaning it is not necessarily the case that 
\begin{equation}\label{eq:restriction}
\pmb{D}'=(D_i, \ldots, D_j)
\end{equation}
for some $1 \leq i \leq j \leq n$. This only occurs if the truncation of $\pmb{D}$ to $\Lambda'$ cuts between the boundaries of tiles and not through the interior of a domino. If such $i,j$ do exist, we call $\pmb{D}'$ the \emph{restriction} of $\pmb{D}$, and say that $\pmb{D}$ can be \emph{restricted} to $\Lambda'$. 
If $\pmb{D}_R = (D_1, \ldots, D_n)$ is a root tiling that can be restricted to $\Lambda'$, then the resulting tiling $\pmb{D}_{R'} = (D_i,\ldots, D_j)$ is itself a root tiling for $\Lambda'$, and we say the associated root configuration $R' \in \cR_{\Lambda'}$ is the \emph{restriction of $R$ onto $\Lambda'$}.\\

There is one final situation of interest. Assume $\Lambda=\Lambda_1\cup\Lambda_2$ for two subintervals $\Lambda_1$, $\Lambda_2$ and suppose that $\pmb{D}_1\in\cD_{\Lambda_1}$, $\pmb{D}_2\in\cD_{\Lambda_2}$ are two VMD tilings whose particle content agrees on $\Lambda_1\cap \Lambda_2$. As we show in the next result, if $|\Lambda_1\cap \Lambda_2|\geq 6$ then there is a unique tiling $\pmb{D}\in\cD_{\Lambda}$ whose particle content agrees with both $\pmb{D}_1$ and $\pmb{D}_2$. A counterexample when $|\Lambda_1\cap\Lambda_2|= 5$ is given by $\pmb{\sigma}_{\Lambda_1}(\pmb{D}_1)=110001$ and $\pmb{\sigma}_{\Lambda_2}(\pmb{D}_2)=100011$. There is a unique configuration $\pmb{\sigma}\in\{0,1\}^\Lambda$ that agrees with the particle content of each tiling, namely $\pmb{\sigma}=1100011$, however this configuration does not correspond to a VMD tiling. 

\begin{lem}[Intersection]\label{lem:intersecD}
	Consider an interval $ \Lambda \subset \mathbb{Z} $ composed of three consecutive intervals $ \Lambda_l $, $ \Lambda_m $ and $ \Lambda_r $  and
	let $ \pmb{D}_{1} \in \mathcal{D}_{\Lambda_{1}} $ and $ \pmb{D}_{2} \in \mathcal{D}_{\Lambda_{2}} $ 
	be VMD tilings on $\Lambda_{1} :=  \Lambda_l \cup \Lambda_m $ and $\Lambda_{2} :=  \Lambda_m \cup \Lambda_r $, respectively. Assume $ | \Lambda_m| \geq 6 $ and that the particle content of both tilings agree on $ \Lambda_m $.
	Then there is a unique VMD tiling $ \pmb{D} \in \mathcal{D}_{\Lambda} $ whose particle content agrees with both $\pmb{D}_{1}$ and $\pmb{D}_{2}$.
\end{lem}
%

\begin{proof}
	The assumptions guarantee the existence of a unique configuration $\pmb{\sigma}\in\{0,1\}^\Lambda$, determined by the particle content of $\pmb{D}_1$ and $\pmb{D}_2$. We construct a VMD tiling $\pmb{D}\in\cD_{\Lambda}$ that agrees with this particle configuration, i.e.~such that $\pmb{\sigma} = \pmb{\sigma}_\Lambda(\pmb{D})$. This completes the proof since the injectivity of $\pmb{\sigma}_\Lambda$ implies uniqueness.
	
	First consider the case that $\pmb{D}_1$ can be restricted to $\Lambda_l$. Then there are tilings $\pmb{D}_\# \in \cD_{\Lambda_\#}$ for $\# \in \{ l , m \}$ such that $\pmb{D}_1=(\pmb{D}_l,\pmb{D}_m)$. The tiling $\pmb{D}_2$ cannot start with a left boundary dimer since the particle content of $\pmb{D}_1$ and $\pmb{D}_2$ agree on $\Lambda_m$, and so the concatenated tiling $\pmb{D} = (\pmb{D}_l, \pmb{D}_2)\in\cD_\Lambda$ is the desired VMD tiling on $\Lambda$. A similar construction will produce the desired tiling $\pmb{D}=(\pmb{D}_1, \pmb{D}_r)\in\cD_\Lambda$ if $\pmb{D}_2$ can be restricted to $\Lambda_r$.
	
	Assume that $\pmb{D}_1$ (respectively, $\pmb{D}_2$) cannot be restricted to $\Lambda_l$ (respectively, $\Lambda_r$). Fix $i=1,2$ and let $\pmb{D}_m^i\in\cD_{\Lambda_m}$ be the tiling induced by $\pmb{D}_i$ on $\Lambda_m$. Since the particle content of $\pmb{D}_1$ and $\pmb{D}_2$ agree on the intersection,
	\[
	\sigma_{\Lambda_m}(\pmb{D}_m^1)=\sigma_{\Lambda_m}(\pmb{D}_m^2),
	\]
	and so by injectivity $\pmb{D}_m^1=\pmb{D}_m^2:=\pmb{D}_m$. We proceed by considering all possible cases for $\pmb{D}_m$.
	
	If $\pmb{D}_m$ begins with a left boundary dimer, the truncation of $\pmb{D}_1$ onto $\Lambda_m$ must have resulted from a cut through the first and second sites of a dimer. Moreover, since the particle content of $\pmb{D}_1$ and $\pmb{D}_2$ agree on $\Lambda_m$, $\pmb{D}_2$ must also start with a left-boundary dimer. Let $\Lambda_l^+$ be the interval consisting of $\Lambda_l$ and the first five sites of $\Lambda_m$, and set $\Lambda_2^-= \Lambda_2\setminus \Lambda_l^+$. Then $\pmb{D}_1$ (respectively, $\pmb{D}_2$) can be restricted to $\Lambda_l^+$ (respectively, $\Lambda_2^-$) and the concatenation of the two resulting restrictions $\pmb{D}=(\pmb{D}_l^+,\pmb{D}_2^-)$ produces the desired VMD tiling on $\Lambda$. A similar procedure will construct the desired VMD tiling $\pmb{D}=(\pmb{D}_1^-,\pmb{D}_r^+)\in\cD_\Lambda$ if $\pmb{D}_m$ ends in a right boundary dimer.
	
	We are left to consider the situation that $\pmb{D}_m$ does not begin or end with a boundary dimer. Let $i=1,2$ be arbitrary, and suppose that the truncation of $\pmb{D}_i$ to $\Lambda_m$ cuts through a domino whose support overlaps with $\Lambda_m$ by $n_i$ sites. Since $\pmb{D}_m$ does not contain any boundary dimers,
	\[
	n_1 \leq 4, \quad n_2\leq 2.
	\]
	Let $\Lambda_l^+$ be the union of $\Lambda_l$ and the first $n_1$ sites of $\Lambda_m$, and $\Lambda_r^+$ be the union of $\Lambda_r$ and the last $n_2$ sites of $\Lambda_m$. Note that $\Lambda_l^+\cap\Lambda_r^+ = \emptyset$ since $|\Lambda_m|\geq 6$. By construction, $\pmb{D}_1$ (respectively, $\pmb{D}_2$) can be restricted to $\Lambda_l^+$ (respectively, $\Lambda_r^+$). If $\Lambda_m^- := \Lambda \setminus (\Lambda_l^+\cup \Lambda_r^+) = \emptyset$, then the concatenation of these two restrictions $\pmb{D}=(\pmb{D}_l^+, \pmb{D}_r^+)\in\cD_\Lambda$ is the desired VMD tiling. Otherwise, both $\pmb{D}_1$ and $\pmb{D}_2$ can be restricted to $\Lambda_m^-$. Since the particle content of these two restrictions agree, these are the same tiling, which we denote by $\pmb{D}_m^-$. In this case, $\pmb{D}=(\pmb{D}_l^+, \pmb{D}_m^-, \pmb{D}_r^+)\in\cD_{\Lambda}$ is the desired VMD tiling.
\end{proof}

Notice that the bijection between tilings~$\cD_{\Lambda}$ and particle configurations~$\ran(\pmb{\sigma}_\Lambda)$ naturally lifts to a bijection between tilings and the canonical orthonormal basis for the VMD tiling space
\begin{equation}\label{eq:tiling_subspace}
\caC_{\Lambda} := \spa\left\{ \ket{\pmb{\sigma}_\Lambda(\pmb{D})} \, \big| \, \pmb{D}\in \cD_{\Lambda} \right\}\subset \cH_{\Lambda}.
\end{equation}
This subspace plays an important role in the proof of Theorem~\ref{thm:gap} in Section~\ref{sec:MM}. A corollary of Lemma~\ref{lem:intersecD} is the following result regarding the factorization of the orthogonal projection onto the VMD tiling space.

\begin{cor}[Factorization of tiling projections]\label{cor:config_projections} 
	Suppose that $\Lambda$ is the union of three consecutive, finite intervals $ \Lambda_l $, $ \Lambda_m $ and $ \Lambda_r $ with $|\Lambda_m|\geq 6$. Define $\Lambda_{1} :=  \Lambda_l \cup \Lambda_m $ and $\Lambda_{2} :=  \Lambda_m \cup \Lambda_r $. Then,
	\begin{equation}
	C_\Lambda = C_{\Lambda_1}C_{\Lambda_2} = C_{\Lambda_2}C_{\Lambda_1},
	\end{equation}
	where $C_{\Lambda'}$ is the orthogonal projection onto $\caC_{\Lambda'}\otimes\cH_{\Lambda\setminus \Lambda'}$ for any $\Lambda'\subset \Lambda$. 
\end{cor}

\begin{proof}
	For any interval $\Lambda'\subseteq \Lambda$ an explicit expression for the projection is  $C_{\Lambda'} =\sum_{\pmb{\sigma}\in \cS_{\Lambda'}} \ketbra{\pmb{\sigma}}$, where \[\cS_{\Lambda'}=\left\{\pmb{\sigma}\in\{0,1\}^\Lambda \, \big| \, \pmb{\sigma}\restriction_{\Lambda'} = \pmb{\sigma}_{\Lambda'}(\pmb{D}') \; \mbox{for some} \; \pmb{D}'\in\cD_{\Lambda'}\right\}.\] 
	Moreover, since the tiling $\pmb{D}'\in\cD_{\Lambda'}$ induced by any $\pmb{D}\in\cD_\Lambda$ satisfies $\pmb{\sigma}_\Lambda(\pmb{D})\restriction_{\Lambda'}=\pmb{\sigma}_{\Lambda'}(\pmb{D}')$ it follows that $\ran \pmb{\sigma}_\Lambda \subset \cS_{\Lambda'}$ if $ |\Lambda'| \geq 5 $.
	
	We now consider the intervals $\Lambda_i$ for $i=1,2$ and let $\pmb{\sigma}^i\in \{0,1\}^\Lambda$ be any configuration so that $\pmb{\sigma}^i\restriction_{\Lambda_i} = \pmb{\sigma}_{\Lambda_i}(\pmb{D}_i)$ for some $\pmb{D}_i\in\cD_{\Lambda_i}$. Trivially, $\braket{\pmb{\sigma}^1}{\pmb{\sigma}^2}=\delta_{\pmb{\sigma}^1,\pmb{\sigma}^2}$, and the scalar product is nonzero only if the particle content of $\pmb{D}_1$ and $\pmb{D}_2$ agree on $\Lambda_m$. In this case, Lemma~\ref{lem:intersecD} implies there is a unique $\pmb{D}\in \cD_{\Lambda}$ so that $\pmb{\sigma}^1=\pmb{\sigma}^2=\pmb{\sigma}_\Lambda(\pmb{D})$. Combining these observations,
	\[
	C_{\Lambda_1}C_{\Lambda_2} = \sum_{\substack{\pmb{\sigma}^1\in \cS_{\Lambda_1} }}\sum_{\pmb{\sigma}^2\in\cS_{\Lambda_2}}\delta_{\pmb{\sigma}^1,\pmb{\sigma}^2}| \pmb{\sigma}^1 \rangle \langle \pmb{\sigma}^2| = \sum_{\pmb{D}\in \cD_{\Lambda}} \ketbra{\pmb{\sigma}_\Lambda(\pmb{D})} = C_\Lambda, 
	\]
	as claimed. The other equality follows from the self-adjointness of $ C_{\Lambda_i} $,  $i=1,2$.
\end{proof}

As we will show in Lemma~\ref{lem:char}, the VMD states form an orthogonal basis for the VMD-space $\cG_{\Lambda}$. Moreover, $\cG_\Lambda$ is the ground-state space of the particle preserving Hamiltonian $H_\Lambda$ for $|\Lambda|\geq 8$, cf.~Theorem~\ref{thm:gss}. Two natural questions are the maximal filling and number of VMD states, which can be answered by considering the root tilings $R=(B,V,M)\in\cR_{\Lambda}$. Let $\pmb{D}_R$ be the domino tiling defined by $R$ and define by
\begin{equation}\label{eq:particlenumber}
N_\Lambda(R) := \sum_{x \in \Lambda}  \pmb{\sigma}_\Lambda(\pmb{D}_R)_x 
\end{equation}
the total number of particles in $R$. For finite $\Lambda$, this number attains a maximum,
\begin{equation}\label{eq:maxN}
N_\Lambda^\textrm{max} := \max\left\{  N_\Lambda(R) \, \big| \, R\in \mathcal{R}_\Lambda \right\} \, . 
\end{equation}
The ratio $ N_\Lambda^\textrm{max} /|\Lambda| $ is the maximal filling factor, and converges to $1/3$ in the thermodynamic limit.
\begin{lem}[Maximal filling]\label{lem:maxfill}
	For any finite interval $ \Lambda \subset \mathbb{Z} $ and any $ N \in \{0, 1, \dots, N_\Lambda^\textrm{max}  \} $ there is a root tiling $ R \in \mathcal{R}_\Lambda $ such that $ N =  N_\Lambda(R) $. 
	The maximal number of particles satisfies
	\begin{equation}\label{eq:maxfill}
	\frac{1}{3} \leq  \frac{N_\Lambda^\textrm{max}}{|\Lambda| } \leq  \frac{1}{3} + \frac{4}{3|\Lambda|},
	\end{equation}
	such that in the thermodynamic limit $ \lim_{|\Lambda| \to \infty }  N_\Lambda^\textrm{max}\big/|\Lambda| = 1/3 $.  
\end{lem}
\begin{proof}
	The first assertion is immediate in the case $ N = N_\Lambda^\textrm{max}  $ for which there is a maximizing root tiling $R$. Starting from this maximizing root tiling  we may successively substitute (i) a monomer by voids or (ii) a boundary dimer by a monomer and (possibly) voids to produce a new root-configuration that decreases $ N $ by one until we reach the empty configuration.
	
	Every finite volume supports a root tiling $R_j$ associated with one of the states $\vp_{L}^{(j)}$ from \eqref{eq:md_12}-\eqref{eq:md_3}. These root tilings consist exclusively of basic and right-monomers, implying that $ N_\Lambda^\textrm{max} \geq |\Lambda|/3$. Alternatively, the particle content of a root tiling associated with a choice of boundary conditions $B=(B^l,B^r)$ is maximized by filling the interior with as many monomers as possible. Therefore, for any $R=(B,V,M)\in\cR_\Lambda$,
	\[
	N_\Lambda(R) \leq \frac{1}{3}\left(|\Lambda|-|B^l|-|B^r|\right) + N(B^l) +N(B^r)
	\]
	where $|B^\#|$ and $N(B^\#)$ are the length and number of particles, respectively, associated with the boundary tile $B^\#$. The above inequality is optimized when $B=(B_d^l,B_d^r)$ which yields the upper bound in~\eqref{eq:maxfill} and hence concludes the proof.	
%
\end{proof}

As the next result shows, the cardinality of $ \mathcal{R}_\Lambda $ grows exponentially in the system size $|\Lambda|$. This produces an estimate on the number of VMD states as these are uniquely characterized by root tilings.
\begin{lem}[Exponential growth]\label{lem:rootgrowth}
	There are constants $ c, C \in (0, \infty) $ such that for all sufficiently large intervals  $ \Lambda \subset \mathbb{Z} $:
	\begin{equation}\label{eq:rootgrowth}
	c \, \mu^{|\Lambda|} \leq   \left| \mathcal{R}_\Lambda \right| \leq C\, \mu^{|\Lambda|}
	\end{equation}
	where $ \mu $ is the unique real solution of  $ \mu^2 (\mu-1)  = 1 $ which is approximately $1.47$.
\end{lem}
\begin{proof}
	Any root tiling consists of voids, monomers and boundary dominoes. A lower bound on $  \left| \mathcal{R}_\Lambda \right|  $ is obtained by disregarding the boundary dominoes
	and counting the number $ r_{|\Lambda|} $ of root tilings obtained from just voids and monomers. 
	To do so, we use the recursion relation for  $ n = |\Lambda | \geq 4 $
	\begin{equation}\label{eq:recrelG}
	r_n = r_{n-1} + r_{n-3} , 
	\end{equation}
	whose initial values are  $ r_1 =  r_2 = 1 $ and $ r_3 = 2 $. 
	The relation results from the observation that (i)~placing a void at the first site reduces the counting problem on $ n $ sites to that on $ n-1 $ sites, 
	and (ii)~placing a monomer 
	on the first 3 sites reduces the problem to that on $ n-3 $ sites. Iterating~\eqref{eq:recrelG} with the above initial condition yields 
	\begin{equation}\label{eq:dim_recursion}
	\left(\begin{matrix} r_n \\ r_{n-1} \\ r_{n-2} \end{matrix} \right) = A^{n-3} \left(\begin{matrix} 2 \\ 1 \\ 1 \end{matrix} \right) \, , 
	\quad A := \left(\begin{matrix} 1 & 0 & 1 \\ 1 & 0 & 0  \\ 0 & 1 & 0  \end{matrix} \right) \, . 
	\end{equation}
	As can be seen from its characteristic polynomial, $ \xi^3 - \xi^2 - 1 $, the matrix $ A $ has three distinct eigenvalues, one of which is real and given by $ \mu $.
	The other two, $ \nu_{\pm} $, are complex conjugates and their modulus is strictly smaller than one. Expressing the initial vector $ (2, 1, 1 )^T = c_0 v_0 + c_+ v_+  + c_- v_- $ in terms of the three linearly independent eigenvectors $ v_0, v_\pm \in \mathbb{C}^3 $ yields 
	\begin{equation}\label{eq:exactgrowth}
	r_n = c_0 \mu^{n-3}  (v_0)_1 + c_+  \nu_+^{n-3} (v_+)_1 + c_- \nu_-^{n-3}  (v_-)_1 \, . 
	\end{equation}
	Since $ A^4 $ is a matrix with strictly positive entries, the Perron-Frobenius theorem ensures that the vector $ v_0 $ can be chosen with strictly positive entries and hence $ (v_0)_1 > 0 $. Taking $ n $ large enough, the last two terms on the right side can be made arbitrarily small.
	From this we conclude that $ c_0 > 0 $ (since otherwise $ r_n $ would be negative for large $n$) and hence the lower bound in~\eqref{lem:rootgrowth}.
	
	For the upper bound, we note that there are $ (1+1) \times (1+3) = 8 $ choices of boundary dominoes for $ |\Lambda | \geq 9 $, see \eqref{eq:bdy_cond}. 
	Once boundary tiles are set (and hence the number  $b$  of sites covered by those tiles), the problem reduces to counting the number of tilings of $n=|\Lambda|-b$ sites with monomers and voids. We thus  arrive at 
	\begin{equation}
	\left| \mathcal{R}_\Lambda \right| = r_{|\Lambda|} +   r_{|\Lambda|-1} +  r_{|\Lambda|-2} +  r_{|\Lambda|-3} + r_{|\Lambda|-5} + r_{|\Lambda|-6} + r_{|\Lambda|-7} + r_{|\Lambda|-8} 
	\end{equation} 
	for any $ |\Lambda | \geq 9 $.  The upper bound in~\eqref{eq:rootgrowth} is apparent after inserting~\eqref{eq:exactgrowth}.
\end{proof}

The recursion relation \eqref{eq:dim_recursion} can also be used to determine the number of root-configurations for intervals $|\Lambda|<9$. However, certain cases will need to be excluded. Setting $b$ to be the number of sites covered by a choice of boundary tiles $(B^l,B^r)$, and introducing the value $r_0=1$, the number of root configurations is given by
\begin{equation}\label{eq:VMD_cardinality}
|\cR_{\Lambda}| = \sum_{(B^l,B^r) \, : \, b\leq |\Lambda|} r_{|\Lambda|-b}.
\end{equation}

\subsection{Basic properties of VMD states}\label{subsec:VMD_props}

In this section we start to describe the ground state space in terms of the VMD space. We first provide a characterization of the VMD space. We then establish some useful properties of the VMD space, including maximal filling and dimension, and prove that $\cG_\Lambda\subseteq \ker(H_\Lambda)$ in Lemma~\ref{lem:gsorth}.

Every VMD state has a canonical expression in terms of the particle configurations, namely
\begin{equation}\label{eq:VMDocc}
\psi_\Lambda(R) = \sum_{\pmb{D}  \in  \mathcal{D}_\Lambda(R) } \lambda^{\#(\pmb{D}) } \,   |  \pmb{\sigma}_\Lambda(\pmb{D})  \rangle 
\end{equation}
where  $ \#(\pmb{D}) $ denotes the number of dimers in $ \pmb{D} $, cf.~\eqref{eq:VMDstate}. For the derivation of this representation, note that by~\eqref{eq:diomincrea} each dimer operator contributes a factor of $ \lambda $. The subspace $ \mathcal{G}_\Lambda $ of all VMD states introduced in~\eqref{eq:vmd_subspace} is uniquely characterized by being supported on particle configurations in $\ran \pmb{\sigma}_\Lambda $ together with a simple hierarchical structure of its weights.
\begin{lem}[Characterization of the VMD subspace]\label{lem:char}
	A vector $\psi \in  \mathcal{G}_\Lambda $ if and only if both of the following properties hold:
	\begin{enumerate}
		\itemsep0pt
		\item $ \psi $ is supported on particle configurations in $ \ran\pmb{\sigma}_\Lambda $, i.e. $ \psi(\pmb{\sigma}) = 0 $ for all $\pmb{\sigma} \not\in  \ran\pmb{\sigma}_\Lambda $,
		\item For any VMD tiling $ \pmb{D} $ and any two subsequent monomers $ D_j $, $ D_{j+1} $  (possibly a right-monomer) in $ \pmb{D} $
		\begin{equation}\label{eq:weightchar}
		\psi\left( \pmb{\sigma}_\Lambda( \mathcal{M}_j(\pmb{D})) \right) = \lambda \, \psi\left( \pmb{\sigma}_\Lambda(\pmb{D})\right)  ,
		\end{equation}
		where $ \mathcal{M}_j(\pmb{D}) \in  \mathcal{D}_\Lambda $ is the VMD tiling in which these monomers are substituted by a dimer. 
	\end{enumerate}
\end{lem}
\begin{proof}
	If $ \psi \in  \mathcal{G}_\Lambda $, then the two properties are straightforward from the representation~\eqref{eq:VMDocc} and the substitution rules.
	Conversely, any $ \psi \in \mathcal{H}_\Lambda $ supported on particle configurations in $ \ran\pmb{\sigma}_\Lambda $ can be written as a linear combination
	\begin{equation}
	\psi = \sum_{\pmb{D} \in \mathcal{D}_\Lambda } c_{\pmb{D}} \, |  \pmb{\sigma}_\Lambda(\pmb{D})  \rangle 
	\end{equation}
	for complex coefficients $c_{\pmb{D}}:=\psi(\pmb{\sigma}_\Lambda(\pmb{D}))$. Recall that each $ \pmb{D}\in\cD_\Lambda$ is associated with a unique $R\in\cR_\Lambda$ and hence root tiling $  \pmb{D}_R $, cf.~\eqref{eq:VMD_tilings}. The second property ensures that 
	\begin{equation}
	c_{\pmb{D}} = \lambda^{\#(\pmb{D}) - \#(\pmb{D}_R)} \, c_{  \pmb{D}_R} 
	\end{equation}
	where $\#(\pmb{D}_R) $ is the number of boundary dimers in $\pmb{D}_R$. Comparing this with \eqref{eq:VMDocc} establishes that $ \psi $ is a linear combination of VMD states $ \psi_\Lambda(R) $. 
\end{proof}

The following lemma summarizes other fundamental aspects of the VMD subspace, including that the set of VMD states is an orthogonal basis for $\cG_\Lambda$.

\begin{lem}[Ground state, orthogonality and filling]\label{lem:gsorth}
	Let $ \Lambda \subset \mathbb{Z} $ be a finite interval and $R\in\cR_{\Lambda}$ be a root tiling.
	\begin{enumerate}
	\itemsep0pt
		\item The VMD state $ \psi_\Lambda(R) $ is a ground-state of $ H_\Lambda $, i.e.\ $ H_\Lambda  \psi_\Lambda(R) = 0 $. 
		\item VMD states with distinct root tilings $R,R'\in\cR_{\Lambda}$ are orthogonal:
		\begin{equation}\label{eq:ortho}
		\left\langle \psi_\Lambda(R')  , \ \psi_\Lambda(R) \right\rangle = \delta_{R,R'} \left\| \psi_\Lambda(R) \right\|^2 
		\end{equation}
		where the Kronecker delta yields one only in case that $B=B'$, $V=V'$ and $M=M'$ where $R=(B,V,M)$ and $R'=(B',V',M')$. Moreover, the normalization is a polynomial in $ |\lambda|^2 $,
		\begin{equation}\label{eq:norm} 
		\left\| \psi_\Lambda(R) \right\|^2 = \sum_{\pmb{D}  \in  \mathcal{D}_\Lambda(R) } |\lambda|^{2\#(\pmb{D}) } .
		\end{equation}
		\item The map $ \mathcal{R}_\Lambda \ni R \mapsto  \psi_\Lambda(R)  $ is injective. Thus, 
		$ \dim \mathcal{G}_\Lambda  = \left| \mathcal{R}_\Lambda \right| $
		is exponentially large in $ |\Lambda | $.
		\item Any VMD state $  \psi_\Lambda(R) $ is an eigenstate of the number operator $ N_\Lambda  = \sum_{x\in \Lambda} \frac{1}{2}(\sigma_x^3 + 1) $. Specifically,
		\begin{equation}
		N_\Lambda \psi_\Lambda(R) = N_\Lambda(R)  \psi_\Lambda(R) \, . 
		\end{equation}
	\end{enumerate} 
\end{lem}
\begin{proof}
	1.~By construction, the particle configuration $   \pmb{\sigma}_\Lambda(\pmb{D}) $ of any VMD tiling $\pmb{D}  \in  \mathcal{D}_\Lambda(R) $ does not have 
	consecutive particles that are two sites apart. Hence, if $ \Lambda = [a,b] $, we have $ n_x n_{x+2}  \psi_\Lambda(R)  = 0 $ for all $ a\leq x \leq b-2$. It remains to show that for all $a\leq x\leq b-3 $,
	\begin{equation}\label{eq:qxzero}
	q_x  \psi_\Lambda(R) =  \sum_{\pmb{D}  \in  \mathcal{D}_\Lambda(R) } \lambda^{\#(\pmb{D}) } \, q_x   |  \pmb{\sigma}(\pmb{D})  \rangle = 0 .
	\end{equation}
	Note that the left and right dimers, $B_d^l$ and $B_d^r$, play no role in the proof of~\eqref{eq:qxzero}, since their particles are located at $ x = a,a+1 $ and $ x= b-1, b $, respectively. 
	The vector $ q_x \,  |  \pmb{\sigma}(\pmb{D})  \rangle  = \left( \sigma^-_{x+1}  \sigma^-_{x+2} -  \lambda \ \sigma^-_{x}  \sigma^-_{x+3}   \right)  |  \pmb{\sigma}(\pmb{D})  \rangle $ is hence trivially zero except for the following two cases:
	\begin{enumerate}[(i)]
		\itemsep0pt
		\item $ \pmb{D} $ contains a dimer starting at $ x $, which may be one of the truncated dimers at the right boundary, cf.~Figure~\ref{fig:Leftdimers}.
		\item  $ \pmb{D} $ contains two successive monomers starting at $ x $ and $ x+3 $, where the rightmost monomer may be a boundary monomer, cf.~Figure~\ref{fig:rightdominos}. 
	\end{enumerate}
	For any tiling $\pmb{D}\in  \mathcal{D}_\Lambda(R) $ of type~(i) there is an associated tiling $\pmb{D}'\in  \mathcal{D}_\Lambda(R)$ of type~(ii) obtained via the substitution rule. By~\eqref{eq:weightchar}, these tilings satisfy $\lambda^{\#(\pmb{D}) } =\lambda^{\#(\pmb{D}')+1} $ from which~\eqref{eq:qxzero} follows.\\
	
	2.~We expand the scalar product and use the orthonormality of the states $ | \pmb{\sigma} \rangle $ to get
	\begin{equation}
	\left\langle \psi_\Lambda(R')  , \ \psi_\Lambda(R) \right\rangle = \sum_{\pmb{D}'  \in  \mathcal{D}_\Lambda(R') } \sum_{\pmb{D}  \in  \mathcal{D}_\Lambda(R) }  \overline{ \lambda}^{\#(\pmb{D}') }\lambda^{\#(\pmb{D}) }\,   \delta_{ \pmb{\sigma}_\Lambda(\pmb{D}) ,  \pmb{\sigma}_\Lambda(\pmb{D}')} . 
	\end{equation}
	The injectivity of map $  \pmb{D} \  \mapsto  \pmb{\sigma}_\Lambda(\pmb{D}) $ established in Lemma~\ref{lem:injective} reduces the double sum to the diagonal $ \pmb{D} = \pmb{D}' $. Since each VMD tiling is associated with a unique root tiling, the above sum is non-zero only if $R=R'$. In that case, the sum reduces to the normalization~\eqref{eq:norm}.  \\
	
	3.~This is an immediate consequence of the orthogonality~\eqref{eq:ortho} and the fact that $  \left\| \psi_\Lambda(R) \right\|^2 \neq 0 $.  
	The exponential growth of $ |\mathcal{R}_\Lambda| $ was established in Lemma~\ref{lem:rootgrowth}.  \\
	
	4.~All configurations $\pmb{D}\in\cD_\Lambda(R)$ are obtained from the replacement rules which do not change the particle number. As such, $  \psi_\Lambda(R) $ is a linear combination of eigenstates of the number operator with common eigenvalue given by the number of particles in the root tiling $ N_\Lambda(R) $, cf.~\eqref{eq:particlenumber}.
\end{proof} 

Among all VMD states on an interval $ \Lambda $  the filling fraction $  N_\Lambda(R) / | \Lambda | $ is bounded from above by $ N_\Lambda^\textrm{max} /\Lambda $, which by Lemma~\ref{lem:maxfill} asymptotically tends to $ 1/3 $. The spin version of the squeezed Tao-Thouless state,  $  \varphi_L $, attains this limit. Lower filling fractions are obtained by inserting a positive faction of voids in the root tiling.

\subsection{Fragmentation of VMD states} \label{subsec:fragmentation}

The VMD states constitute a class of fragmented matrix product states. Up to boundary dimers, the building blocks for this fragmentation are voids and squeezed Tao-Thouless states (as well as its truncations). This fragmentation gives a natural way to decompose any VMD state into a product of states. The focus of this section is to prove the fragmentation property and discuss other useful factorizations of the VMD states.


%

\begin{lem}[Fragmentation I]\label{thm:product}
Assume that $ \psi_\Lambda(R) $ is a VMD state for a root-configuration $R=(B,V,M)$ with two consecutive voids $ v, v' \in V $ separated by $ L \geq 0$ monomers. Define $R_l$ and $R_r$ to be the restriction of $R$ onto $ \Lambda_l = \Lambda \cap (-\infty, v] $ and $ \Lambda_r = \Lambda \cap [v', \infty)  $, respectively. Then,
\begin{equation}\label{eq:fragment}
\psi_\Lambda(R)  = \psi_{\Lambda_l}(R_l)   \otimes \varphi_L \otimes \psi_{\Lambda_r}(R_r). 
\end{equation}
Moreover, both the left and right states are themselves products
\begin{equation}\label{eq:product}
 \psi_{\Lambda_l}(R_l) =  \psi_{\Lambda_l'}(R_l') \otimes | 0 \rangle ,  \qquad 
  \psi_{\Lambda_l}(R_r) =   | 0 \rangle  \otimes \psi_{\Lambda_r'}(R_r') ,
\end{equation}
where $R_l'$ and $R_r'$ are the restrictions of $R$ onto $ \Lambda_l' = \Lambda \cap (-\infty, v) $ and $ \Lambda_r' = \Lambda \cap (v', \infty)  $, respectively.
\end{lem}
\begin{proof}
By assumption, the root tiling $R$ has $L$-monomers which lay across an interval $\Lambda_m$ of length $3L$ between $ v< v' $. Let $R_m$ be the restriction of $R$ to $\Lambda_m$. Since voids are unchanged by the replacement rule, the locations of void tiles are fixed for any tiling $\pmb{D}\in\cD_\Lambda(R)$. Therefore, every tiling $\pmb{D}\in\cD_\Lambda$ can be restricted to $\Lambda_l,$ $\Lambda_m$ and $\Lambda_r$. Moreover,
\begin{equation*}
 \cD_\Lambda = \left\{ \pmb{D} =  (\pmb{D}_l ,   \pmb{D}_m ,   \pmb{D}_r ) \; \big| \; D_j \in\cD_{\Lambda_j}(R_j) \;\mbox{ for } j=l,m,r \right\}.
\end{equation*}
where $  \pmb{D}_l $ ends with a void at $v$, $  \pmb{D}_r $ begins with a void at $v'$, and $\pmb{D}_m $ is a monomer-dimer tiling. 
The sum in~\eqref{eq:VMDocc} thus expands to a sum over the three sets $\cD_{\Lambda_j}(R_j)$, $j=l,m,r$, individually. Since 
\[ \#(  \pmb{D} ) = \#(  \pmb{D}_l ) + \#(  \pmb{D}_m ) +\#(  \pmb{D}_r ) ,\] 
the three sums reduce to the factors on the right side of~\eqref{eq:fragment}. 

The second claim~\eqref{eq:product} is immediate from the disjoint unions $V_l = V_l' \, {\cup}\, \{ v \} $ and $V_r = \{ v' \} \, {\cup} \, V_r' .$	
\end{proof}

Iterating the above result and distinguishing all possible boundary cases, we thus arrive at the following general form of VMD states. 
As discussed earlier, the squeezed Tao-Thouless state has a matrix product representation~\cite{Nakamura:2012bu} and so the next theorem shows that the VMD states are an example of a fragmented MPS. 
	\begin{theorem}[Fragmentation II]\label{thm:product1}
	Let $ \Lambda $ be an interval and $R=(B,V,M)\in \mathcal{R}_\Lambda $ be a root configuration with $ K\in \mathbb{N} $ voids, i.e.\ $ V = \{ v_1,\ldots , v_K\} $. Then there are $ L_1, \ldots , L_{K-1} \in \mathbb{N}_0 $ and two boundary states $ \psi^{l} , \psi^{r} $ such that 
	\begin{equation}\label{eq:genVMDform}
	\psi_{\Lambda}(R) =  \psi^{l}  \otimes | 0 \rangle_{v_1} \otimes \varphi_{L_1} \cdots  \varphi_{L_{K-1}} \otimes | 0 \rangle_{v_K} \otimes   \psi^{r} .
	\end{equation}
	The boundary states are of the form for some $ L_0 , L_K \in \mathbb{N}_0 $ 
	\begin{align}\label{eq:bdy_states}
	\psi^{l} = \begin{cases}  \lambda | \pmb{\sigma}_{B^l_d} \rangle \otimes  \varphi_{L_0} & \mbox{if $ B^l = B^l_d $} \\
	\varphi_{L_0} & \mbox{if $ B^l = \emptyset $,}  \end{cases} \qquad 
	\psi^{r} = \begin{cases}  \lambda \varphi_{L_K}  \otimes  | \pmb{\sigma}_{B^r_d} \rangle& \mbox{if $ B^r = B^r_d $} \\
	\varphi_{L_K}^{(j)}  & \mbox{if $ B^r =B^r_{jm} $ and $ j = 1,2$} \\
	\varphi_{L_K} & \mbox{if $ B^r = \emptyset $,} \end{cases} 
	\end{align}
	with $  \pmb{\sigma}_{B^l_d}  = 11000 $ and $ \pmb{\sigma}_{B^r_d} = 011 $. In the above, we again employ the convention that $ \varphi_{L} = 1 $ (is absent) if $ L = 0 $.
\end{theorem} 
\begin{proof}
	Assume first that $K\geq 2$ and define $\Lambda_l = \Lambda \cap (-\infty, v_1)$ and $\Lambda_r = \Lambda \cap (v_K,\infty)$. Iteratively applying Lemma~\ref{thm:product} to successive voids $v_i,\,v_{i+1}$ produces the factorization
	\[\psi_{\Lambda}(R) = \psi_{\Lambda_l}(R_l)\otimes \ket{0}_{v_1}\otimes \vp_{L_1}\ldots \vp_{L_{K-1}}\otimes \ket{0}_{v_K}\otimes \psi_{\Lambda_r}(R_r)\]
	where $R_l$ and $R_r$ are the restrictions of $R$ to $\Lambda_l$ and $\Lambda_r$, respectively. Consider the state $\psi^l = \psi_{\Lambda_l}(R_l)$. Since $R_l$ contains no voids, all tilings $\pmb{D}\in\cD_{\Lambda_l}(R_l)$ consist of monomers and dimers (possibly including the left boundary dimer). If $B^l = B_d^l$, then the first tile for any $\pmb{D} \in \cD_{\Lambda_l}(R_l)$ is always the left boundary dimer. Given \eqref{eq:VMDocc} the left boundary state can thus be factored as
	\[
	\psi_{\Lambda_l}(R_l) = \lambda \ket{\pmb{\sigma}_{B_d^l}}\otimes \vp_{L_0}
	\]
	where $L_0\in\bN_0$ is the number of monomers in the root $R_l$. If $B^l =\emptyset$, then the root $R_l$ just consists of (basic) monomers and the result once again follows.  A similar analysis holds for $\psi^r = \psi_{\Lambda_r}(R_r)$ by considering the four possible right boundary conditions.
	
	Now consider $K =1$. Since the placement of the void $v_1$ is invariant over all tilings $\pmb{D}\in\cD_{\Lambda}(R)$, arguing as in Lemma~\ref{thm:product} we can factor
	\begin{equation}\label{eq:single_void}
	\psi_{\Lambda}(R) = \psi_{\Lambda_l}(R_l)\otimes \ket{0} \otimes \psi_{\Lambda_r}(R_r)
	\end{equation}
	where $R_l$ and $R_r$ are the restrictions of $R$ to $\Lambda_l = \Lambda \cap (-\infty,v_1)$ and $\Lambda_r = \Lambda\cap (v_1,\infty)$, respectively. The analysis of the boundary states $\psi^l = \psi_{\Lambda_l}(R_l)$ and $\psi^r = \psi_{\Lambda_r}(R_r)$ follows as in the previous case.
\end{proof}

In the case that there are no voids, i.e. $K=0$, the fragmented form of the VMD state is a little different. In this case, the root tiling $R$ consists of monomers and boundary tiles, and produces a state of the form:
\begin{equation}\label{eq:no_voids}
\psi_\Lambda(R) = \psi^l \otimes \psi^r
\end{equation}
where there is some $L\in\bN_0$ for which
	\begin{align}\label{eq:bdy_states2}
\psi^{l} = \begin{cases}  \lambda | \pmb{\sigma}_{B^l_d} \rangle  & \mbox{if $ B^l = B^l_d $} \\
1 & \mbox{if $ B^l = \emptyset $,}  \end{cases} \qquad 
\psi^{r} = \begin{cases}  \lambda \varphi_{L}  \otimes  | \pmb{\sigma}_{B^r_d} \rangle& \mbox{if $ B^r = B^r_d $} \\
\varphi_{L}^{(j)}  & \mbox{if $ B^r =B^r_{jm} $ and $ j = 1,2$} \\
\varphi_{L} & \mbox{if $ B^r = \emptyset $.} \end{cases} 
\end{align}

The factorization~\eqref{eq:genVMDform} can be read off from the root tiling $\pmb{D}_R$ associated with $R$. 
Moreover, one can use the product structure~\eqref{eq:genVMDform} to write down various factorizations of $\psi_\Lambda(R)$ in terms of restricted root tilings. There are three types of tensor factors in~\eqref{eq:genVMDform}: boundary dimers, voids, and (truncated) squeezed Tao-Thouless states. 
If $\Lambda=\Lambda_1\cup\Lambda_2$ is the disjoint union of two consecutive intervals across which $\psi_{\Lambda}(R)$ transitions between factor types, then the state factorizes as
\begin{equation}\label{eq:restriction_factorization}
\psi_{\Lambda}(R) = \psi_{\Lambda_1}(R_1)\otimes \psi_{\Lambda_2}(R_2)
\end{equation}
where $R_1$ and $R_2$ are the restrictions of $R$ to $\Lambda_1$ and $\Lambda_2$, respectively. For example, if $v_1:=\max(\Lambda_1)\in V$ is the first void, then
$
\psi_{\Lambda_1}(R_1) = \psi^l \otimes \ket{0}_{v_1} $ and $ \psi_{\Lambda_2}(R_2) = \vp_{L_1}\otimes \ldots \otimes\psi^r $.


Not only do the VMD states factor across voids, but they also factorize along any site sufficiently close to a void, which will be useful below.
\begin{lem}[Factorization]\label{lem:factorization}
	Fix $ R=(B, V, M) \in \mathcal{R}_\Lambda$ where $ \Lambda = \Lambda_1 \cup \Lambda _2 $ is the disjoint union of two consecutive, finite intervals such that $|\Lambda_1|\geq 5$ and $|\Lambda_2|\geq 3$. If $ V\cap[ \max \Lambda_1, \min \Lambda_2+2] \neq \emptyset$, then the associated VMD state factorizes as
	\begin{equation}\label{eq:factorization}
	\psi_\Lambda(R)  =  \psi_{\Lambda_1}(R_1)  \otimes  \psi_{\Lambda_2}(R_2) 
	\end{equation}
	with $ R_1\in\cR_{\Lambda_1} $ and $ R_2\in\cR_{\Lambda_2}$ the induced roots. In particular, if $ |\Lambda_2| = 3 $, then there exists $\pmb{\sigma}\in\{0,1\}^{3}$ so that
	\begin{equation}\label{eq:factor000}
	\psi_\Lambda(R)  =  \psi_{\Lambda_1}(R_1)  \otimes  | \pmb{\sigma}\rangle .
	\end{equation}
\end{lem}

The constraint $|\Lambda_1|\geq 5$ and $|\Lambda_2|\geq 3$ is to guarantee that we do not factor in the middle of a left or right boundary dimer. A similar result holds if $|\Lambda_1|<5$ or $|\Lambda_2|<3$, but it is not guaranteed that both of the factors are VMD-states. More generally, if $V\cap[\max{\Lambda_1},\min{\Lambda_2}+2]\neq \emptyset$, one can write
\begin{equation}\label{eq:small_factor}
\psi_{\Lambda}(R) = \begin{cases}
\xi_{1} \otimes \psi_{\Lambda_2}(R_2), & \mbox{if $|\Lambda_1|<5$ and $|\Lambda_2|\geq 3$} \\
\psi_{\Lambda_1}(R_1)  \otimes \xi_{2}, & \mbox{if $|\Lambda_1|\geq 5$ and $|\Lambda_2|< 3$} \\
\end{cases}
\end{equation}
for some $\xi_{j}\in\cH_{\Lambda_j}$.

\begin{proof} We denote by $v=\min \left(V\cap[ \max \Lambda_1, \min \Lambda_2+2]\right)$ the first void in the interval. Applying Theorem~\ref{thm:product1} the state $\psi_{\Lambda}(R)$ factorizes as
	\begin{equation}
	\psi_{\Lambda}(R) = \psi^l \otimes \ldots \otimes\vp_{L}\otimes \ket{0}_{v}\otimes \ldots \otimes \psi^r
	\end{equation}
	for appropriately defined integer $L\in \bN_0$ and boundary states $\psi^l$, $\psi^r$.
	
	If $v \in \{\max \Lambda_1,\,\min \Lambda_2\}$, then the result holds with $R_1$ and $R_2$ the restrictions of $R$ to $\Lambda_1$ and $\Lambda_2$, respectively.
	
	If  $\min\Lambda_2<v\leq \min\Lambda_2+2$, then $L\geq 1$ by the minimality of $v$. Otherwise $L= 0$ would imply that $v$ is preceded by the boundary state $\psi^l = \ket{\pmb{\sigma}_{B_d^l}}$ which contradicts $|\Lambda_1|\geq 5$.  Since $j:=v-\min\Lambda_2 \leq 2$, the claim~\eqref{eq:factorization} follows from~\eqref{eq:factorize} which allows us to define
	\begin{equation*}
	\psi_{\Lambda_1}(R_1) :=  \psi^l \otimes \ldots \otimes\vp_{L}^{(3-j)}, \quad
	\psi_{\Lambda_2}(R_2) :=  \ket{0}^{\otimes j}\otimes \ket{0}_v\otimes \ldots \otimes \psi^r.
	\end{equation*}
	
	Finally, \eqref{eq:factor000} is trivial from \eqref{eq:factorization} as every VMD state on an interval $|\Lambda_2|=3$ is necessarily a configuration.
\end{proof}

\subsection{Recursion relations} \label{subsec:recursion}
In addition to the factorization properties from the previous section, the main building block of VMD states, i.e.\ the squeezed  Tao-Thouless state $ \varphi_n $, has a recursive structure that will be key in our proofs of the spectral gap in Section~\ref{sec:MM}, and the decay of correlations in Section~\ref{sec:correlation_decay}.
\begin{lem}[Recursion relations]\label{lem:recursion}
For any $ n = l + r $ with $ l, r \in \mathbb{N} $:
\begin{equation}\label{eq:recursion}
 \varphi_n =  \varphi_l \otimes  \varphi_r + \lambda \,  \varphi_{l-1} \otimes |  \pmb{\sigma}_d \rangle \otimes   \varphi_{r-1} , 
\end{equation}
where we use the convention that $ \varphi_0 = 1 $, and $ \ket{ \pmb{\sigma}_d} := \ket{\pmb{\sigma}_d^{(3)}} :=  \ket{011000} $ is the configuration corresponding to the basic dimer tile. Moreover, for  all $ n \geq 2 $  and $ j \in \{1,2,3\} $:
\begin{equation}\label{eq:recrel2}
	\varphi_{n}^{(j)}  = \varphi_{n-1}\otimes \vp_{1}^{(j)}+\lambda \, \varphi_{n-2}\otimes\ket{  \pmb{\sigma}_d^{(j)}} \, , 
\end{equation}
where $  \ket{\pmb{\sigma}_d^{(1)}} := \ket{0110} $ and $  \ket{\pmb{\sigma}_d^{(2)}} := \ket{01100}$ correspond to the truncated dimers, cf. Figure~\ref{fig:Leftdimers}. 
\end{lem}
A recursion relation like~\eqref{eq:recursion} easily follows along similar arguments for the cropped states $ j \in \{1,2\} $. Specifically,
\begin{equation}\label{eq:recrel3}
\varphi_n^{(j)} =  \varphi_l \otimes  \varphi_r^{(j)} + \lambda \,  \varphi_{l-1} \otimes |  \pmb{\sigma}_d \rangle \otimes   \varphi_{r-1}^{(j)} 
\end{equation}
holds as in Lemma~\ref{lem:recursion} given $r\geq 2$.
\begin{proof}
The set of all monomer-dimer tilings of $ \Lambda = \Lambda_l \cup \Lambda_r $ with $ \Lambda_l = [1, 3l] $ and $\Lambda_r = [3l+1, 3l+3r ] $ is the disjoint union
of (i)~the union of all monomer-dimer tilings of $  \Lambda_l  $ with the monomer-dimer tilings of $ \Lambda_r $, and (ii)~a dimer on the set $ [3l-2,3l+3] $ together with all monomer-dimer tilings of the remainder $  \Lambda \backslash [3l-2,3l+3]  $. The latter is itself a disjoint union of monomer-dimer tilings 
of the sets $ [1,3l-3] $ and $ [3l+4, 3(l+r)] $. Gathering the terms, completes the proof of \eqref{eq:recursion}. 

For $ j = 3 $ the identity~\eqref{eq:recrel2} is the special case $ r = 1 $ from~\eqref{eq:recursion}. For $ j \in \{ 1, 2\} $ the proof of~\eqref{eq:recrel2} is similar. We decompose the monomer-dimer tilings of $ [1,3(n-1)+j] $ into (i)~all monomer-dimer tilings of $ [1,3(n-1)] $ and a right $ j $-monomer, 
and (ii)~all monomer-dimer tilings of $ [1,3(n-2)] $ and a truncated $ j $-dimer. This completes the proof of \eqref{eq:recrel2}. 
\end{proof}

A simple consequence of~\eqref{eq:recrel2}, which is based on the observation that the two states on the right side
are orthogonal, is the recursion relation for the norms for $ n \geq 2 $ and $ j \in \{ 1, 2 , 3\} $:
\begin{equation}\label{eq:normrec}
  \|  \varphi_n^{(j)}  \|^2 \  =  \|  \varphi_{n-1} \|^2 + |\lambda|^2 \|  \varphi_{n-2}  \|^2.
\end{equation}
The ratio of these norms 
\begin{equation}\label{def:alpha1}
\alpha_n := \frac{\|\vp_{n-1}\|^2}{\|\vp_n\|^2}
\end{equation}
plays an important role in the proofs of both Theorem~\ref{thm:gap} and Theorem~\ref{thm:expcluster}. The closed-form solution for both $\|\vp_n\|^2$ and $\alpha_n$ are given in the next result.

\begin{lem}[Normalization]\label{lem:normalization}
For any $ n \geq 1 $ and $ j \in \{ 1,2, 3 \} $
\begin{equation}
\|  \varphi_n^{(j)}  \|^2 = \frac{\mu_+^{n+1} - \mu_-^{n+1} }{\mu_+ - \mu_- } 
\end{equation}
where  $\mu_{\pm} := (1\pm \sqrt{1+4|\lambda|^2})/2$. Consequently, in terms of the ratio  $ \mu = \frac{\mu_-}{\mu_+} \in (-1,0) $,
\begin{equation}\label{def:alpha}
	\alpha_n = \frac{1}{\mu_+} \frac{1-\mu^n}{1-\mu^{n+1}}. 
	\end{equation}
\end{lem}
Analyzing the above formula shows that $\alpha_{2n}$ (respectively, $\alpha_{2n-1}$) is increasing (respectively, decreasing) in $n$, and converges to $\alpha = \mu_+^{-1}$. Moreover,  $\alpha_{2n}\leq \alpha_{2m-1}$ for any $n,m\in\bN$. 
\begin{proof}
Let $C_n:=\|\vp_n\|^2$. Using the convention $\vp_0 =1$ and the recursion relation~\eqref{eq:normrec}, we can recast the question as a dynamical system with initial conditions $C_0= C_1= 1  $:
	\[
	\begin{pmatrix}
	C_n\\
	C_{n-1}
	\end{pmatrix}
	=
	\begin{pmatrix}
	1 & |\lambda|^2 \\
	1 & 0
	\end{pmatrix}
	\begin{pmatrix}
	C_{n-1} \\
	C_{n-2}
	\end{pmatrix} =:  A \begin{pmatrix}
	C_{n-1} \\
	C_{n-2}
	\end{pmatrix} .
	\]
	Its solution 
	is expressed 
	in terms of the eigenvalues $\mu_{\pm} $ and eigenvectors $(\mu_{\pm}, 1)$ of $ A $,
	\begin{align*} \begin{pmatrix}
	C_n\\
	C_{n-1}
	\end{pmatrix} \ & =   A^{n-1}  \begin{pmatrix}
	1 \\
	1
	\end{pmatrix} \\ \ & =  A^{n-1}  \frac{\mu_+}{\mu_+-\mu_-} \begin{pmatrix}
	\mu_+ \\
	1
	\end{pmatrix} +  A^{n-1}  \frac{\mu_-}{\mu_- - \mu_+} \begin{pmatrix}
	\mu_- \\
	1
	\end{pmatrix}  \\
	& =  \frac{\mu_+^n }{\mu_+-\mu_-} \begin{pmatrix}
	\mu_+ \\
	1
	\end{pmatrix} +  \frac{\mu_-^n}{\mu_- - \mu_+} \begin{pmatrix}
	\mu_- \\
	1
	\end{pmatrix}.	
	\end{align*}
	This completes the proof.
\end{proof}

\subsection{The ground state space and proof of Theorem~\ref{thm:gsenergy}} \label{subsec:gss}

The main goal of this section is to prove the VMD space is the ground state space of $H_\Lambda$ for sufficiently large intervals. We then use this result to establish Theorem~\ref{thm:gsenergy}, which produces a threshold on the ground state energy for a fixed filling fraction. To begin, we show that the VMD subspace satisfies  two properties which they share with the ground-state space of any frustration-free system.

\begin{lem}[Nesting and intersection property]\label{lem:intersection}
	Consider a finite interval $ \Lambda \subset \mathbb{Z} $ composed of three consecutive intervals $ \Lambda_l $, $ \Lambda_m $ and $ \Lambda_r $. 
	\begin{enumerate}
	\itemsep0pt
		\item If $ |\Lambda_m | \geq 5 $, then
			$ \mathcal{G}_\Lambda \subseteq \mathcal{H}_{\Lambda_l} \otimes \mathcal{G}_{\Lambda_m}  \otimes  \mathcal{H}_{\Lambda_r} $. 
		\item If 
	$ |\Lambda_m| \geq 6 $, then:
	\begin{equation}\label{eq:intersect}
	\cG_\Lambda = \left(\mathcal{G}_{\Lambda_{1} } \otimes \mathcal{H}_{\Lambda_r} \right) \cap \left(   \mathcal{H}_{\Lambda_l} \otimes \mathcal{G}_{\Lambda_{2} }\right)
	\end{equation}
	where $ \Lambda_{1}  =   \Lambda_l \cup \Lambda_m $ and $ \Lambda_{2}  =   \Lambda_m \cup \Lambda_r $.
	\end{enumerate}
\end{lem}
\begin{proof}
			1.~The general nesting property follows by iteration from the two special cases $ \Lambda_l = \emptyset $ and $ \Lambda_r= \emptyset $. 
	
	We first spell out the proof in case $ \Lambda_l = \emptyset $. We assume without loss of generality that $ \Lambda_r \neq \emptyset $ and consider the line that cuts $\Lambda$ into $\Lambda_m$ and $\Lambda_r$. Given the fragmented representation of $ \psi_{\Lambda}(R) \in \cG_\Lambda $ from Theorem~\ref{thm:product1}, this line can cut the state in three types of places:
	\begin{enumerate}[(i)]
	\itemsep0pt
		\item Between two consecutive tensor factors (squeezed Tao-Thouless state, void, boundary dimer) in~\eqref{eq:genVMDform}-\eqref{eq:bdy_states}.
		\item In the interior of a boundary dimer.
		\item In the interior of a squeezed Tao-Thouless state, $\vp_n^{(j_0)}$.
	\end{enumerate}
	In either case (i) or (ii), the state factorizes as $\psi_{\Lambda}(R) =\psi_{\Lambda_m}(R_m)\otimes \xi_{\Lambda_r}\in \cG_{\Lambda_m}\otimes \cH_{\Lambda_r}$ where $R_m$ is the root tiling induced by $R$ on $\Lambda_m$, cf.~\eqref{eq:bdy_states}-\eqref{eq:restriction_factorization}. To verify this when the cut runs through a right boundary dimer, note that any truncation of $B_d^r$ produces on its left a configuration consistent with a VMD root tiling. Namely, either a single void, or a void followed by a right $1$-monomer. Both possibilities produce a VMD state on $\Lambda_m$. It is not possible to cut through the interior of a left boundary dimer since $|\Lambda_m|\geq 5$. 
	
	We are left to consider case (iii). In this situation, there is a squeezed Tao-Thouless state $\vp_n^{(j_0)}$ in the fragmentation~\eqref{eq:genVMDform} which covers $x :=\max{\Lambda_m}$ and $x+1=\min{\Lambda_r}$, and we can write
	\begin{equation}\label{eq:tt_factor}
	\psi_{\Lambda}(R) = \psi_{\Lambda'}(R')\otimes \vp_{n}^{(j_0)}\otimes\psi_{\Lambda''}(R'')
	\end{equation}
	where $R'$ and $R''$ are the restrictions of $R$ to appropriately defined intervals $\Lambda'$ and $\Lambda''$. 
	
	We first consider the case that $x$ (and thus, $x+1$) is supported on the last monomer of $\vp_n^{(j_0)}$. Then the first site ($x+1$) of $\Lambda_r \neq \emptyset $ has particle content zero and we may use~\eqref{eq:factorize} to factor off the zeros from  $\vp_n^{(j_0)}$ that are supported on~$\Lambda_r$ . As a result, we once again write $\psi_{\Lambda}(R) = \psi_{\Lambda_m}(R_m)\otimes\xi_{\Lambda_r}$ with $ R_m \in \mathcal{R}_{\Lambda_m} $ the induced root tiling and $\xi_{\Lambda_r} = \ket{0}^{\otimes j} \otimes\psi_{\Lambda''}(R'')$ for some $j<j_0$.
	
	If $x$ is not supported on the last monomer, then $n\geq 2$. Let  $x$ be supported on the $l$th monomer in the root tiling of $\vp_n^{(j_0)}$ and set $r:=n-l\geq 1$. Then the recursion relation~\eqref{eq:recrel2} or~\eqref{eq:recrel3} applies, and
	\begin{equation}\label{eq:recrel4}
	\varphi_{n}^{(j_0)} = \begin{cases} \varphi_{l} \otimes \varphi_{r}^{(j_0)} + \lambda  \varphi_{l-1} \otimes  |  \pmb{\sigma}_d \rangle \otimes  \varphi_{r-1}^{(j_0)} , & \mbox{if $ r > 1 $,} \\
	\varphi_{l} \otimes \varphi_{1}^{(j_0)} + \lambda  \varphi_{l-1} \otimes  |  \pmb{\sigma}_d^{(j_0)} \rangle , & \mbox{if $ r =1 $.}
	\end{cases}  
	\end{equation} 
	If $ r > 1 $, we set $ j \in\{1,2,3\}$  the position on the $l$th monomer which supports $x $, and split  $  |  \pmb{\sigma}_d \rangle =  |  \pmb{\sigma}_{d,1}^{(j)} \rangle \otimes | \pmb{\sigma}_{d,2}^{(j)} \rangle $ 
	with 
	\begin{equation}\label{eq:splitconf} 
	\pmb{\sigma}_{d,1}^{(j)} = \begin{cases} 0 & j= 1 \\ 01 & j = 2 \\ 011 & j = 3  \end{cases} \quad \qquad  \pmb{\sigma}_{d,2}^{(j)} = \begin{cases} 11000 & j= 1 \\ 1000 & j = 2 \\ 000 & j = 3 . \end{cases} 
	\end{equation}
	Note that $  \pmb{\sigma}_{d,1}^{(j)} $ is a possible right boundary tiling configuration. It corresponds to a void if $ j = 1 $, a void followed by a right $ 1 $-monomer if $ j = 2 $, and a right dimer if $ j = 3 $. Substituting the factorization of $|  \pmb{\sigma}_d \rangle$ as well as $\vp_{n}=\vp_{n}^{(j)}\otimes \ket{0}^{\otimes 3-j}$ (see~\eqref{eq:factorize}) produces a refined factorization of the recursion relation from \eqref{eq:recrel4}. Inserting the refined form into \eqref{eq:tt_factor} produces a linear combination of the form
	\[
	\psi_{\Lambda}(R) = \psi_{\Lambda_m}(R_m^1)\otimes  \xi_{\Lambda_r}^1 + \lambda \psi_{\Lambda_m}(R_m^2)\otimes \xi_{\Lambda_r}^2\in \cG_{\Lambda_m}\otimes \cH_{\Lambda_r},
	\]
	with induced root tilings $R_m^1$, $R_m^2\in \cR_{\Lambda_m}$. Since  $  \pmb{\sigma}_{d,2}^{(j)} $ represents a valid tiling configuration at the left boundary, 
	we even have $ \xi_{\Lambda_r}^k= \psi_{\Lambda_r}(R_r^k)  $ for both $ k = 1,2 $ with for appropriately defined $R_r^1$, $R_r^2\in \cR_{\Lambda_r}$ in case $ \Lambda_r $ is sufficiently large. 
	This completes the proof if $r>1$.
	
	If $ r = 1 $, we proceed similarly, the only difference being the factorization $ |  \pmb{\sigma}_d^{(j_0)} \rangle  = |  \pmb{\sigma}_{d,1}^{(j)}\rangle \otimes |  \pmb{\sigma}_{d,2}^{(j,j_0)} \rangle $ with $  \pmb{\sigma}_{d,1}^{(j)}  $ from~\eqref{eq:splitconf}  and $  \pmb{\sigma}_{d,2}^{(j,j_0)} $ appropriately defined. This concludes the proof in case $ \Lambda_l = \emptyset $.
	
	Now suppose that $\Lambda_r = \emptyset$, and wlog $ \Lambda_l \neq \emptyset $, i.e. $\Lambda = \Lambda_l\cup \Lambda_m$. The three cases (i)-(iii) still hold for such $\Lambda$. However, the proof has the following modifications:
	\begin{itemize}
		\item In case (ii): Cutting through a left boundary dimer always produces a configuration on the right that is consistent with a VMD root tiling. Namely, one obtains either a string of voids, or a monomer followed by a void. It is not possible to cut through a right boundary dimer since $|\Lambda_m|\geq 5$.
		\item In case (iii): The proof runs the same with $x = \max{\Lambda_l}$ and $x+1=\min{\Lambda_m}$.  In case of two consecutive monomers we use that  $  \pmb{\sigma}_{d,2}^{(j)} $ from~\eqref{eq:splitconf}  represents a valid tiling configuration at the left boundary. The proof for $r=1$ is simplified since $j_0=3$ due to $|\Lambda_m|\geq 5$. 
	\end{itemize}

	2.~The nesting guarantees that $ \mathcal{G}_{\Lambda }  $ is contained in the right side. It therefore remains to show that any $ \psi $ in the right side is in the VMD subspace $ \mathcal{G}_{\Lambda }  $. The latter is shown using Lemma~\ref{lem:char}. We first show that $ \psi $ is supported on particle configurations $ \pmb{\sigma} \in \ran \pmb{\sigma}_\Lambda $. By assumption $ \psi(\pmb{\sigma}) \neq 0 $ only if 
	\begin{enumerate}[(i)]
		\itemsep0pt
		\item the restriction of $ \pmb{\sigma} $ to $ \Lambda_{1}   $ coincides with the particle configuration $  \pmb{\sigma}_{\Lambda_{1} }(\pmb{D}_{1}) $ of some VMD tiling $ \pmb{D}_{1} \in \cD_{\Lambda_1} $, and 
		\item the restriction of $ \pmb{\sigma} $ to $ \Lambda_{2}  $ coincides with the particle configuration $  \pmb{\sigma}_{\Lambda_{2} }(\pmb{D}_{2}) $ of some VMD tiling $ \pmb{D}_{2} \in \cD_{\Lambda_2} $. 
	\end{enumerate}
	Since  $ |\Lambda_m| \geq 6 $, Lemma~\ref{lem:intersecD} guarantees that there is a unique VMD tiling $ \pmb{D} $ on $ \Lambda $ for which $  \pmb{\sigma} =  \pmb{\sigma}_{\Lambda}(\pmb{D}) $. This proves the first property in  Lemma~\ref{lem:char}.
	For a proof of the second property,  we note that any two consecutive monomers $ D_j , D_{j+1} $ in  $ \pmb{D} $ will either lie entirely in $ \Lambda_{1} $ or  $ \Lambda_{2} $. In either case, \eqref{eq:weightchar} follows from the VMD-property of $ \psi $ on the respective segment.
\end{proof}

We now characterize the ground state space $H_\Lambda$ for all intervals $|\Lambda|\geq 5$.

\begin{theorem}[Ground State Space]\label{thm:gss}
	The VMD space is the ground state space of the Hamiltonian $H_\Lambda$ for all $\lambda\neq 0$ and all finite intervals $|\Lambda|\geq 8$ or $|\Lambda|=5,6$. That is,
	\begin{equation}
	\mathcal{G}_\Lambda = \ker H_\Lambda.
	\end{equation}
	For $|\Lambda|=7$, $\ker H_{\Lambda}= \cG_{\Lambda}\oplus \spa\{\ket{1100011}\}$.
\end{theorem} 
\begin{proof}
	%
	
	The statement is proved by induction on the size of $ \Lambda =[1,n] $. Using \eqref{eq:VMD_cardinality} to calculate $ \dim \cG_\Lambda = |\cR_\Lambda|$ for $n=5,6,7,8$, produces
	\[
	\dim \cG_{[1,5]} = 11, \quad  \dim \cG_{[1,6]} = 17, \quad  \dim \cG_{[1,7]} = 25, \quad  \dim \cG_{[1,8]} =37. \]
	It can be verified numerically (or explicitly by a tedious calculation) that $\dim \cG_\Lambda = \dim \ker(H_\Lambda)$ for $n=5,6,8$. This implies equality since $\cG_{\Lambda}\subset \ker(H_\Lambda)$ by Lemma~\ref{lem:gsorth}(i). For $n=7$, numerical results show $\dim \ker(H_\Lambda) = \dim \cG_\Lambda+1$. One can check that $\ket{1100011}\in\cG_\Lambda^\perp$ is a ground state of $H_{\Lambda}$. Given Lemma~\ref{lem:gsorth}(i), this implies
	\[
	\ker(H_\Lambda) = \cG_\Lambda \oplus \spa \{\ket{1100011} \}
	\]
	as claimed.
	
	Now assume $n>8$. Since the Hamiltonian $H_{[1,n]}$ is frustration-free, 
	\begin{align}
	\ker H_{[1,n]}  & = (\ker H_{[1,n-1]}\otimes \bC^2) \cap (\bC^2 \otimes \ker H_{[2,n]}) \nonumber \\
	& =  (\mathcal{G}_{[1,n-1]} \otimes \bC^2) \cap (\bC^2 \otimes \cG_{[2,n]})   \notag \\
	& =  \mathcal{G}_{[1,n]} ,
	\end{align}
	where we have used the inductive hypothesis in the second equality, and applied Lemma~\ref{lem:intersection} in the last equality. This completes the proof.
\end{proof}

Recall from~\eqref{eq:gap2} that for fixed particle number $ N \in \mathbb{N}_0 $, the ground state energy
of $ H_\Lambda $ is given by 
\begin{equation}\label{eq:gsen}
E_\Lambda(N) = \inf \left\{ \langle \psi , H_\Lambda \psi \rangle \, \Big| \, \psi \in \cH_{\Lambda}\, \wedge \,   \| \psi \| = 1 \, \wedge \,  N_\Lambda \psi = N \psi  \right\} \, . 
\end{equation}
Theorem~\ref{thm:gsenergy} asserts a threshold for the filling fraction $ N/|\Lambda| $ for which $ E_\Lambda(N) $ ceases to be zero. We are now ready to spell out the proof of this result (in the spin language) using Theorem~\ref{thm:gss}.
\begin{proof}[Proof of Theorem~\ref{thm:gsenergy}]
	Since $H_\Lambda$ is non-negative, we trivially have $E_\Lambda(N)\geq 0$. Recall from~\eqref{eq:maxN} that $ N_\Lambda^\textrm{max} $ is the maximal number of particles in any root tiling $R\in\cR_\Lambda $. If $ N \leq N_\Lambda^\textrm{max} $, then according to Lemma~\ref{lem:maxfill}  there is a root tiling $R \in \mathcal{R}_\Lambda $ with $ N = N_\Lambda(R) $ and hence
	$$ E_\Lambda(N) \leq \langle \psi_\Lambda(R) , H_\Lambda \psi_\Lambda(R) \rangle / \| \psi_\Lambda(R) \|^2 = 0 .
	$$ 
	If $ N > N_\Lambda^\textrm{max} $, then any state $\psi$ as in~\eqref{eq:gsen} is orthogonal to the VMD space $\cG_\Lambda$. Thus, from Theorem~\ref{thm:gss}, which applies to intervals with $ |\Lambda|\geq 8 $, we conclude $ E_\Lambda(N) \geq \gap(H_\Lambda) >0 $. 
	Setting $ N_\Lambda^m := N_\Lambda^\textrm{max} $ we have thus established~\eqref{eq:gsenergy}. 
	The bound~\eqref{eq:maxNumber} has been established in~\eqref{eq:maxfill}. 
%
\end{proof}

\section{Proof of the spectral gap}\label{sec:MM}

In this section we prove the spectral gap results for the FQH system stated in Theorem~\ref{thm:gap} and Theorem~\ref{thm:periodic_gap}. We begin by reviewing the  martingale method~\cite{nachtergaele:1996, nachtergaele:2016b} for producing lower bounds on the spectral gap of a quantum spin Hamiltonian in Section~\ref{subsec:martingale}. In Section~\ref{subsec:obc_gap}, we show how this method may be applied to the FQH system with open boundary conditions to prove a nonzero spectral gap estimate that is uniform in the system size. The key condition for this application is a norm bound of a particular operator defined in terms of ground-state projections. This assumption is particularly non-trivial for the FQH system as the ground-state space grows exponentially in the system size. To deal with the large degeneracy, in Section~\ref{subsec:dim_reduction} we identify a subspace defined in terms of tiling configurations that maximizes the norm, and then use this to bound the norm in Section~\ref{subsec:epsilon_calc}. Using a minor generalization of the finite-size criteria proved by Knabe~\cite{knabe:1988} for arbitrary finite-range interactions, we extend our open boundary condition result to prove a uniform lower bound for the spectral gap of the FQH system with periodic boundary conditions in Section~\ref{subsec:periodic_gap} . 

\subsection{The martingale method}\label{subsec:martingale}
The martingale method can be used to estimate the spectral gap above the ground state of a frustration-free Hamiltonian on a finite-dimensional Hilbert space $\cH$. It assumes a sequence $\{h_n: 1\leq n \leq N \}$ of non-negative operators on $\cH$ with $ N \in \mathbb{N} $ fixed. The latter gives rise to
\begin{equation}\label{ham_sequence}
H_n : = \sum_{m=1}^n h_m \;\; \text{for all} \;\; 1\leq n \leq N ,
\end{equation}
an increasing sequence of self-adjoint operators $0 \leq H_1 \leq H_2 \leq \ldots \leq H_{N}$ for which it is assumed that $\ker(H_{N}) \neq \{0\},$ i.e. the ground-state energy of $H_{N}$ is zero. 
Under certain assumptions, the method produces a non-zero lower bound on the spectral gap of $H_{N}$. To state these assumptions, let $G_n$ be the orthogonal projection onto $ \ker(H_n)$. By construction $ \ker(H_{n+1}) \subseteq  \ker(H_n)$, and so $G_n \neq 0$ for all $n$. Additionally, define the decomposition of unity
\begin{equation}\label{En}
E_n := \begin{cases}
\1 - G_1 & n = 0 \\
G_n - G_{n+1} & 1\leq n \leq N-1 \\
G_{N} & n=N
\end{cases}
\end{equation}
and let $g_n$ be the orthogonal projection onto $\ker(h_n)$. 

\begin{assumption}[Conditions for the Martingale Method]\label{assump:MM} \phantom{\hfill}
	\begin{enumerate}
		\item There exists $\gamma>0$ so that $h_n\geq \gamma(\1-g_n)$ for all $1\leq n \leq N$.
		\item There exists $\ell>0$ so that for all $1\leq n \leq N$, $[g_n,E_m]\neq 0$ implies $m\in[n-\ell,n-1]$.
		\item There exists $\epsilon<1/\sqrt{\ell}$ so that $\|g_{n+1}E_n\|\leq \epsilon$ for all $1\leq n \leq N-1$.
	\end{enumerate}
\end{assumption}
Note that $ g_1 E_0 = g_1 (\1-g_1) = 0 $ so that Assumption~3 trivially holds in case $ n = 0 $. Given these assumptions, the following bound on the spectral gap of $H_{N}$ can be concluded:

\begin{theorem}[Martingale Method]\label{thm:MM} Assume Assumption~\ref{assump:MM} holds. Then, for all $\psi \in \cG_{N}^\perp$,
	\[
	\braket{\psi}{H_{N} \psi} \geq \gamma(1-\epsilon\sqrt{\ell})^2\|\psi\|^2.
	\] 
\end{theorem}
This result is the modified version of the martingale method proved in \cite{nachtergaele:2016b}. It is most effective for establishing the spectral gap of frustration-free quantum spin systems with open boundary conditions. The assumptions for this method can be adjusted so that the result can also be applied to systems with periodic conditions \cite{young:2016}. However, we choose to use a finite-size criterion for the periodic result and 
hence do not state the generalized form here. We now focus on applying the martingale method to the FQH spin model with open boundary conditions, and discuss the finite-size criterion and periodic boundary result in Section~\ref{subsec:periodic_gap}.

\subsection{Bounding the spectral gap for open boundary conditions}\label{subsec:obc_gap}
We produce a lower bound on the spectral gap of the FQH spin model on a finite interval $\Lambda=[1,L]$ of length $L\geq 8$ with open boundary conditions, i.e.,  $H_\Lambda$ as defined in \eqref{def:Hspin} on the tensor product $ \cH_\Lambda = \bigotimes_{x=1}^L \bC^2 $. By Theorem~\ref{thm:gss} the set of VMD states is an orthogonal basis for the ground-state space for all such $H_\Lambda$. This will play a key role in our application of the martingale method.

To obtain a bound on the spectral gap of $H_\Lambda$  for the interval $\Lambda$ at some fixed $L\geq 8$, we let $N\geq 2$ and $k\in\{2,3,4\}$  be the unique integers so that $L=3N+k$, and  define a sequence of increasing and absorbing finite intervals via
\begin{equation}\label{mm_vols}
\Lambda_n = [1,3n+k], \quad \quad 0 \leq n \leq N.
\end{equation}
We use this sequence to define the operators $h_n$ and $H_n$ in the martingale method. For labeling convenience, we define these operators for $n\geq 2$ and so the indices in~\eqref{En} and Assumption~\ref{assump:MM} need to be shifted by one.  Specifically, for $2\leq n \leq N$ we set
\begin{equation}\label{mm_seq_def}
H_n = \sum_{m=2}^n h_m, \quad \quad
h_n = \begin{cases}
H_{\Lambda_2} & n = 2 \\
H_{\Lambda_n\setminus \Lambda_{n-3}} &n \geq 3
\end{cases}
\end{equation}
Our choice of values for $k$ guarantees that $|\Lambda_n|\geq 8$ for all $n\geq 2$ so that Theorem~\ref{thm:gss} applies and that $\Lambda_n\setminus \Lambda_{n-3}$ is an interval of length nine for each $n\geq 3$.  
By considering how the intervals~\eqref{mm_vols} overlap, one sees that the range of each interaction term ($n_xn_{x+2}$ or $q_{x}^*q_x$) is contained in the support of at least one and at most three of the operators $h_n$. As a consequence 
\begin{equation}\label{eq:equivalence}
H_{\Lambda_n} \leq H_n \leq 3H_{\Lambda_n}
\end{equation}
for all $n\geq 2$. These Hamiltonians therefore have the same kernel, i.e. ground-state space. 
We denote by $G_{\Lambda'}$ the orthogonal projection onto the ground-state space $\ker( H_{\Lambda'}) \otimes \cH_{\Lambda\setminus \Lambda'}\subset \cH_\Lambda$ for any finite volume $\Lambda'\subseteq \Lambda$, 
and then define $G_n$ (resp. $g_n$) as in the last subsection, but with shifted indices:
\begin{equation}\label{gs_proj}
G_n = G_{\Lambda_n},\quad \quad
g_n = \begin{cases}
G_{\Lambda_2} & n = 2\\
G_{\Lambda_{n}\setminus \Lambda_{n-3}} & n \geq 3 .
\end{cases}
\end{equation}

The martingale method will produce a lower bound of $\gap(H_{N})$ that is independent of $N$ and $k$. Since $\Lambda = \Lambda_N$, the second inequality~\eqref{eq:equivalence} implies
\begin{equation}\label{equiv_hams_gap}
\gap(H_{\Lambda}) \geq \frac{1}{3}\gap(H_{N}),
\end{equation}
and so the spectral gap of $H_\Lambda$ will also have a nonzero lower bound independent of $|\Lambda|=L=3N+k$. The lower bound in Theorem~\ref{thm:gap} will then follow by verifying the conditions in Assumption~\ref{assump:MM}. The first two conditions are easy to check for the collection of operators defined above. Proving the third condition of Assumption~\ref{assump:MM} is the content of the following lemma, whose proof is the main focus of this section and can be found in Subsection~\ref{subsec:epsilon_calc} below.

\begin{lem}[Norm bound]\label{lem:epsilon_calc} Suppose that $\Lambda$ is an interval of length $|\Lambda|\geq 11$, and define $\Lambda_1$ to be the first $|\Lambda| -3$ sites of $\Lambda$, and $\Lambda_2$ to be the last nine sites of $\Lambda$. Then for any $\lambda\neq0$,
	\begin{equation}\label{def:f}
	\|G_{\Lambda_2}(\1-G_{\Lambda}) G_{\Lambda_1}\|^2 \leq \sup_{n \geq 4 } f_n(|\lambda|^2) =:  f(|\lambda|^2), 
	\end{equation}
	where in terms of $\alpha_k = \|\varphi_{k-1}\|^2/\|\varphi_{k}\|^2$,
	\begin{equation}\label{f_N}
	f_n(r) := r \alpha_n\alpha_{n-2}\left(\frac{ \left[1-\alpha_{n-1} (1+r)\right]^2 }{ 1+2r} + \alpha_{n-3}  \frac{ r (1-\alpha_{n-1})^2}{1+r} \right) .
	\end{equation}
	In particular, $ f(|\lambda|^2) <1/3$ for $|\lambda|<5.3$ .
\end{lem}


We now provide the conditional proof of Theorem~\ref{thm:gap} given that Lemma~\ref{lem:epsilon_calc} holds.

\begin{proof}[Proof of Theorem~\ref{thm:gap}]
We consider the FQH spin Hamiltonian $H_\Lambda$ on $\Lambda = [1,3N+k]$ at some fixed integers $N\geq 2$, $k\in\{2,3,4\}$ and fixed $\lambda\neq 0$.  
We set $h_n$ and $H_n$ as in \eqref{mm_seq_def} for  $2\leq n \leq N$ and verify the conditions in Assumption~\ref{assump:MM} for this collection of operators.
	
	For Assumption~1, we recall that  $\Lambda_n\setminus \Lambda_{n-3}$ is an interval of length 9 for any $n\geq 3$, while for $n=2$ we have $|\Lambda_2|\in \{8,9,10\}$. Since the FQH model is translation invariant,  we thus conclude for all $2\leq n \leq N$:
	\begin{equation}
	h_n \geq \gamma(\1-g_n)\;\; \text{where}\;\; \gamma = \min_{m=8,9,10}\gap(H_{[1,m]})>0.
	\end{equation} 
	
	To verify Assumption~2 with $\ell=3$,  we show that $[g_n,E_m]=0$ for $m \notin [n-3,n-1]$. To this end, we conclude  from~\eqref{gs_proj} and \eqref{En} (whose shifted index starts at $ n = 1 $) that  $\supp(E_m)\subseteq \Lambda_{m+1}$ for all $ 1 \leq m \leq N-1  $.  In case $n\geq 3$, we also have $\supp(g_n) \subseteq \Lambda_{n}\setminus \Lambda_{n-3}$. Since the supports of these intervals are disjoint for $m<n-3$, we have $[g_n,E_m]=0$.When $m\geq n$, the claim $[g_n,E_m]=0$ follows from the frustration-free property
	\begin{equation}\label{ff_prop}
	G_{\Lambda'}G_{\Lambda''} = G_{\Lambda''} G_{\Lambda'}=G_{\Lambda'} \;\; \text{for all} \;\; \Lambda''\subseteq \Lambda' , 
	\end{equation}
	combined with the definitions of $E_m$, $G_n$ and $g_n$. When $n=2$, Assumption~2 with $\ell=3$ trivially follows from~\eqref{ff_prop} since $\Lambda_2\subset \Lambda_{m+1}$ for all $m\geq 1$.
	
	To verify Assumption~3, we fix $2\leq n \leq N-1$ and use \eqref{ff_prop} to write
	\[
	\|g_{n+1}E_n\| = \|G_{\Lambda_{n+1}\setminus\Lambda_{n-2}}(G_{\Lambda_n}-G_{\Lambda_{n+1}})\| = \|G_{\Lambda_{n+1}\setminus\Lambda_{n-2}}(\1-G_{\Lambda_{n+1}})G_{\Lambda_n}\| .
	\]
	Lemma~\ref{lem:epsilon_calc} now applies since $|\Lambda_{n+1}\setminus \Lambda_n| = 3$ and $|\Lambda_{n+1}\setminus\Lambda_{n-2}| = 9$. Hence Assumption~3 holds with $\epsilon^2 = f(|\lambda|^2)$, since 
	Lemma~\ref{lem:epsilon_calc} also guarantees that $f(|\lambda|^2)<1/3$ for all $|\lambda|<5.3$. 
	
	Therefore, combining~\eqref{equiv_hams_gap} and Theorem~\ref{thm:MM} we arrive at
	\[
	\gap(H_\Lambda) \geq \min_{m=8,9,10}\gap(H_{[1,m]})\frac{\left(1-\sqrt{3f(|\lambda|^2)}\right)^2}{3}>0 .
	\]
	This completes the proof of Theorem~\ref{thm:gap}.
\end{proof}

\subsection{A dimensional reduction}\label{subsec:dim_reduction}
In the previous section, we showed how Theorem~\ref{thm:gap} can be concluded from an upper bound on
\begin{equation*}
\|G_{\Lambda_2}(\1-G_{\Lambda}) G_{\Lambda_1}\| = \sup_{\substack{\psi\in \cG_{\Lambda}^\perp\cap (\cG_{\Lambda_1}\otimes\cH_{\Lambda\setminus\Lambda_1})\\ \psi\neq 0}}\frac{\|G_{\Lambda_2}\psi\|}{\|\psi\|}
\end{equation*}
where $\Lambda_1,\Lambda_2\subseteq\Lambda$ are as in Lemma~\ref{lem:epsilon_calc} and $G_{\Lambda'}$ is the orthogonal projection onto the ground state space $\cG_{\Lambda'}\otimes\cH_{\Lambda\setminus \Lambda'}$ for any $\Lambda'\subseteq \Lambda$. Here, we use that $\ker(H_{\Lambda'}) = \cG_{\Lambda'}$ for $|\Lambda'|\geq 8$ by Theorem~\ref{thm:gss}. Since $\cG_{\Lambda_1}\otimes\cH_{\Lambda\setminus\Lambda_1}$ is highly degenerate, producing an upper bound on this norm is rather nontrivial. The main goal of this section is to reduce the complexity by identifying a subspace $\cP_{\Lambda_1}\subset \cG_{\Lambda_1}\otimes\cH_{\Lambda\setminus\Lambda_1}$ and an associated orthogonal projection $P_{\Lambda_1}$ for which
\begin{equation}\label{eq:norm_reduction}
\|G_{\Lambda_2}(\1-G_{\Lambda}) G_{\Lambda_1}\|= \|G_{\Lambda_2}(\1-G_{\Lambda}) P_{\Lambda_1}\|.
\end{equation}
After defining $ \cP_{\Lambda_1} $ we prove some basic properties in Lemma~\ref{lem:P_properties}, which will allow us to establish~\eqref{eq:norm_reduction} in Lemma~\ref{lem:reduction}. We finish the subsection by describing an orthogonal basis for $\cP_{\Lambda_1}$ in Lemma~\ref{lem:orthogonal_basis} which we then use to prove Lemma~\ref{lem:epsilon_calc} in Section~\ref{subsec:epsilon_calc}.\\

Assuming that $\Lambda_1\subset \Lambda$ is a finite interval with $|\Lambda_1|\geq 8$, the set
\begin{equation}\label{eq:orthog_basis}
\caB_{\Lambda_1}:=\left\{\psi_{\Lambda_1}(R_1)\otimes \ket{\pmb{\pmb{\tau}}} \, \big| \, R_1\in\cR_{\Lambda_1} \; \text{and}\; \pmb{\pmb{\tau}}\in\{0,1\}^{\Lambda\setminus\Lambda_1}\right\}
\end{equation}
forms an orthogonal basis of the ground state space $\cG_{\Lambda_1}\otimes\cH_{\Lambda\setminus \Lambda_1}$. Notice that these vectors are supported on tilings of $ \Lambda_1 $. Recall that $C_\Lambda$ is the orthogonal projection onto the space of VMD tilings, 
\begin{equation}\label{def:C_Lambda}
\mathcal{C}_{\Lambda} = \spa\left\{ |  \pmb{\sigma}_\Lambda(\pmb{D}) \rangle \, \big| \,  \pmb{D} \in \mathcal{D}_\Lambda \right\}\subseteq\cH_\Lambda,
\end{equation}
and that $\cG_{\Lambda}\subset\caC_{\Lambda}$ since all VMD states are supported on tiling configurations. With this notation, the subspace $\cP_{\Lambda_1}\subset \cG_{\Lambda_1}\otimes \cH_{\Lambda \setminus \Lambda_1}$ that we show satisfies \eqref{eq:norm_reduction} is given by
\begin{align}
\mathcal{P}_{\Lambda_1} :=
\spa\left\{ \psi_{\Lambda_1}(R_1) \otimes |\pmb{\tau} \rangle \in\caB_{\Lambda_1} \, \big| \mbox{ $ C_\Lambda \psi_{\Lambda_1}(R_1) \otimes |\pmb{\tau} \rangle   \neq 0 $} \right\}. \label{P_def}
\end{align}
As we will see in the following,  $\cP_{\Lambda_1}$ is also a subspace of $ \mathcal{C}_{\Lambda} $.

\begin{lem}[Properties of $\cP_{\Lambda_1}$]\label{lem:P_properties} Given a pair of finite intervals $\Lambda_1\subset \Lambda$ with $|\Lambda_1| \geq 8 $, the subspace $\cP_{\Lambda_1}\subset \cH_\Lambda$ has the following properties:
	\begin{enumerate}
		\item It is equivalently described by
		\begin{equation}\label{eq:P_VMD}
		\cP_{\Lambda_1} = \spa\left\{\psi_{\Lambda_1}(R_1)\otimes \ket{\pmb{\tau}} \in\caB_{\Lambda_1} \, \big| \,
		\left\langle \psi_\Lambda(R) | \psi_{\Lambda_1}(R_1)\otimes|\pmb{\tau}\right\rangle \neq 0 \; \text{for some} \; R\in\cR_\Lambda \right\}.
		\end{equation}
		\item Any  $ \psi_{\Lambda_1}(R_1)\otimes \ket{\pmb{\tau}} \in\caB_{\Lambda_1}  $ is an eigenstate of the projection $ C_\Lambda $ and for any $  \psi_{\Lambda_1}(R_1)\otimes \ket{\pmb{\tau}} \in \cP_{\Lambda_1} $:
		\begin{equation}\label{eq:tiling_projection_property}
			C_\Lambda \psi_{\Lambda_1}(R_1) \otimes |\pmb{\tau} \rangle  =\psi_{\Lambda_1}(R_1) \otimes |\pmb{\tau} \rangle . 
		\end{equation}
		\item The inclusion $\cG_{\Lambda}\subset \cP_{\Lambda_1}\subset \caC_\Lambda$ holds.
		\item  The projection $P_{\Lambda_1}$ onto $\cP_{\Lambda_1}$ factorizes as  $P_{\Lambda_1} = C_\Lambda G_{\Lambda_1} = G_{\Lambda_1}C_\Lambda.$
	\end{enumerate}
\end{lem}

The characterization~\eqref{eq:P_VMD} is motivated by the system's frustration free property in the sense that the spanning set on the right side is the smallest subset $\cB\subset \cB_{\Lambda_1}$ for which $\cG_\Lambda \subset \spa \cB$. 
There are other equivalent expressions for $\cP_{\Lambda_1}$. E.g., the spanning set in \eqref{def:C_Lambda} is an orthonormal basis of $\caC_\Lambda$, and so for any $\psi\in\cH_{\Lambda}$: 
\[
C_{\Lambda}\psi = \sum_{\pmb{D} \in \mathcal{D}_\Lambda} \braket{ \pmb{\sigma}_\Lambda(\pmb{D})}{\psi} \ket{\pmb{\sigma}_\Lambda(\pmb{D})}.
\]
Therefore, $C_{\Lambda}\psi \neq 0$ if and only if $\braket{ \pmb{\sigma}_\Lambda(\pmb{D})}{\psi}\neq 0$ for some $\pmb{D} \in \mathcal{D}_\Lambda$ and hence
\begin{equation}\label{eq:P_config}
\cP_{\Lambda_1} = \spa\left\{\psi_{\Lambda_1}(R_1)\otimes \ket{\pmb{\tau}} \in\caB_{\Lambda_1} \, \big| \,
\left\langle \pmb{\sigma}_\Lambda(\pmb{D}) | \psi_{\Lambda_1}(R_1)\otimes|\pmb{\tau}\right\rangle \neq 0 \; \text{for some} \; \pmb{D}\in\cD_{\Lambda} \right\}.
\end{equation}
We will use \eqref{eq:P_config} rather than \eqref{P_def} to prove the first property in Lemma~\ref{lem:P_properties}. 
Moreover, the second property, which shows that any state $\psi_{\Lambda_1}(R)\otimes \ket{\pmb{\tau}}\in\cB_{\Lambda_1}$ is either in $\caC_\Lambda$ or its orthogonal complement $\caC_{\Lambda}^\perp$, allows us to substitute \eqref{eq:tiling_projection_property} for the condition in~\eqref{P_def}. 
Combining this observation and the factorization property in Lemma~\ref{lem:P_properties}, it follows that the subspace $\cP_{\Lambda_1}$ is the intersection
\begin{align}\label{eq:intersection_form}
\cP_{\Lambda_1} & = \caC_\Lambda\cap (\cG_{\Lambda_1}\otimes \cH_{\Lambda\setminus \Lambda_1})  =  C_{\Lambda}(\cG_{\Lambda_1}\otimes \cH_{\Lambda\setminus \Lambda_1}) = G_{\Lambda_1}\caC_\Lambda.
\end{align}
The nesting of the various subspaces is schematically depicted in~Figure~\ref{fig:boxes}. 

\begin{figure}
	\begin{center}
		\begin{tikzpicture}
		\draw[color=blue, fill = blue!10, thick] (0,0) rectangle (2.5,-4) node[align=left, above] {$\color{blue}{\caG_{\Lambda_1}\otimes \cH_{\Lambda\setminus\Lambda_1}\quad\quad\quad\quad\quad\quad}$};
		\draw[color=red, fill = red!10, thick] (0,0) rectangle (4,-2.5) node[align=left, above] {$\color{red}{\caC_{\Lambda}\quad\quad}$};
		\draw[color=Plum, fill = Plum!10, thick] (0,0) rectangle (2.5,-2.5) node[align=left, above] {$\color{Plum}{\cP_{\Lambda_1}\quad\quad}$};
		\draw[color=black, fill=white, thick](0,0) rectangle (1.5,-1.5) node[align=left, above] {$\caG_{\Lambda}\quad\quad$};
		\end{tikzpicture}
	\end{center}	
	\caption{Sketch of the nesting of subspaces introduced in the dimensional reduction. The inclusion $\caG_{\Lambda}\subseteq \caG_{\Lambda_1}\otimes \cH_{\Lambda\setminus\Lambda_1} $  reflects the frustration-freeness of the ground state space in case $ |\Lambda_1| \geq 8 $. The nesting $\caG_{\Lambda}\subseteq \caC_{\Lambda} $ is a consequence of the VMD-state construction, which are supported on tilings. That  $ \cP_{\Lambda_1} = \caC_{\Lambda} \cap (\caG_{\Lambda_1}\otimes \cH_{\Lambda\setminus\Lambda_1})  $ is proven in Lemma~\ref{lem:P_properties}.}\label{fig:boxes}
\end{figure}
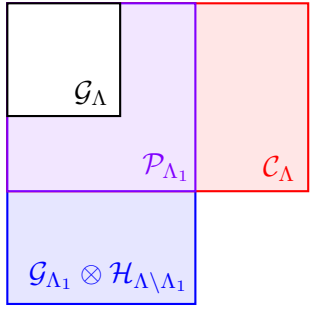

\begin{proof}[Proof of Lemma~\ref{lem:P_properties}] 1.~We show that the set conditions in \eqref{eq:P_VMD} and \eqref{eq:P_config} are equivalent. Fix $\psi_{\Lambda_1}(R_1)\otimes\ket{\pmb{\tau}}\in\cB_{\Lambda_1}$. If there exists $R\in \cR_\Lambda$ so that $\left\langle \psi_\Lambda(R) | \psi_{\Lambda_1}(R_1)\otimes|\pmb{\tau}\right\rangle \neq 0,$ then from~\eqref{eq:VMDocc} we trivially conclude that there is some $\pmb{D}\in \cD_\Lambda(R)\subseteq\cD_{\Lambda}$ such that 
	\[\left\langle \pmb{\sigma}_\Lambda(\pmb{D})  | \psi_{\Lambda_1}(R_1)\otimes|\pmb{\tau}\right\rangle \neq 0.\]
Conversely, suppose there is a tiling $\pmb{D}\in \cD_\Lambda$ such that
	\begin{equation*}
	0\neq 
	\left\langle \pmb{\sigma}_\Lambda(\pmb{D})  | \psi_{\Lambda_1}(R_1)\otimes|\pmb{\tau}\right\rangle 
	= \sum_{\pmb{D}_1\in\cD_{\Lambda_1}(R_1)} \lambda^{\#(\pmb{D}_1)} \braket{ \pmb{\sigma}_\Lambda(\pmb{D}) }{\pmb{\sigma}_{\Lambda_1}(\pmb{D}_1)\,\pmb{\tau}},
	\end{equation*}
	where we have expanded using \eqref{eq:VMDocc}. Since $\pmb{\sigma}_{\Lambda_1}$ is injective, there is a unique $\pmb{D}_1\in\cD_{\Lambda_1}(R_1)$ such that $\pmb{\sigma}_\Lambda(\pmb{D})= (\pmb{\sigma}_{\Lambda_1}(\pmb{D}_1), \, \pmb{\tau})$. Hence for any $\pmb{D}_1'\in\cD_{\Lambda_1}$ and $\pmb{\sigma}\in\{0,1\}^{\Lambda\setminus\Lambda_1}$:
	\begin{equation}\label{cond_1}
	\braket{\pmb{\sigma}_\Lambda(\pmb{D})}{\pmb{\sigma}_{\Lambda_1}(\pmb{D}_1')\,\pmb{\sigma}} = \braket{\pmb{\sigma}_{\Lambda_1}(\pmb{D}_1)\,\pmb{\tau}}{\pmb{\sigma}_{\Lambda_1}(\pmb{D}_1')\,\pmb{\sigma}} = \delta_{\pmb{D}_1,\pmb{D}_1'}\delta_{\pmb{\tau},\pmb{\sigma}}.
	\end{equation}
	Since $\cD_\Lambda$ is the disjoint union of all $\cD_\Lambda(R)$, there is a unique root tiling $R\in \cR_\Lambda$ such that
	\begin{equation}\label{cond_2}
	\left\langle \pmb{\sigma}_\Lambda(\pmb{D})  | \psi_{\Lambda}(R)\right\rangle \neq 0.
	\end{equation}
	We show that $\left\langle \psi_\Lambda(R) | \psi_{\Lambda_1}(R_1)\otimes|\pmb{\tau}\right\rangle\neq 0$ for this root tiling $R$. By the nesting $\cG_\Lambda \subseteq \cG_{\Lambda_{1}}\otimes \cH_{\Lambda\setminus\Lambda_1}$ we may expand $\psi_\Lambda(R)$ using~\eqref{eq:VMDocc} as
	\begin{align}
	\psi_\Lambda(R) 
	& = 
	\sum_{\substack{R' \in\cR_{\Lambda_1} \\ \pmb{\sigma}\in\{0,1\}^{\Lambda\setminus\Lambda_1}}}
	c_{R',\sigma}^R\psi_{\Lambda_1}(R')\otimes\ket{\pmb{\sigma}} \notag \\
	& = 
	\sum_{\substack{R' \in\cR_{\Lambda_1}\\ \pmb{\sigma}\in\{0,1\}^{\Lambda\setminus\Lambda_1}}} \sum_{\pmb{D}_1'\in\cD_{\Lambda_1}(R')}
	c_{R',\pmb{\sigma}}^R \, \lambda^{\# (\pmb{D}_1')}\ket{\pmb{\sigma}_{\Lambda_1}(\pmb{D}_1') \, \pmb{\sigma}} \label{VMD_config_decomp}
	\end{align}
	where  $c_{R',\pmb{\sigma}}^R := \left(  \langle \psi_{\Lambda_1}(R') \otimes \langle \pmb{\sigma} | \right) \big|    \psi_\Lambda(R) \rangle / \|  \psi_{\Lambda_1}(R') \|^2 $.
	In combination with~\eqref{cond_1} and \eqref{cond_2} this implies that $0\neq \left\langle \pmb{\sigma}_\Lambda(\pmb{D})  | \psi_{\Lambda}(R)\right\rangle = c_{R_1,\pmb{\tau}}^R\, \lambda^{\#(\pmb{D}_1)}$. In particular $c_{R_1,\pmb{\tau}}^R\neq 0$, and therefore 
	\[
	\left\langle \psi_\Lambda(R) | \psi_{\Lambda_1}(R_1)\otimes|\pmb{\tau}\right\rangle = \overline{c_{R_1,\pmb{\tau}}^R} \,  \|\psi_{\Lambda_1}(R_1)\|^2 \neq 0
	\]
	as desired. This establishes~\eqref{eq:P_VMD}.\\

	2.~Either $ \psi_{\Lambda_1}(R_1)\otimes \ket{\pmb{\tau}}\in\cB_{\Lambda_1} $ is in the kernel of $ C_\Lambda $, or $  C_\Lambda \psi_{\Lambda_1}(R_1)\otimes \ket{\pmb{\tau}} \neq 0 $, and the vector is in~$ \mathcal{P}_{\Lambda_1} $. In the latter case, using the previous argument we conclude that there is some root tiling $R\in \cR_{\Lambda}$ for which $c_{R_1,\pmb{\tau}}^R\neq 0$. 
	Recall that $\psi_\Lambda(R)$ is only supported on configurations corresponding to the VMD tilings $\cD_{\Lambda}(R)$, and note that \eqref{VMD_config_decomp} is a decomposition of $\psi_{\Lambda}(R)$ in terms of configurations. Since $c_{R_1,\pmb{\tau}}^R\neq 0$, \eqref{VMD_config_decomp} implies the vector $  \psi_{\Lambda_1}(R_1)\otimes \ket{\pmb{\tau}} $ is supported on $\cD_{\Lambda}(R)$. In particular,
	\[\left\{\ket{\pmb{\sigma}_{\Lambda_1}(\pmb{D}_1)\,\pmb{\tau}} \big| \,  \pmb{D}_1 \in \mathcal{D}_{\Lambda_1}(R_1)\right\} \subseteq \left\{\ket{\pmb{\sigma}_{\Lambda}(\pmb{D})} \big| \,  \pmb{D} \in \mathcal{D}_{\Lambda}\right\}.\] 
	Therefore, $C_{\Lambda}\psi_{\Lambda_1}(R_1)\otimes \ket{\pmb{\tau}} = \psi_{\Lambda_1}(R_1)\otimes \ket{\pmb{\tau}}$, which proves \eqref{eq:tiling_projection_property}.\\

	3.~The first inclusion follows from Property~1 and the nesting $\caG_\Lambda \subseteq \caG_{\Lambda_1}\otimes\cH_{\Lambda\setminus\Lambda_1}$ . The second inclusion is an immediate consequence of Property~2.\\

	4.~This is again an immediate consequence of \eqref{eq:tiling_projection_property}.
	\end{proof} 

We are now ready to prove  the reduction property~\eqref{eq:norm_reduction}.
\begin{lem}[Reduction]\label{lem:reduction}
	Let $ \Lambda $ be a finite interval with two subintervals $ \Lambda_1, \Lambda_2\subseteq \Lambda $ of length $ |\Lambda_1|,  |\Lambda_2|\geq 8 $, such that $\Lambda = \Lambda_1\cup\Lambda_2$ and $ | \Lambda_1\cap \Lambda_2 | \geq 6 $. Then
	\begin{equation}\label{eq:reduction}
	G_{\Lambda_2}(\1-G_{\Lambda})  G_{\Lambda_1} = 
	G_{\Lambda_2}(\1-G_{\Lambda}) P_{\Lambda_1}=
	P_{\Lambda_2}(\1-G_{\Lambda}) P_{\Lambda_1}. 
	\end{equation}
\end{lem}
\begin{proof}
 To ease notation, we set $P^\perp := \1-P$ for any orthogonal projection $P$. We recall that for nested subspaces $\cQ \subseteq \cP$ of a Hilbert space $\cH$, the associated orthogonal projections satisfy $Q = PQ = QP.$ Since $\cG_{\Lambda_j}\subseteq\caC_{\Lambda_j}$ for both $ j = 1,2 $, we conclude using $C_\Lambda = C_{\Lambda_2} C_{\Lambda_1}$ from Corollary~\ref{cor:config_projections} that
	\begin{align}
	G_{\Lambda_2}G_{\Lambda_1} & = G_{\Lambda_2}C_{\Lambda_2}C_{\Lambda_1}G_{\Lambda_1} = G_{\Lambda_2}C_\Lambda G_{\Lambda_1} \label{eq:g1g2equiv} \\
	G_\Lambda^\perp & = G_\Lambda^\perp(C_\Lambda^\perp + C_\Lambda) = C_\Lambda^\perp + G_\Lambda^\perp C_\Lambda . \label{eq:gperp_equiv}
	\end{align}
	The first identity implies $ G_{\Lambda_2} C_\Lambda^\perp G_{\Lambda_1}=0$, which combined with \eqref{eq:gperp_equiv} gives the first equality in \eqref{eq:reduction}:
	\begin{align*}
	G_{\Lambda_2} G_\Lambda^\perp  G_{\Lambda_1} =  G_{\Lambda_2}G_\Lambda^\perp C_\Lambda G_{\Lambda_1}
	= G_{\Lambda_2}G_\Lambda^\perp P_{\Lambda_1},
	\end{align*}
	where in the last step we use Lemma~\ref{lem:P_properties}(3). By a similar argument, the second equality in \eqref{eq:reduction} is obtained from
	\begin{align*}
	G_{\Lambda_2}G_\Lambda^\perp P_{\Lambda_1}
	& =
	G_{\Lambda_2} (C_\Lambda+C_\Lambda^\perp)G_\Lambda^\perp P_{\Lambda_1} 
	= P_{\Lambda_2}G_\Lambda^\perp P_{\Lambda_1} + G_{\Lambda_2} C_\Lambda^\perp P_{\Lambda_1}
	=  P_{\Lambda_2}G_\Lambda^\perp P_{\Lambda_1}
	\end{align*}
	where $C_\Lambda^\perp P_{\Lambda_1}=0$ holds by Lemma~\ref{lem:P_properties}(3).
\end{proof}

We close this section by determining an orthogonal basis of $\cP_{\Lambda_1}$. To described this basis, it is helpful to introduce a subset of restricted root-tilings. Suppose that $\Lambda'=[1,N] \subset \Lambda=[1,L]$ are two intervals with $L>N$, and let $n\geq 1$ and $j\in\{1,2,3\}$ be the unique integers so that $L- N =3(n-1)+j$. We are interested in root tilings $R'\in\cR_{\Lambda'}$ for which there exists a $R\in \cR_{\Lambda}$ such that
\begin{equation}\label{eq:vmd_factor}
\psi_{\Lambda}(R) = \psi_{\Lambda'}(R')\otimes \vp_n^{(j)}.
\end{equation}
Said differently, the root tilings $R'$ of $\Lambda'$ can be extended by monomers to a root-tiling $R\in\cR_{\Lambda}$ for which the resulting VMD-state $\psi_\Lambda(R)$ factors across $\Lambda'$. We denote the set of all such root tilings on $\Lambda'$ by
\begin{equation}\label{eq:factorable_tilings}
\cR_{\Lambda'}^{\mathrm{f}} := \left\{ R'\in\cR_{\Lambda'} \, \big| \, \exists \, R\in \cR_{\Lambda} \; \text{s.t.} \; \psi_{\Lambda}(R) = \psi_{\Lambda'}(R')\otimes \vp_n^{(j)}\right\}.
\end{equation}
In view of Theorem~\ref{thm:product1}, if $|\Lambda'| \neq 5$ this is simply the set of all root tilings $R'=(B',V',M')\in\cR_{\Lambda'}$ so that $\max{\Lambda'} \in V'$.\footnote{We do not consider truncated versions of the left boundary dimer for $|\Lambda'|\leq 4$; see the footnote in Section~\ref{subsec:tilings}.} When $|\Lambda'|= 5$, $\cR_{\Lambda'}^{\mathrm{f}}$ also contains the root tiling composed of just the left-boundary domino, $B_d^l$. Note that while the values of $n$ and $j$ in \eqref{eq:factorable_tilings} depend on both $\Lambda$ and $\Lambda'$, the set of tilings $\cR_{\Lambda'}^{\mathrm{f}}$ is the same regardless of the choice $\Lambda\supseteq \Lambda'$ and so there is no ambiguity in the notation.

\begin{lem}[Orthogonal basis for $\cP_{\Lambda_1}$]\label{lem:orthogonal_basis} Let $\Lambda_1 \subseteq \Lambda$ be a pair of finite intervals such that $|\Lambda_1|\geq 8$ and $\Lambda\setminus \Lambda_1$ consists of the last three sites of $\Lambda$. Then an orthogonal basis  $ \cB_{\Lambda_1}'\subseteq \cB_{\Lambda_1}  $  for $\cP_{\Lambda_1}$ consists of the following: 
	\begin{enumerate}
		\item all VMD-states indexed by $R\in \cR_\Lambda$ such that  $ \psi_\Lambda(R) \in   \cB_{\Lambda_1} $.
		\item  the set of vectors of the form
		\begin{equation}\label{eq:P_vectors}
		\psi_{\Lambda'}(R')\otimes\vp_{n-1}\otimes\vp_1^{(j)}  \quad \mbox{or} \quad 
		\psi_{\Lambda'}(R')\otimes\vp_{n-2}\otimes \ket{\pmb{\sigma}_d^{(j)}} 
		\end{equation}
		which are indexed by $n\geq 2 $, $j\in \{1,2,3\}$ and $ \Lambda' \subset \Lambda $ such that $|\Lambda\setminus \Lambda'|=3(n-1)+j$  and $ R'\in\cR_{\Lambda'}^{\mathrm{f}} $. 
Here we use the convention that $\psi_{\Lambda'}(R')=1$ if $\Lambda'=\emptyset$.
	\end{enumerate}
\end{lem}

To quickly check that the vectors from \eqref{eq:P_vectors}  belong to $\cB_{\Lambda_1}$, one may use~\eqref{eq:factorize} to rewrite
\begin{equation}\label{eq:factor1}
\psi_{\Lambda'}(R')\otimes\vp_{n-1}\otimes\vp_1^{(j)} = \left(\psi_{\Lambda'}(R')\otimes\vp_{n-1}^{(j)}\right)\otimes\ket{\pmb{\tau}} =: \psi_{\Lambda_1}(R_1)\otimes\ket{\pmb{\tau}}
\end{equation}
with  $\ket{\pmb{\tau}}:=\ket{0}^{\otimes 3-j}\otimes\vp_1^{(j)}$. Furthermore, for each $j$ there is a $\pmb{\tau_j}\in\{0,1\}^3$ so that $\ket{\pmb{\sigma}_d^{(j)}}=\ket{\pmb{\sigma_j}}\otimes\ket{\pmb{\tau_j}}$ where
$
\ket{\pmb{\sigma_1}} = \ket{0},$ $\ket{\pmb{\sigma_2}}=\ket{01},$ and $\ket{\pmb{\sigma_3}} = \ket{011}.
$
In terms of tilings, each $\pmb{\sigma_j}$ corresponds to placing a void possibly followed by a boundary monomer, or a right boundary dimer. In all cases, there is some root tiling $R_1\in\cR_{\Lambda}$ for which
\begin{equation}\label{eq:factor_2}
\psi_{\Lambda'}(R')\otimes\vp_{n-2}\otimes \ket{\pmb{\sigma}_d^{(j)}} = \psi_{\Lambda_1}(R_1)\otimes\ket{\pmb{\tau_j}}.
\end{equation}
We now prove Lemma~\ref{lem:orthogonal_basis}.

\begin{proof}[Proof of Lemma~\ref{lem:orthogonal_basis}] Recall from~\eqref{eq:orthog_basis} that $\cB_{\Lambda_1}$ is an orthogonal basis of $\cG_{\Lambda_1}\otimes\cH_{\Lambda\setminus\Lambda_1}$. By Lemma~\ref{lem:P_properties} the set
	\begin{equation}\label{P_cond}
	\cB_{\Lambda_1}':=\left\{\psi_{\Lambda_1}(R_1)\otimes\ket{\pmb{\tau}}\in\cB_{\Lambda_1} \, \big| \, \left\langle \psi_\Lambda(R) | \psi_{\Lambda_1}(R_1)\otimes|\pmb{\tau}\right\rangle \neq 0 \;\; \text{for some} \;\; R\in \cR_\Lambda\right\},
	\end{equation}
	then forms an orthogonal basis for $\cP_{\Lambda_1}$.  The claimed orthogonality of these states is immediate since $\cB_{\Lambda_1}'\subset \cB_{\Lambda_1}$, the latter of which is an orthogonal basis of $\cG_{\Lambda_1}\otimes\cH_{\Lambda\setminus\Lambda_1}$. 
	In order to determine an explicit form of $\cB_{\Lambda_1}'$, we write each VMD-state $\psi_\Lambda(R) $ as a linear combination of states $\psi_{\Lambda_1}(R_1)\otimes\ket{\pmb{\tau}}\in\cB_{\Lambda_1}$. Fix a root tiling $R=(B,V,M)\in\cR_{\Lambda}$. There are two possible kinds of decompositions of $\psi_{\Lambda}(R)$, cf.~Theorem~\ref{thm:product1}:\\
	
	1.~For the first type, $R$ ends in a right-boundary dimer or has a void in the last four sites, i.e. $B^r = B^r_d$ or $[\max{\Lambda_1}, \,\max\Lambda]\cap V\neq \emptyset$. In this case, there is a root-tiling $R_1\in\cR_{\Lambda_1}$ and configuration $\pmb{\tau}\in\{0,1\}^{3}$ so that 
	\begin{equation}\label{eq:case1}
	\psi_\Lambda(R) = \psi_{\Lambda_1}(R_1)\otimes \ket{\pmb{\tau}}.
	\end{equation}
	More precisely, for any root tiling $R$ for which $B^r = B_d^r$, one has $\ket{\pmb{\tau}} =\ket{011}$ and $R_1$ is the restriction of $R$ to $\Lambda_1$. In the case that $[\max{\Lambda_1}, \,\max\Lambda]\cap V\neq \emptyset$, the existence of such an $R_1$ and $\pmb{\tau}$ is guaranteed by Lemma~\ref{lem:factorization}.\\
	
	2.~For the second type, the ordered root tiling of $ R $ ends in at least two monomers (possibly including a right boundary monomer). By Theorem~\ref{thm:product1} there is some $ n \geq 2 $, $ j \in \{ 1,2,3\} $ and some $ \Lambda' \subseteq\Lambda $, $R'\in\cR_{\Lambda'}^{\mathrm{f}}$ such that 
	\[ \psi_{\Lambda}(R)   = \psi_{\Lambda'}(R')\otimes \vp_{n}^{(j)}.  \]
	The recursion relation \eqref{eq:recrel2} then yields
	\begin{equation}
	\psi_{\Lambda}(R)  = \psi_{\Lambda'}(R')\otimes \vp_{n-1}\otimes \vp_1^{(j)} + \lambda \psi_{\Lambda'}(R')\otimes \vp_{n-2}\otimes\ket{\pmb{\sigma}_d^{(j)}} .\label{eq:better_form}
	\end{equation}
	From \eqref{eq:factor1}-\eqref{eq:factor_2} we see that \eqref{eq:better_form} is a linear combination of states from $\cB_{\Lambda_1}$ as desired. \\
	
	Summarizing \eqref{eq:case1}-\eqref{eq:better_form}, it follows that $\psi_{\Lambda_1}(R)\otimes \ket{\pmb{\tau}}\in \cB_{\Lambda_1}'$ if and only if it is a VMD-state on $\Lambda$ or it is of the form displayed in~\eqref{eq:P_vectors}. This completes the proof.	
\end{proof}

\subsection{Proof of Lemma~\ref{lem:epsilon_calc} } \label{subsec:epsilon_calc}
The main goal of this section is to prove Lemma~\ref{lem:epsilon_calc}. If $\Lambda_1,\, \Lambda_2\subseteq \Lambda$ are intervals as in Lemma~\ref{lem:epsilon_calc}, then the inclusion property from Lemma~\ref{lem:P_properties} implies $(\1-G_\Lambda)P_{\Lambda_1}$ is the orthogonal projection onto $\cP_{\Lambda_1}\cap\cG_{\Lambda}^\perp$, see also \eqref{eq:intersection_form}. Therefore, using the first equality in Lemma~\ref{lem:reduction}, the objective is to produce an upper bound on
\begin{equation}\label{eq:final_goal}
\|G_{\Lambda_2}(\1-G_\Lambda)P_{\Lambda_1}\| = \sup_{0\neq \psi \in \cP_{\Lambda_1}\cap \cG_\Lambda^\perp} \frac{\|G_{\Lambda_2}\psi\|^2}{\|\psi\|^2}.
\end{equation}
To do so, we first use the basis of $\cP_{\Lambda_1}$ from Lemma~\ref{lem:orthogonal_basis} to determine an orthogonal basis for $\cP_{\Lambda_1}\cap\cG_{\Lambda}^\perp$ defined in terms of special vectors that are orthogonal to the states $\vp_{n}^{(j)}$. In Lemmas~\ref{lem:exNsmall} and~\ref{lem:Geta} we calculate the action of $G_{\Lambda_2}$ on this basis, the result of which is again an orthogonal set. The section concludes with the proof of Lemma~\ref{lem:epsilon_calc}.\\

To define the orthogonal basis for $\cP_{\Lambda_1}\cap\cG_{\Lambda}^\perp$, we introduce  for any $ n \geq 2 $ three states indexed by $ j \in \{1,2,3 \} $:
\begin{align}
\eta_n^{(j)} & := -\bar{\lambda}\alpha_{n-1} \vp_{n-1}\otimes\vp_1^{(j)} 
+ \vp_{n-2}\otimes \ket{\pmb{\sigma}_d^{(j)}}\in \cH_{[1,3n-3+j]}. \label{eq:eta_N}
\end{align}
Here we use the convention $ \varphi_0=1$, and recall that $ \alpha_n = \| \varphi_{n-1}^{(j)}  \|^2/ \| \varphi_{n}^{(j)}  \|^2 $ is independent of $ j $, cf.~Lemma~\ref{lem:normalization}. These states are not normalized, but rather satisfy
\begin{equation}\label{eq:normeta}
\| \eta_n^{(j)}  \|^2 = \| \varphi_{n-2} \|^2 \left( 1+ \alpha_{n-1} |\lambda|^2 \right) = \| \varphi_{n-3} \|^2/\left( \alpha_{n-2}\alpha_n\right) 
\end{equation}
where the identity $\alpha_n^{-1}=1+\alpha_{n-1}|\lambda|^2$ follows from \eqref{eq:normrec}.  Given the recursion relation~\eqref{eq:recrel2}, the above states are constructed so that 
\begin{equation} \label{eq:orthogonal}
\langle  \vp_n^{(j)} , \eta_n^{(j)}  \rangle = 0 ,
\end{equation}
and the subspace $\cP_{\Lambda_1}\cap \cG_\Lambda^\perp$ has the following explicit form with respect to these vectors.

\begin{lem}[Excess excitation subspace]\label{lem:ON_mm}
	Let $\Lambda_1 \subseteq \Lambda$ be two finite intervals such that $|\Lambda_1|\geq 8$ and $\Lambda\setminus \Lambda_1$ consists of the last three sites of $\Lambda$. Then,
	\begin{equation}\label{eq:mm_vspace1}
	\cP_{\Lambda_1}\cap \cG_\Lambda^\perp = 
	\spa\left\{ \psi_{\Lambda'}(R') \otimes \eta_n^{(j)} \, \big| \,  \mbox{$ n \geq 2 $, $j\in\{1,2,3\},$  $|\Lambda\setminus \Lambda'|=3(n-1)+j$ and $ R'\in \mathcal{R}_{\Lambda'}^{\mathrm{f}}$ } \right\}
	\end{equation}
	where we use the convention that $\psi_{\Lambda'}(R') = 1$ if $\Lambda'=\emptyset$. Moreover, the spanning set is an orthogonal basis for $\cP_{\Lambda_1}\cap \cG_\Lambda^\perp$. 
\end{lem}
\begin{proof} Let $\cB_{\Lambda_1}'$ be the orthogonal basis of $\cP_{\Lambda_1}$ from Lemma~\ref{lem:orthogonal_basis}. Since the VMD-states are an orthogonal basis for $\cG_\Lambda$, it follows that $\eta\in \cP_{\Lambda_1}\cap \cG_{\Lambda}^\perp$ if and only if $\eta \in \spa \cB_{\Lambda_1}'$ and
	\begin{equation}\label{orthogonality}
	\braket{\psi_{\Lambda}(R)}{\eta} = 0 \quad \text{for all} \quad R\in \cR_\Lambda. 
	\end{equation}
	As a consequence, $\eta$ must be orthogonal to any vector from the first case in Lemma~\ref{lem:orthogonal_basis}. 
	Given the form of the vectors from the second case in Lemma~\ref{lem:orthogonal_basis}, we hence conclude that $\eta\in \cP_{\Lambda_1}\cap\cG_{\Lambda}^\perp$ if and only if
	\begin{equation}\label{eq:eta}
	\eta = 
	\sum_{n\geq 2}
	\sum_{j\in\{1,2,3\}}
	\sum_{\substack{R' \in \cR_{\Lambda'}^{\mathrm{f}} : \\ |\Lambda\setminus\Lambda'|=3(n-1)+j}} 
	\psi_{\Lambda'}(R')\otimes \left[c^{R'}_{n,j}\vp_{n-1}\otimes\vp_1^{(j)} 
	+ 
	d^{R'}_{n,j} \vp_{n-2}\otimes \ket{\pmb{\sigma}_d^{(j)}}\right]
	\end{equation}
	with coefficients $c^{R'}_{n,j}, d^{R'}_{n,j}\in \bC$ such that \eqref{orthogonality} is satisfied. Using the recursion relation~\eqref{eq:recrel2} to apply the orthogonality constraint~\eqref{orthogonality} with $\psi_\Lambda(R)=\psi_{\Lambda'}(R')\otimes \vp_n^{(j)}$ implies
	\begin{align}\label{eq:dependence}
	 0 = & \left\langle  \vp_n^{(j)} \big|  \left[c^{R'}_{n,j}\vp_{n-1}\otimes\vp_1^{(j)} 
	+ 
	d^{R'}_{n,j} \vp_{n-2}\otimes \ket{\pmb{\sigma}_d^{(j)}}\right] \right\rangle = c_{n,j}^{R'}\|\vp_{n-1}\|^2 + d_{n,j}^{R'}\bar{\lambda}\|\vp_{n-2}\|^2 \notag \\
	& \implies \quad
	c_{n,j}^{R'} = -d_{n,j}^{R'}\, \bar{\lambda}\, \alpha_{n-1}.
	\end{align}
	The validity of these relations for any $ n \geq 2 $ and  $ j \in \{1,2,3 \} $, then implies the orthogonality~\eqref{orthogonality} for all VMD states. This claim  is seen by noting that the vector $ \eta $ in~\eqref{eq:eta} is supported on tiling configurations $ \pmb{\sigma}_\Lambda(\pmb{D}) $. In order to have a non-zero scalar product with a VMD state $ \psi_{\Lambda}(R) $, the tiling $ \pmb{D} $ would have to have $ R $ as its root. 
	This requires $ \psi_\Lambda(R)=\psi_{\Lambda'}(R')\otimes \vp_n^{(j)}$ with some $ n \geq 2 $, $ j \in \{1,2,3 \} $ and $ R' \in \cR_{\Lambda'}^{\mathrm{f} }$, for which the orthogonality~\eqref{orthogonality} is ensured by~\eqref{eq:dependence}.

	The orthogonality of the spanning set follows by explicit inspection since each vector is supported on a unique set of tiling configurations.
\end{proof}

As a first step towards bounding~\eqref{eq:final_goal}, we calculate the action of $G_{\Lambda_2}$ on the basis constructed in the previous lemma. In the situation of interest,  the interval $\Lambda_2$ consists of the last nine sites of an interval $\Lambda$, and $G_{\Lambda_2}$ acts as the identity on all sites $\Lambda \setminus \Lambda_2$. Given the fragmentation property of the VMD states, cf. Theorem~\ref{thm:product1},  the action of $G_{\Lambda_2}$ on the basis in Lemma~\ref{lem:ON_mm} falls into two cases depending on the support 
\[ \Lambda_{n,j}:=\supp\big(\eta_n^{(j)}\big) \]
of the vector~$ \eta_n^{(j)} $. The next lemma deals with $ \Lambda_{n,j} \subseteq \Lambda_2$ and further down we treat the case $\Lambda_{n,j}\supset \Lambda_2$.

\begin{lem}[Excited states I]\label{lem:exNsmall} Let $\Lambda_2$ be an interval with $|\Lambda_2|=9$. Fix $n\in\{ 2,3\} $ and $j\in\{1,2,3\}$, and define $  \Lambda_{n,j} \subseteq \Lambda_2 $ to be the last $3(n-1)+j$ sites of $\Lambda_2$.	Then for any $\pmb{\sigma} \in \{ 0, 1 \}^{\Lambda_{n,j}^c}$, 
	\begin{equation}\label{eq:excited_state}
	G_{\Lambda_2}\big(|\pmb{\sigma} \rangle \otimes\eta_{n}^{(j)}\big) = 0,
	\end{equation}
	where we employ the convention that $| \pmb{\sigma} \rangle = 1 $ if $\Lambda_{n,j}^c = \emptyset $ (i.e. for $ n =j =  3 $).
\end{lem}
\begin{proof}
	
Fix $n\in\{2,3\}$. We prove that every VMD state $ \psi_{\Lambda_2}(R) $ is orthogonal to $ |\pmb{\sigma} \rangle \otimes\eta_{n}^{(j)} $. Expanding $ \psi_{\Lambda_2}(R) $ into tiling configurations as in~\eqref{eq:VMDocc},  we consider the scalar product
\[ \big\langle \pmb{\sigma}_{\Lambda_2}( \pmb{D}) \big|  \pmb{\sigma} \rangle \otimes\eta_{n}^{(j)}  \big\rangle 
\]
for any $\pmb{D}\in\cD_{\Lambda_2}(R)$. Clearly, $ \psi_{\Lambda_2}(R) $ is orthogonal to $ |\pmb{\sigma} \rangle \otimes\eta_{n}^{(j)} $ if this scalar product is zero for all $\pmb{D}\in \cD_{\Lambda_2}(R)$. For the scalar product to be non-zero for some $\pmb{D}$, the root-tiling $  \pmb{D}_R $ must cover $  \Lambda_{n,j} $ with $ n $ monomers the last of which being a right $ j $-monomer if $ j \in \{ 1,2 \} $. This requires $  \psi_{\Lambda_2}(R) = \psi_{\Lambda'}(R') \otimes \vp_m^{(j)} $ for some $n\leq m \leq 3$ and $R' \in  \cR_{\Lambda'}^{\mathrm{f}}  $ with $ \Lambda' $ appropriately defined (possibly empty). In the case that $n=m$, the orthogonality relation in~\eqref{eq:orthogonal} implies
\[\big\langle \psi_\Lambda(R) \big|  \pmb{\sigma} \rangle \otimes\eta_{n}^{(j)}  \big\rangle = \big\langle \ \psi_{\Lambda'}(R') \big|  \pmb{\sigma} \rangle \cdot  \big\langle \vp_n^{(j)} \big| \eta_{n}^{(j)}  \big\rangle = 0.\] 
If $n=2$ and $m=3$, using~\eqref{eq:recrel3} to write
\[\psi_{\Lambda'}(R') \otimes \vp_3^{(j)} =   \psi_{\Lambda'}(R') \otimes \left(  \vp_1 \otimes \vp_2^{(j)} + \lambda \ket{ \pmb{\sigma}_d } \otimes   \vp_1^{(j)} \right), \]  
a short calculation involving~\eqref{eq:orthogonal} and the definition of $\eta_2^{(j)}$, see~\eqref{eq:eta_N}, again verifies $\big\langle \psi_\Lambda(R) \big|  \pmb{\sigma} \rangle \otimes\eta_{2}^{(j)}  \big\rangle = 0$ .

%
%
%
\end{proof}

Lemma~\ref{lem:exNsmall} implies that if $\Lambda_2$ is the last nine sites of a finite interval $\Lambda$, then for $n=2,3$ and $j\in\{1,2,3\}$
\begin{equation}
G_{\Lambda_2}\psi_{\Lambda'}(R')\otimes \eta_n^{(j)} = 0
\end{equation}
where $R'\in \cR_{\Lambda'}^{\mathrm{f}}$ is arbitrary and $|\Lambda\setminus\Lambda' |= 3(n-1)+j$.

We now turn to the case that $\Lambda_2\subset \Lambda_{n,j}\subseteq \Lambda$. In this situation, given any $\psi_{\Lambda'}(R')\otimes\eta_{n}^{(j)}\in \cH_{\Lambda}$
\[
G_{\Lambda_2}(\psi_{\Lambda'}(R')\otimes\eta_{n}^{(j)}) = \psi_{\Lambda'}(R')\otimes G_{\Lambda_2}\eta_{n}^{(j)}.
\]
In the next lemma, we compute $G_{\Lambda_2}\eta_{n}^{(j)}$ for $n\geq 4$.

\begin{lem}[Excited states II] \label{lem:Geta}
	Fix $n\geq 4$ and $j\in\{1,2,3\}$, and let $\Lambda_2$ be the last nine sites of $[1,3(n-1)+j]$. Then
	\begin{equation}\label{eq:excitedG2}
	G_{\Lambda_2}\eta_n^{(j)}  = \frac{\bar{\lambda}\left[1-\alpha_{n-1}(1+|\lambda|^2)\right]}{\|\vp_3\|^2} \vp_{n-3}\otimes\vp_3^{(j)} + \frac{|\lambda|^2(1-\alpha_{n-1})}{\|\vp_2\|^2}\vp_{n-4}\otimes\ket{\pmb{\sigma}_d}\otimes \vp_2^{(j)}
	\end{equation}
	where $\ket{\pmb{\sigma}_d}=\ket{\pmb{\sigma}_d^{(3)}}=\ket{011000}$.
\end{lem}

\begin{proof}
We use the recursion relation \eqref{eq:recursion} to write
\begin{align}
\eta_n^{(j)} & = \vp_{n-3} \otimes  \left( -\overline{\lambda} \alpha_{n-1}   \varphi_{2}\otimes  \varphi_{1}^{(j)} +   \varphi_{1} \otimes | \pmb{\sigma}_d^{(j)}\rangle \right) 
+ \lambda  \vp_{n-4} \otimes | \pmb{\sigma}_d \rangle   \otimes \left(   -\overline{\lambda} \alpha_{n-1}   \varphi_{1}\otimes  \varphi_{1}^{(j)} +   | \pmb{\sigma}_d^{(j)} \rangle \right) \notag \\
& =   \vp_{n-3}^{(j)}  \otimes \xi_1^{(j)} + \lambda \vp_{n-4} \otimes  | \pmb{\sigma}_{d,1}^{(j)} \rangle  \otimes \xi_2^{(j)} 
\end{align} 
where in the last step we factored out vectors on the last nine sites~using:
\begin{enumerate}
	\itemsep0pt
\item[(i)] the factorization~\eqref{eq:factorize} to extract terminating zeros and write  $  \vp_{n-3} = \vp_{n-3}^{(j)}\otimes\ket{ \pmb{0}_j} $ where $\ket{ \pmb{0}_j}=\ket{0}^{\otimes 3-j}$. As a result,
\begin{equation}
\xi_1^{(j)} := \ket{ \pmb{0}_j } \otimes  \left( -\overline{\lambda} \alpha_{n-1}   \varphi_{2}\otimes  \varphi_{1}^{(j)} +   \varphi_{1} \otimes | \pmb{\sigma}_d^{(j)}\rangle \right) \in \mathcal{H}_{\Lambda_2} ;
\end{equation}
\item[(ii)]  the factorization $  \ket{\pmb{\sigma}_d} =  \ket{\pmb{\sigma}_{d,1}^{(j)}}\otimes \ket{ \pmb{\sigma}_{d,2}^{(j)}} $ from~\eqref{eq:splitconf}  
for which we define 
\begin{equation}
\xi_2^{(j)} :=   | \pmb{\sigma}_{d,2}^{(j)} \rangle  \otimes  \left(   -\overline{\lambda} \alpha_{n-1}   \varphi_{1}\otimes  \varphi_{1}^{(j)} +   | \pmb{\sigma}_d^{(j)} \rangle \right)  \in \mathcal{H}_{\Lambda_2}  .
\end{equation}
\end{enumerate}
The projection $ G_{\Lambda_2} $ only acts non-trivially on the vectors $  \xi_1^{(j)} $ and  $\xi_2^{(j)} $. For each vector $\xi_k^{(j)}$ there is a unique root $R_k^{(j)}\in\cR_\Lambda$ for which it has a non-zero scalar product. Specifically, $ \psi_{\Lambda_2}(R_1^{(j)}) =  \ket{ \pmb{0}_j }  \otimes \vp_3^{(j)} $ is the unique VMD state with non-zero scalar product with $   \xi_1^{(j)} $ given by
\begin{equation}
\left\langle  \psi_{\Lambda_2}(R_1^{(j)})  \big|   \xi_1^{(j)}  \right\rangle = \left\langle   \vp_3^{(j)}  \big|  \left( -\overline{\lambda} \alpha_{n-1}   \varphi_{2}\otimes  \varphi_{1}^{(j)} +   \varphi_{1} \otimes | \pmb{\sigma}_d^{(j)}\rangle \right)  \right\rangle =  \overline{\lambda} \left( 1 - \alpha_{n-1}  \| \vp_2 \|^2  \right) . 
\end{equation}
Note that $\|\vp_2\|^2=1+|\lambda|^2$. For  $   \xi_2^{(j)} $, the vector $ \psi_{\Lambda_2}(R_2^{(j)}) =  \ket{  \pmb{\sigma}_{d,2}^{(j)}}  \otimes \vp_2^{(j)} $ has a non-zero scalar product:
\begin{equation}
\left\langle  \psi_{\Lambda_2}(R_2^{(j)})  \big|   \xi_2^{(j)}  \right\rangle = \left\langle   \vp_2^{(j)}  \big|   \left(   -\overline{\lambda} \alpha_{n-1}   \varphi_{1}\otimes  \varphi_{1}^{(j)} +   | \pmb{\sigma}_d^{(j)} \rangle \right)\right\rangle =  \overline{\lambda} \left( 1 - \alpha_{n-1}   \right) . 
\end{equation}
Since $ \|  \psi_{\Lambda_2}(R_1^{(j)})  \|^2 =  \|  \vp_3 \|^2 $ and $ \|  \psi_{\Lambda_2}(R_2^{(j)})  \|^2 =  \|  \vp_2 \|^2$, we then arrive at
\[
G_{\Lambda_2}\eta_n^{(j)}  =  \frac{ \overline{\lambda} }{ \|  \vp_3 \|^2 }  \left( 1 - \alpha_{n-1}  \| \vp_2 \|^2  \right) \vp_{n-3}^{(j)}  \otimes   \ket{ \pmb{0}_j }  \otimes \vp_3^{(j)}  +  \frac{|\lambda|^2}{  \|  \vp_2 \|^2 }  \left( 1 - \alpha_{n-1}   \right) \vp_{n-4} \otimes  | \pmb{\sigma}_{d,1}^{(j)} \rangle  \otimes   \ket{  \pmb{\sigma}_{d,2}^{(j)}}  \otimes \vp_2^{(j)} .
\]
Regrouping the factorized terms from (i) and (ii), we thus arrive at~\eqref{eq:excitedG2}.
\end{proof}

For any sufficiently large finite interval $\Lambda$, the set of vectors created from applying Lemma~\ref{lem:Geta} to the basis in Lemma~\ref{lem:orthogonal_basis}, i.e.
\[
\left\{\varrho_{n}^{(j)}(R') := G_{\Lambda_2}(\psi_{\Lambda'}(R')\otimes \eta_n^{(j)}) \; \big| \; n\geq 4, \, j\in\{1,2,3\}, \, R'\in\cR_{\Lambda'}^{\mathrm{f}}\, \;\text{where} \;|\Lambda\setminus\Lambda'|=3(n-1)+j\right\}
\]
is again orthogonal as each vector is supported on a unique set of tiling configurations. We use this orthogonality to prove Lemma~\ref{lem:epsilon_calc}, thus completing the proof of Theorem~\ref{thm:gap}.

\begin{proof}[Proof of Lemma~\ref{lem:epsilon_calc}]
	Suppose that $\Lambda$ is a finite interval with $|\Lambda| \geq 11$, and define subintervals $\Lambda_1$, $\Lambda_2$ so that $\Lambda\setminus \Lambda_1$ consists of the last three sites of $\Lambda$, and $\Lambda_2$ is the last nine sites of $\Lambda$. Then $|\Lambda_1\cap \Lambda_2|=6$ and the reduction Lemma~\ref{lem:reduction} is applicable
	\[
	\|G_{\Lambda_2}(\1-G_\Lambda)G_{\Lambda_1}\|^2 = \|G_{\Lambda_2}(\1-G_\Lambda)P_{\Lambda_1}\|^2 = \sup_{0\neq \psi \in  \cP_{\Lambda_1}\cap\cG_{\Lambda}^\perp}\frac{\|G_{\Lambda_2}\psi\|^2}{\|\psi\|^2}.
	\]
	Given the orthogonal basis for $\cP_{\Lambda_1}\cap\cG_{\Lambda}^\perp$ from Lemma~\ref{lem:ON_mm} and Lemma~\ref{lem:exNsmall}, it suffices to consider the vectors
	\[ \varrho_n^{(j)}(R') =  G_{\Lambda_2}  (\psi_{\Lambda'}(R') \otimes \eta_n^{(j)} )
	\]
	for $ n \geq 4 $, $j\in\{1,2,3\}$ and $R'\in\cR_{\Lambda'}^{\mathrm{f}}$, which can be explicitly computed using Lemma~\ref{lem:Geta}. 
	
	Since $\varrho_n^{(j)}(R')$  are orthogonal for distinct $n, j $ and $R'$, it suffices to compute their norms
	\begin{align}
	\| \varrho_n^{(j)}(R') \|^2  = & \ \| \psi_{\Lambda'}(R') \|^2 \left( \|  \varphi_{n-3}\|^2 \  \frac{ |\lambda|^2 (1- \alpha_{n-1} (1+|\lambda|^2))^2 }{ \| \varphi_3 \|^2} +\| \varphi_{n-4} \|^2 \frac{ |\lambda|^4 (1-\alpha_{n-1})^2}{\| \vp_2 \|^2 }\right) \notag \\
	= &  \ \| \psi_{\Lambda'}(R') \|^2  \| \eta_n^{(j)} \|^2   f_n(|\lambda|^2) 
	\end{align}
	where in the last line we used the normalization~\eqref{eq:normeta} and the definition of $f_n$:
	\[ f_n(|\lambda|^2) := \alpha_{n-2}\alpha_n |\lambda|^2 \left(\frac{ (1-\alpha_{n-1} (1+|\lambda|^2))^2 }{ 1+2|\lambda|^2} + \alpha_{n-3}  \frac{ |\lambda|^2 (1-\alpha_{n-1})^2}{1+|\lambda|^2} \right) .
	\]
	By the above mentioned orthogonality and the independence of $ f_n $ of $ j $ and $ R' $, we thus have
	\begin{equation}
	\| G_{\Lambda_2}(\1-G_{\Lambda}) P_{\Lambda_1} \|^2 \leq \sup_{\substack{ p_n \geq 0 \\ \sum_{n\geq 4} p_n = 1} } \sum_{n\geq 4} p_n   f_n(|\lambda|^2) \leq \sup_{n \geq 4 } f_n(|\lambda|^2) = f(|\lambda|^2). 
	\end{equation}
	The claimed properties of $ f $ are established in Appendix~\ref{app:f_estimates}. 
\end{proof}

\subsection{Bounding the spectral gap for periodic boundary conditions}\label{subsec:periodic_gap}

We now consider the spectral gap of the FQH system with periodic boundary conditions. To prove our result, we use a slightly generalized form of the finite-size criterion proved by Knabe~\cite{knabe:1988}. The martingale method  for the periodic systems can also be applied~\cite{young:2016}, but its computations are significantly more involved. \\

For the convenience of the reader, we first state the generalized form of the finite-size criteria. Let $\cH$ be a finite-dimensional Hilbert space and suppose that $\{P_i \; \big| \; i=1,\ldots,L\}$ with $ L \in \mathbb{N} $ is a set of orthogonal projections on $\cH$ such that the following commutation condition holds:
\begin{equation}\label{eq:fsc_projections}
[P_i,P_j] \neq 0 \quad \implies \quad |i-j|=1 \;\; \text{or} \;\; i=1, \; j=L.
\end{equation}
We fix $1< n <L $ and employ the periodic identification $i\equiv i+L$ to define on $\cH$ the self-adjoint Hamiltonians
\begin{equation}\label{eq:fsc_Hams}
H_L = \sum_{i=1}^L P_i, \quad\quad
H_{n,k} = \sum_{i=k}^{n+k-1} P_i
\end{equation}
for $k=1,\ldots, L$. Provided the Hamiltonian is frustration free, i.e. $\ker(H_L) \neq \emptyset$, the following estimate holds for the spectral gap 
\[ \gap(H_L) = \inf \left\{ \langle \psi , H_L \psi \rangle \, | \, \, \psi \in  \ker(H_L)^\perp \, \wedge \, \| \psi \| = 1\right\}.  
\]
\begin{theorem}[Finite-Size Criterion]\label{thm:knabe}
	Fix $1< n < L$ and consider the frustration-free Hamiltonians $H_L$ and $H_{n,k}$, $k=1,\ldots, L$, on a finite-dimensional Hilbert space  defined as in \eqref{eq:fsc_Hams} with orthogonal projections satisfying \eqref{eq:fsc_projections}. Then,
	\begin{equation}\label{eq:knabe_gap}
	\gap(H_L) \geq \frac{n}{n-1}\left( \min_{1\leq k \leq L} \gap(H_{n,k}) - \frac{1}{n}\right).
	\end{equation}
\end{theorem}

\begin{proof}
	The proof follows the argument of Knabe in~\cite{knabe:1988}, i.e.\ we show that $H_L^2 \geq \gamma_n H_L$, where $\gamma_n$ is the quantity on the right side of \eqref{eq:knabe_gap}. 
	In turn, since $H_{n,k}^2\geq \gap(H_{n,k})H_{n,k}$ and
	\begin{equation}\label{eq:sum_equality}
	\sum_{k=1}^LH_{n,k}=nH_L
	\end{equation}
	this operator inequality is concluded from the bound
	\begin{equation}\label{eq:knabe_ineq}
	(n-1)H_L^2 + H_L \geq \sum_{k=1}^L H_{n,k}^2.
	\end{equation}
	
	To prove the latter, we denote by $\mathcal{I}_{L} := \{(i,j) \, : \, 1\leq i<j \leq L\}$ the set of all (distinct) pairings. Since each $P_i$ is an orthogonal projection, we can expand $H_L^2$ and $H_{n,k}^2$ in terms of anti-commutators, 
	\begin{align}
	H_{L}^2 & = 
	H_L + \sum_{(i,j)\in\mathcal{I}_L}\{P_i,\, P_j\} , \\
	H_{n,k}^2 & = H_{n,k}+\sum_{\substack{(i,j)\in\mathcal{I}_L \, : \\ i,j\in\{k,\ldots, n+k-1\}}}\{P_i,\, P_j\} . \label{eq:Hnk_squared}
	\end{align}
	Let $d(i,j)\leq L/2$ to be the distance between $i$ and $j$ on a ring of $L$ sites. 
	If $d(i,j)\geq n$, then $i,j\notin \{k,\ldots, n+k-1\}$ for any $k$ as each interval $[k,n+k-1]$ has length $n-1$. Consider the case that $m=d(i,j)<n$ and assume without loss of generality that $j\equiv i+m$. Then $i,j\in\{k,\ldots, n+k-1\}$ if and only if $k\in S_{i,j}:= \{i, \, i-1, \, \ldots,\, j-n+1\} $, where we again used the identification $ i \equiv i+L $. Since $|S_{i,j}|=n-m$, summing \eqref{eq:Hnk_squared} over $k$ and using \eqref{eq:sum_equality} produces
	\begin{align*}
	\sum_{k=1}^L H_{n,k}^2 
	& = nH_L+\sum_{m=1}^{n-1}\sum_{\substack{(i,j)\in\mathcal{I}_L : \\ d(i,j)=m}}(n-m)\{P_i,\, P_j\}.
	\end{align*}
	For any pair with $d(i,j)\geq2$, the assumption~\eqref{eq:fsc_projections} implies $\{P_i,P_j\} = 2P_iP_j \geq 0$. In particular, terms with $ m \geq 2 $ are hence non-negative. Adding more of these terms produces the inequalities	
	\begin{align*}
	\sum_{k=1}^L H_{n,k}^2
	& \leq  nH_L+(n-1)\sum_{m=1}^{n-1}\sum_{\substack{(i,j)\in\mathcal{I}_L : \\ d(i,j)=m}}\{P_i,\, P_j\}\\
	& \leq  nH_L+(n-1)\sum_{(i,j)\in\mathcal{I}_L}\{P_i,\, P_j\}\\
	& = (n-1)H_L^2 + H_L . 
	\end{align*}
	This completes the proof of~\eqref{eq:knabe_ineq} and hence the theorem.
\end{proof}
We apply the finite size criteria to prove that the FQH system with periodic boundary conditions is gapped in the thermodynamic limit. Using the identification $x\equiv x+L$, the Hamiltonian is defined by
\begin{equation}\label{eq:periodic}
H_{[1,L]}^{\per} = \sum_{x=1}^{L}(n_xn_{x+2} + \kappa q_x^*q_x), \quad q_x = \sigma_{x+1}^-\sigma_{x+2}^- - \lambda\sigma_x^-\sigma_{x+3}^-.
\end{equation}
We prove a uniform lower bound on the spectral gap for the periodic Hamiltonian defined on $[1,L]$ with $L$ sufficiently large. Our estimate in terms of the quantities 
 \[ \gamma := \min_{m \in \{ 6,7 \} } \gap(H_{[1,m]}), \qquad \Gamma := \max_{m \in \{ 6,7 \} }  \|H_{[1,m]} \| , 
 \]
where $H_{[1,m]}$ is the Hamiltonian~\eqref{def:Hspin} with open boundary conditions on fixed small system sizes. The next theorem is a reformulation of Theorem~\ref{thm:periodic_gap}. 
\begin{theorem}[Periodic Spectral Gap]\label{thm:pbc_gap}
	Let $n\geq 2 $ be an integer such that 
	\[ g_n := \min_{0\leq r \leq 5}  \gap(H_{[ 1, 3(n+1) + r]})  > \frac{\Gamma}{ n}. \] Then for all $N> n$ and any $ 0 \leq r \leq 5 $:
	\begin{equation}\label{eq:pergap3}
	\gap\left(H_{[1,6N+r]}^\per\right) \geq \frac{\gamma\,  n}{2\Gamma(n-1)}\left[g_n -\frac{\Gamma}{n}\right] . 
	\end{equation}
\end{theorem}
\begin{proof}
	We consider for $ N > n \geq 2 $  the interval $[1,6N + r]$ as a ring of $L= 6N+r$ sites, and define $ 2N $ open intervals $\Lambda_1, \dots \Lambda_{2N} $ on the ring by:
	\begin{equation*}
	\Lambda_i := [3i-2,3i+3], \quad i = 1, 2, \ldots, 2N-r ,
	\end{equation*}
	and in the case $r \geq 1$
	\begin{equation*}
	\Lambda_{2N-r+j} := [ 3 (2N-r+j) -3+j , 3 (2N-r+j) +3+j ], \quad j = 1 , 2 , \ldots, r .
	\end{equation*}
	For the last interval $ \Lambda_{2N} $, we again use the convention that $x\equiv x+L$. Notice that for $ 1 \leq i \leq 2N-r $ the interval $ \Lambda_i  $ has six sites, whereas the last $ r $ intervals have seven sites. 
	
	We denote by $H_{\Lambda_i}$  the Hamiltonian on $\Lambda_i$ with open boundary conditions~\eqref{def:Hspin}, and let $P_i$ be the orthogonal projection onto $\ran(H_{\Lambda_i})$. These projections satisfy \eqref{eq:fsc_projections} since $\supp(P_i)\cap\supp(P_j) = \emptyset$ for all $j\neq i\pm 1$. Therefore, setting $H_L = \sum_{i=1}^{2N} P_i$, and $H_{n,k} = \sum_{i=k}^{n+k-1}P_i$, Theorem~\ref{thm:knabe} implies
	\begin{equation}\label{eq:proj_bound}
	\gap(H_L) \geq \frac{n-1}{n}\left(\min_{1\leq k \leq 2N}\gap(H_{n,k})-\frac{1}{n}\right).
	\end{equation}
	The interval $\Lambda_{n,k}:=\bigcup_{i=k}^{n+k-1}\Lambda_i$  denotes the supports of  $ H_{n,k} $. 
	Since every interaction term, $n_xn_{x+2}$ or $q_x^*q_x$, is supported on at least one and at most two of the volumes $\Lambda_i$, it readily follows that
	\begin{equation}
	H_{\Lambda_{n,k}} \leq \sum_{i=k}^{n+k-1} H_{\Lambda_i} \leq 2H_{\Lambda_{n,k}}, \quad
	H_{[1,L]}^{\per} \leq \sum_{i=1}^{2N}H_{\Lambda_i} \leq 2H_{[1,L]}^\per \label{eq:first_bound}.
	\end{equation}
	Moreover, as $P_i$ is the orthogonal projection onto $\ran(H_{\Lambda_i})$
	\begin{equation}\label{eq:projection_bounds}
	\gamma P_i \leq H_{\Lambda_i} \leq \Gamma P_i.
	\end{equation}
	Summing \eqref{eq:projection_bounds} over appropriate values of $i$ and using \eqref{eq:first_bound} produces the operator inequalities
	\begin{equation}
		\frac{\gamma}{2} H_{n,k}  \leq H_{\Lambda_{n,k}} \leq \Gamma H_{n,k} , \quad \frac{\gamma}{2} H_L  \leq H_{[1,L]}^{\per} \leq \Gamma  H_L ,
	\end{equation}
	from which it readily follows that   for all $ k \in \{1, \ldots, 2N \} $:
	\begin{equation}\label{eq:gap_inequalities}
	\gap(H_{n,k}) \geq  \Gamma^{-1}\gap(H_{\Lambda_{n,k}}) , \quad  \gap(H_{[1,L]}^\per) \geq \frac{\gamma}{2}\gap(H_L) .
	\end{equation}
	Depending on how many intervals $ \Lambda_i $ of size seven it includes, the interval  $ \Lambda_{n,k} = \bigcup_{i=k}^{n+k-1}\Lambda_i $ has at least $ 3(n+1) $ sites and at most $ 3(n+1) + r \leq 3(n + 1 ) +5 $ sites.
	By translation invariance, we thus have 
	\begin{equation}\label{eq:n_gap}
	\min_{1\leq k \leq 2N} \gap(H_{\Lambda_{n,k}}) \geq \min_{0\leq r \leq 5}  \gap(H_{[ 1, 3(n+1) + r]}) = g_n  .
	\end{equation}
	Using \eqref{eq:gap_inequalities} and \eqref{eq:n_gap} to further bound \eqref{eq:proj_bound} produces the result.
\end{proof}

\section{Exponential clustering} \label{sec:correlation_decay}

In this section, we prove exponential decay of correlations for the VMD states (in the spin representation) and use this to establish Theorem~\ref{thm:expcluster} via the Jordan-Wigner transformation. Moreover, we provide an explicit example which shows that correlations cannot decay uniformly exponentially for more general pure states in the kernel of~$ H_\Lambda $.

\subsection{Exponential decay of correlations for VMD states} \label{subsec:exp_decay}

To restate the first claim in Theorem~\ref{thm:expcluster} in the spin language, we recall from~\eqref{eq:algebra} that in this setting the algebra of observables $ \mathcal{A}_\Lambda $ is the set of bounded operators on the underlying tensor-product Hilbert space $ \mathcal{H}_\Lambda =  \bigotimes_{x\in \Lambda } \mathbb{C}^2 $.  
On this algebra we define the VMD-functional $\omega_{R}^\Lambda:\cA_\Lambda\to\bC$, 
\begin{equation} \label{VMD_functional}
\omega_{R}^\Lambda(A) := \braket{\widehat{\psi}_\Lambda(R)}{A \widehat{\psi}_\Lambda(R)} ,
\end{equation} 
associated with a normalized 
VMD state characterized by $R\in \cR_\Lambda $. 
Here and in the following, $\widehat{\psi}:= \psi/\|\psi\|$ stands for the normalized version of any nonzero vector $\psi$. 
We prove the following equivalent form of~\eqref{eq:ABexpdecay} for spin observables.

\begin{theorem}[Exponential decay of correlations]\label{thm:exp_decay}
	Let $\Lambda$ be a finite interval and $X,Y\subset \Lambda$ be two sub-intervals such that $\dist(X,Y)\geq 20$. For any $R\in\cR_\Lambda$ and observables $A_1\in\cA_X$, $A_2\in\cA_Y$,
	\begin{equation} \label{exp_decay}
	\left| \omega_{R}^\Lambda(A_1A_2)- \omega_{R}^\Lambda(A_1)\, \omega_{R}^\Lambda(A_2) \right| \leq \, 8 \|A_1\|\|A_2\| \, e^{-c(\lambda)(\dist(X,Y)-20) /2}
	\end{equation}
	with $ c(\lambda ) $  from~\eqref{decay_rate}.
\end{theorem}

The proof of Theorem~\ref{thm:exp_decay}, which is found at the end of this subsection, rests on the explicit factorization property, which can be read of from 
the canonical form of VMD states in Theorem~\ref{thm:product1}. 
It requires us to estimate the result of trimmed expectations $\braket{\psi}{A \psi}$ for vectors of the form
\[
\psi = \xi\otimes\vp_{k+n} 
\]
with $ \xi \in \mathcal{H}_{\Lambda_0} $ and $\supp(A) \subset \Lambda_k := [a,b+3k]$ and $ k \in \mathbb{N}_0 $. The following lemma will be instrumental for this task.

\begin{lem}[Trimming]\label{lem:vp_shrink}
	Consider a normalized vector $ \xi \in \mathcal{H}_{\Lambda_0} $ and an observable $ A \in \mathcal{A}_{\Lambda_k} $ supported in the interval $ \Lambda_k : = [a,b+3k]$ with fixed $ k \in \mathbb{N}_0 $. Then for any $m>  n\geq 1$ the vectors $\psi_n := \xi \otimes  \widehat \varphi_{k+n} \in \mathcal{H}_{[a,b+3(k+n)]}$ satisfy
	\begin{equation} \label{one_step_decay}
	\left| \braket{\psi_m}{A\psi_m}-\braket{\psi_n}{A\psi_n}\right| \leq 2(1+\beta-\beta^3)\|A\|\beta^{n}
	\end{equation}
	where $\beta = e^{-3c(\lambda)} \in(0,1)$.
\end{lem}
\begin{proof}
	We first consider the case that $m=n+1$. For $ n = 0 $ the claimed bound is trivial as $ \beta \in (0,1) $ implies
	\[
	\left| \braket{\psi_1}{A\psi_1}-\braket{\psi_0}{A\psi_0}\right| \leq 2\|A\| \leq 2(1+\beta-\beta^3)\|A\|.
	\]
	We assume now $n\geq 1$. Using \eqref{eq:recrel2} with $ j = 3 $ and recalling from~\eqref{def:alpha1} the definition $\alpha_n = \|\vp_{n-1}\|^2/\|\vp_n\|^2$, we rewrite 
	\[
	\psi_{n+1} = \xi\otimes \widehat{\vp}_{l+n+1} = \sqrt{\alpha_{l+n+1}} \psi_n\otimes \vp_1 + \lambda \sqrt{\alpha_{l+n}\alpha_{l+n+1}} \psi_{n-1} \otimes \ket{\pmb{\sigma}_d},
	\]
	where in the second line we use the recursion $\alpha_n\alpha_{n-1}=\|\vp_{n-2}\|^2/\|\vp_n\|^2.$ Inserting the above to compute $\braket{\psi_{n+1}}{A\psi_{n+1}}$, one finds that the cross terms vanish since $\supp A \subset X\subset \Lambda_{n-1}$ and $\left(\1 \otimes \langle  \vp_1 |\right) \ket{\pmb{\sigma}_d} = 0$. As a result, we have
	\begin{align}
	\braket{\psi_{n+1}}{A\psi_{n+1}} 
	& =
	\alpha_{l+n+1}\braket{\psi_{n}}{A\psi_{n}} + |\lambda|^2\alpha_{l+n}\alpha_{l+n+1}\braket{\psi_{n-1}}{A\psi_{n-1}} \nonumber\\
	& =
	\alpha_{l+n+1}\braket{\psi_{n}}{A\psi_{n}} + (1-\alpha_{l+n+1})\braket{\psi_{n-1}}{A\psi_{n-1}}
	\end{align}
	where the last equality follows from the normalization of the states. Subtracting $\braket{\psi_{n}}{A\psi_{n}}$ and taking the absolute value produces the following recursion relation:
	\begin{equation*}
	\left|\braket{\psi_{n+1}}{A\psi_{n+1}}  -  \braket{\psi_{n}}{A\psi_{n}} \right| =    | 1-\alpha_{l+n+1}| \left|\braket{\psi_{n}}{A\psi_{n}}  -\braket{\psi_{n-1}}{A\psi_{n-1}}  \right| , 
	\end{equation*}
	which by iteration yields 
	\begin{align}
	\left|\braket{\psi_{n+1}}{A\psi_{n+1}}  -  \braket{\psi_{n}}{A\psi_{n}} \right|
	& = \prod_{k=2}^{n+1}| 1-\alpha_{l+k}| \left|\braket{\psi_{1}}{A\psi_{1}}  -  \braket{\psi_{0}}{A\psi_{0}} \right| \nonumber \\
	& \leq 2\|A\| \prod_{k=2}^{n+1}|1-\alpha_{l+k}| . \label{iter_bound}
	\end{align}	
	
	We bound $1-\alpha_N$ by rewriting the closed form of $\alpha_N$ from \eqref{def:alpha} in terms of $\beta = -\mu \in (0,1) $, i.e.
	\begin{equation} \label{alpha_beta}
	1-\alpha_N = 
	\begin{cases}
	\beta + \frac{\beta^N(1-\beta^2)}{1+\beta^{N+1}}, & N \; \textrm{even}\\
	\beta  - \frac{\beta^N(1-\beta^2)}{1-\beta^{N+1}}, & N \; \textrm{odd}
	\end{cases}
	\end{equation}
	from which the following inequalities hold in the case $N\geq 1$ is odd:
	\begin{align}
	0 < 1-\alpha_N &< \beta \label{alpha_1}\\
	\beta< 1-\alpha_{N+1}&\leq\beta(1+\beta-\beta^3) \label{alpha_2}\\
	(1-\alpha_N) (1-\alpha_{N+1}) & \leq \beta^2. \label{alpha_3}
	\end{align}
	Note that \eqref{alpha_1} and the lower bound in \eqref{alpha_2} are immediate from \eqref{alpha_beta} since $\beta \in (0,1) $. Since $N+1\geq 2$ is even, the upper bound in \eqref{alpha_2} follows from
	\[
	1-\alpha_{N+1}=\beta + \frac{\beta^{N+1}(1-\beta^2)}{1+\beta^{N+2}} \leq \beta + \beta^2(1-\beta^2) = \beta (1+\beta-\beta^3) . 
	\] 
	For a proof of \eqref{alpha_3}, since $N$ is odd one can use \eqref{alpha_beta} to rewrite the product as
	\[
	(1-\alpha_N)(1-\alpha_{N+1}) = \beta^2 - \frac{\beta^N(1-\beta^2)}{(1-\beta^{N+1})(1+\beta^{N+2})}\left[1-\beta+\beta^N+\beta^{N+2}\right]<\beta^2
	\]	
	where the last inequality holds since $\beta \in (0,1) $.	
	
	As a consequence of~\eqref{alpha_1}--\eqref{alpha_3}, the product in \eqref{iter_bound} satisfies
	\begin{equation}\label{eq:prod_decay_bd}
	\prod_{k=2}^{n+1}(1-\alpha_{l+k})\leq (1+\beta-\beta^3)\beta^{n}
	\end{equation}
	where \eqref{alpha_2} is only needed if the first term, $l+2$, is even. Substituting \eqref{eq:prod_decay_bd} into \eqref{iter_bound} proves the result for $m=n+1$. 
	
	In the case $m>n+1$, the recursion relation~\eqref{eq:recursion} states
	$
	\vp_{l+m} = \vp_{l+n+1}\otimes\vp_{m-n-1} + \lambda \vp_{l+n}\otimes \ket{\pmb{\sigma}_d}\otimes \vp_{m-n-2} $. 
	If we expand $\vp_{m-n-1}$ in the occupation basis, we see that the particle content of the first three sites never agrees with the last three sites of $\ket{\pmb{\sigma}_d}$. Thus, we once again have $\left(\1\otimes \langle{\vp_{m-n-1}|}\right)\left(\ket{\pmb{\sigma}_d}\otimes \ket{\vp_{m-n-2}}\right)=0$. Hence, as in the first case, we can use this recursion relation to rewrite 
	\[
	\braket{\psi_m}{A\psi_m} = c\braket{\psi_{n+1}}{A\psi_{n+1}} + (1-c)\braket{\psi_{n}}{A\psi_{n}} \;\; \text{with} \;\; c = \frac{\|\vp_{l+n+1}\|^2\|\vp_{m-n-1}\|^2}{\|\vp_{l+m}\|^2}.
	\]
	Moreover, using this recursion relation to calculate~$ \|\vp_{l+m}\|^2$, we deduce that $0<c<1$. Since 
	\[
	\left| \braket{\psi_m}{A\psi_m}-\braket{\psi_n}{A\psi_n}\right| = c\left| \braket{\psi_{n+1}}{A\psi_{n+1}}-\braket{\psi_n}{A\psi_n}\right|,
	\]
	the claim~\eqref{one_step_decay} thus follows from the case $m=n+1$. This completes the proof. 
\end{proof} 

The proof of Lemma~\ref{lem:vp_shrink} only used the recursion relation~\eqref{eq:recursion} and an orthogonality property. Analogous properties hold if we instead consider a sequence of finite-volumes that increases to the left. As such, a similar argument shows that the bound in \eqref{one_step_decay} still holds for observables supported in $\Lambda_k' = [a-3k,b]$ if we replace the vectors $\psi_n$ with 
\begin{equation}\label{eq:mirror}
\psi'_n = \widehat{\vp}_{n+k}\otimes \xi \;\; \text{or} \;\; \psi'_n =\widehat{\vp}_{n+k}^{(j)} \; \; \mbox{if $ j \in \{ 1,2 \} $,} 
\end{equation}
where  $ \xi \in \mathcal{H}_{\Lambda_0} $ is some fixed normalized vector. Taking this for granted, we are now ready to prove Theorem~\ref{thm:exp_decay}. 

\begin{proof}[Proof of Theorem~\ref{thm:exp_decay}] 
	Fix a root tiling $R\in \cR_{\Lambda}$. We assume without loss of generality that $\max X=: x < y := \min Y$, and distinguish the following two cases suggested by 
	the canonical form~\eqref{eq:genVMDform} of the VMD-state $ \psi_\Lambda(R) $, corresponding to whether or not the interval $ [x,y] $ contains a void:
	\begin{enumerate}
	\itemsep0pt
	\item[(i)] Either the interval $ \Lambda$ decomposes into two consecutive intervals 
	$\Lambda_l\cup \Lambda_r$ with $x\in \Lambda_l $, $y\in \Lambda_r$ and
		\begin{equation}\label{eq:factorexp}
		\psi_\Lambda(R) = \psi_{\Lambda_l}(R_l) \otimes \psi_{\Lambda_r}(R_r)
		\end{equation}
		where $R_k\in\cR_{\Lambda_k}$ is the restriction of $R$ to $\Lambda_k$.
	\item[(ii)] Or the interval $\Lambda $ decomposes into three consecutive intervals (possibly empty) $\Lambda_l \cup \Lambda_m \cup \Lambda_r$ such that $ \Lambda_m $ contains $ [x+5,y-3] $ and there is an $ N \geq 2 $ (since $ y-x \geq 13 $) and $j\in\{1,2,3\}$ so that: 
		\begin{equation}\label{eq:stringphi}
		\psi_\Lambda(R) = \psi_{\Lambda_l}(R_l) \otimes \varphi^{(j)}_N \otimes \psi_{\Lambda_r}(R_r)
		\end{equation}
		where $R_k\in\cR_{\Lambda_k}$ is the restriction of $R$ to $\Lambda_k$. Here, we use the convention that $\psi_{\Lambda'}(R')=1$ if $\Lambda'=\emptyset$. The cases $ j \in \{1,2\} $ only appear if $ \Lambda_r = \emptyset $ and $ B_r = B^r_{jm} $; otherwise $ j  = 3 $ and can be omitted. 
	\end{enumerate}
	In the first case we have $\supp A_1 \subset X\subset \Lambda_l$ and $\supp A_2 \subset Y\subset \Lambda_r$ and  
	\begin{equation}\label{exp_decay_void}
	\omega_{R}^\Lambda(A_1A_2)= \omega_{R_l}^{\Lambda_l}(A_1) \, \omega_{R_r}^{\Lambda_r}(A_2) =\omega_{R}^\Lambda(A_1)\,  \omega_{R}^\Lambda(A_2) .
	\end{equation}
	Thus~\eqref{exp_decay} holds trivially.
	
	The  second case is more involved and we let
	\[ m := \left\lfloor \frac{x+y}{2}\right\rfloor + 1 \]
	be the unique midpoint in $  [x+5,y-3] $, which is closer to the left of the interval, and
	set $ L\geq 1  $ the number of monomers in the root tiling of $ \varphi_N^{(j)} $, which are contained in   $ (-\infty,m] $. 
	We may now decompose $ N = L + R $  and 
	and use the recursion relation~\eqref{eq:recrel3} 
	to obtain the following decomposition of the normalized VMD state
	\begin{equation} \label{eq:decomp}
	\widehat{\psi}_{\Lambda}(R) = c_1 \,  \widehat{\psi}_{\Lambda_1} \otimes \widehat{\psi}_{\Lambda_2} + c_2 \,  \widehat{\psi}_{\Lambda_1'} \otimes  | \pmb{\sigma}_d \rangle \otimes  \widehat{\psi}_{\Lambda_2'}
	\end{equation}
	into states supported on $ \Lambda_1 \supset \Lambda_1' $ and $ \Lambda_2\supset \Lambda_2' $ defined by
	\begin{align}
	\psi_{\Lambda_1}&:= \psi_{\Lambda_l}(R_l) \otimes \varphi_L , \qquad & \psi_{\Lambda_2}& :=  \varphi_R^{(j)} \otimes  \psi_{\Lambda_r}(R_r)  \notag \\
	\psi_{\Lambda_1'} & :=  \psi_{\Lambda_l}(R_l) \otimes \varphi_{L-1} , \qquad & \psi_{\Lambda_2'}& :=  \varphi_{R-1}^{(j)} \otimes  \psi_{\Lambda_r}(R_r) .
	\end{align}
	The complex coefficients $c_1, \, c_2$ are such that $ |c_1|^2 + |c_2|^2 = 1 $ by the orthonormality of the states in the decomposition~\eqref{eq:decomp} and, moreover,
	\[
	c_1  := \frac{\| \psi_{\Lambda_1}\| \, \|  \psi_{\Lambda_2} \| }{\| \psi_{\Lambda}(R) \| }.
	\]
	
	Since $ \supp A_ 1 \subset X \subset \Lambda_1' $ and $  \supp A_ 2 \subset Y \subset \Lambda_2' $, the two states on the right side of~\eqref{eq:decomp} remain orthogonal even after the application of $ A_1 $ and $ A_2 $. We thus have 
	\begin{equation}
	\omega_{R}^\Lambda(A_1A_2) = |c_1|^2 \, \langle   \widehat{\psi}_{\Lambda_1} , A_1  \widehat{\psi}_{\Lambda_1} \rangle \,  \langle \widehat{\psi}_{\Lambda_2} , A_2  \widehat{\psi}_{\Lambda_2} \rangle  +  |c_2|^2 \, \langle   \widehat{\psi}_{\Lambda_1'} , A_1  \widehat{\psi}_{\Lambda_1'} \rangle \,  \langle \widehat{\psi}_{\Lambda_2'} , A_2  \widehat{\psi}_{\Lambda_2'} \rangle .
	\end{equation}
	To relate the expectations in the right side to expectations involving the original VMD state, we note that 
	\begin{align}
	\omega_{R}^\Lambda(A_1) & = \langle  \widehat{\psi}_{\Lambda_l}(R_l) \otimes \widehat{\varphi}_N^{(j)} , A_1\,   \widehat{\psi}_{\Lambda_l}(R_l) \otimes \widehat{\varphi}_N^{(j)} \rangle \notag \\
	\omega_{R}^\Lambda(A_2) & =  \langle \widehat{ \varphi}_N^{(j)} \otimes  \widehat{\psi}_{\Lambda_r}(R_r)  , A_2 \, \widehat{ \varphi}_N^{(j)} \otimes  \widehat{\psi}_{\Lambda_r}(R_r)  \rangle . 
	\end{align} 
	Lemma~\ref{lem:vp_shrink} and its mirror analogue (cf.~\eqref{eq:mirror}) are therefore applicable. Set $ L ' \geq 1 $ to be the number of monomers in the root tiling of $ \varphi_L $ which are supported in $ [x+1, m] $. Then by Lemma~\ref{lem:vp_shrink}: 
	\begin{equation}
	\max\left\{ \left| \omega_{R}^\Lambda(A_1) -  \langle   \widehat{\psi}_{\Lambda_1} , A_1  \widehat{\psi}_{\Lambda_1} \rangle \right| , \left| \omega_{R}^\Lambda(A_1) -  \langle   \widehat{\psi}_{\Lambda_1'} , A_1  \widehat{\psi}_{\Lambda_1'} \rangle \right| \right\}  \leq 4 \| A_1 \| \, \beta^{L'-1} . 
	\end{equation}
	Likewise, if $ R' \geq 1$ stands for the number of monomers in the root tiling of  $ \varphi_R$ which are contained in $ (-\infty , y-1] $:
	\begin{equation}
	\max\left\{ \left| \omega_{R}^\Lambda(A_2) -  \langle   \widehat{\psi}_{\Lambda_2} , A_2  \widehat{\psi}_{\Lambda_2} \rangle \right| , \left| \omega_{R}^\Lambda(A_2) -  \langle   \widehat{\psi}_{\Lambda_2'} , A_2  \widehat{\psi}_{\Lambda_2'} \rangle \right| \right\}  \leq  4 \| A_2 \| \, \beta^{R'-1} 
	\end{equation}
	Our definition of $m$ guarantees that  $R' \geq L' $, and so the above inequalities result in the estimate
	\begin{equation}
	\left| \omega_{R}^\Lambda(A_1A_2)- \omega_{R}^\Lambda(A_1)\, \omega_{R}^\Lambda(A_2) \right| \leq \, 8 \|A_1\|\|A_2\| \, \beta^{L'-1} .
	\end{equation}
	The proof is completed by recalling $ \beta = e^{-3 c(\lambda) } $ and noting that
	\[
	L' \geq \left\lfloor \frac{ m - (x+4) }{3} \right\rfloor \geq  \left\lfloor \frac{ d(X,Y)  }{6} - \frac{4}{3} \right\rfloor   \geq \frac{d(X,Y)}{6}  - \frac{7}{3} , 
	\]
	which implies $ 3 (L' - 1) \geq   \frac{d(X,Y)}{2} - 10 $.
\end{proof}

\subsection{Proof of Theorem~\ref{thm:expcluster} }\label{subsec:Proofexp}

\begin{proof}[Proof of Theorem~\ref{thm:expcluster} ]
	Recall that the fermionic and spin VMD states associated with the same root tiling $R$ are related to each other up to a sign via the Jordan-Wigner transformation, cf. \eqref{eq:JW_VMD}. Moreover, for a fixed $ X \subset \Lambda $, any even fermionic observable $ \mathcal{A}^e_X $ is mapped under Jordan-Wigner to a spin observable in $ \mathcal{A}_X $. Thus, the first claim in Theorem~\ref{thm:expcluster}, i.e.~\eqref{eq:ABexpdecay}, is a consequence of Theorem~\ref{thm:exp_decay}. 
	
	Since any (fermionic) VMD state $\psi_\Lambda(R)$ has a fixed number of particles, the expectation of any odd fermionic observable $ A \in  \mathcal{A}^o_X $ with $ X \subset \Lambda $ in this state vanishes, i.e.
	\begin{equation}\label{eq:eq:oddferm}
	\omega_{R}^\Lambda \left( A \right) = 0 \, . 
	\end{equation}
	This establishes the last equality in~\eqref{eq:zeroodd} and as well as the claim for  $ A_1 \in  \mathcal{A}^o_X  $ and $ A_2 \in  \mathcal{A}^e_Y $.
	
	We are thus left with the case $ A_1 A_2 $ with $ A_1 \in  \mathcal{A}^o_X  $ and $ A_2 \in  \mathcal{A}^o_Y $. Transforming \eqref{eq:zeroodd} into its spin representation, the Jordan-Wigner transformation maps this product to a spin observable of the form
	\[ 
	A := A_1 \Big(  \prod_{x < k < y } \sigma^3_k \Big)  A_ 2
	\] 
	for some spin observables $ A_1 \in \mathcal{A}_X $ and $ A_2 \in \mathcal{A}_Y$. Without loss of generality we have assumed here that $ \max X =: x < y := \min Y $. 
	
	We again distinguish  the two cases in the proof of Theorem~\ref{thm:exp_decay}. If the VMD state factorizes, i.e. \eqref{eq:factorexp} holds, then using the same notation as in \eqref{exp_decay_void} we have 
	\begin{equation}\label{eq:factorferm_void} 
	\omega_{R}^\Lambda(A)= \omega_{B_l,V_l}^{\Lambda_l}(A_{1,l}) \, \omega_{B_r,V_r}^{\Lambda_r}(A_{2,r}) 
	\end{equation} 
	where 
	\begin{equation}\label{eq:partialstring}
	A_{1,l} := A_1 \Big(  \prod_{x < k \in \Lambda_l  } \sigma^3_k \Big)  , \qquad A_{2,r} := \Big(  \prod_{\Lambda_r \ni k < y} \sigma^3_k \Big)  A_2 . 
	\end{equation}
	Since $ A_{2,r} $ maps back under Jordan-Wigner to an odd observable on $ \Lambda_r $, the last factor in~\eqref{eq:factorferm_void}  vanishes by~\eqref{eq:eq:oddferm}.
	
	In the other case in which~\eqref{eq:stringphi} holds, we repeat the construction described there and use the recursion relation 
	to once again arrive at the decomposition~\eqref{eq:decomp}. Since all configurations are written in the eigenbasis of $\sigma^3$, applying $ A $ to the two vectors on the right side of~\eqref{eq:decomp} will again result in two orthogonal vectors. In particular, $ \sigma^3_1 \cdots \sigma^3_6 | \pmb{\sigma}_d \rangle =  | \pmb{\sigma}_d \rangle  $ and we have 
	\begin{equation}
	\omega_{R}^\Lambda(A) = |c_1|^2 \, \langle   \widehat{\psi}_{\Lambda_1} , A_{1,l}  \widehat{\psi}_{\Lambda_1} \rangle \,  \langle \widehat{\psi}_{\Lambda_2} , A_{2,r} \widehat{\psi}_{\Lambda_2} \rangle  +  |c_2|^2 \, \langle   \widehat{\psi}_{\Lambda_1'} , A_{1,l}'  \widehat{\psi}_{\Lambda_1'} \rangle \,  \langle \widehat{\psi}_{\Lambda_2'} , A_{2,r}'  \widehat{\psi}_{\Lambda_2'} \rangle 
	\end{equation}
	where the operators in the first term on the right are defined as in~\eqref{eq:partialstring} and the primed operators are defined similarly with $ \Lambda_\# $ substituted by $ \Lambda_\#' $ for both $ \# \in \{ l, r \} $. Since both $ A_{2,r} $ and $ A_{2,r}' $ corresponds to an odd observable on $ \Lambda_r $ and $ \Lambda_r' $, respectively, the terms on the right hand side are again zero.
	This completes the proof. 
\end{proof}

\subsection{A pure ground state without exponential decay of correlations}\label{subsec:noexpdecay}

In the previous section we showed that the finite-volume VMD states have exponential clustering. This result extends to the infinite-volume VMD states since Theorem~\ref{thm:exp_decay} is uniform in both volume and tiling. Infinite-volume states $ \omega:  \cA_{\bZ}  \to \bC $ are defined on the $C^*$-algebra of quasi-local observables  $\cA_{\bZ} $, which coincides with the norm-closure of the union of all bounded operators supported in finite volume. The question remains, though, if all pure infinite-volume ground states exhibit exponential clustering. In this subsection, we answer this question negatively by providing examples of pure infinite-volume ground states whose correlations do not decay exponentially -- in fact, the decay can be made arbitrarily slow.\\

For the construction, we start from the infinite-volume VMD state $\omega_{3k} :\cA_{\bZ}\to\bC$ defined by a single void at $3k\in \bZ$:
\begin{equation}\label{eq:omegak}
\omega_{3k}(A) := \lim_{n\to\infty}\braket{\widehat{\psi}_{k}^{(n)}} {A \widehat{\psi}_{k}^{(n)}} \quad \text{where} \quad  \psi_{k}^{(n)} = \vp_{n+k}\otimes \ket{0}_{3k}\otimes \vp_{n-k} \in \cH_{[-3n,3n]}.
\end{equation}
The subscript is used to record the location of the void.   

To illustrate the idea, we first turn to the special case $\lambda=0$, for which it is clear that
$  \psi_{k+1}^{(n)} = S_k  \psi_{k}^{(n)} $
with $S_k := \ket{0100}\langle{1000} |\in\cA_{[3k,3(k+1)]}$. Since $S_k$ has finite support, the set of states $\{\omega_{3k}\}_{k\in\bZ}$  belong to the same GNS representation. 
We may therefore start from the GNS representation  $(\pi_0,\cH_0,\Omega_0)$ of the state $\omega_0$, and denote by $\Omega_{3k}\in \cH_0$ a state such that $\omega_{3k}(A) = \braket{\Omega_{3k}}{\pi_0(A) \Omega_{3k}}$ for any $ A \in \cA_{\bZ} $; since $ \lambda = 0 $ this is formally
\[
\Omega_{3k} =\bigotimes_{n_1\in\bN}\ket{100}\otimes \ket{0}_{3k} \otimes\bigotimes_{n_2\in\bN}\ket{100} . 
\]
For a normalized sequence $(c_k)\in\ell^2(\bZ)$ the state
$ \omega(A) := \braket{\Omega}{A\Omega} $ where $ \Omega := \sum_{k\in \bZ} c_k \Omega_{3k}\in \cH_0 $
then defines a pure state on $\cA_{\bZ}$. The state $\Omega_{3k}$ represents a dislocation in the (squeezed) Tao-Thouless state at $x=3k$ and the linear combination $ \Omega $ smears this dislocation over the volume.
To compute its two-point correlation function for the on-site number operator $n_x$, we note that for $ x \neq y $:
\begin{equation}\label{eq:productrule}
\omega(n_x n_{y}) = \braket{\Omega}{n_x n_y \Omega}  = \sum_{k\in \mathbb{Z} } |c_k|^2  \omega_{3k}(n_x n_y) = \sum_{k\in \mathbb{Z} } |c_k|^2  \omega_{3k}(n_x) \,  \omega_{3k}(n_y)  
\end{equation}
where the last equality holds because $\lambda=0$. Since $ \omega_{3k}(n_{3j}) = 1 $ for $ k > j $ and $  \omega_{3k}(n_{3j})  = 0 $ otherwise, we conclude  that for  the sites $x=0<3j=y$,
\begin{align}\label{eq:noexpdecay0} 
\omega(n_0n_{3j})-\omega(n_0)\omega(n_{3j}) & = \sum_{k\in \mathbb{Z} } \sum_{l\in \mathbb{Z}} |c_k|^2 |c_l|^2  \omega_{3k}(n_{3j}) \left(  \omega_{3k}(n_0)  -  \omega_{3l}(n_0) \right) \notag \\
& =  \sum_{k> j  } |c_k|^2 \sum_{l\leq 0 }  |c_l|^2  \notag \\
& =  p(0)(1-p(j)) \;\; \text{where} \;\; p(j):= \sum_{k\leq j}|c_k|^2.
\end{align}
As $(c_k)$ is normalized in $\ell^2(\bZ)$, we have $\lim_{j\to\infty} p(j) =1$. The rate of convergence, and hence the decay of correlations, can be arbitrarily slow based on the choice of sequence.\\

To extend the above construction to $ \lambda \neq 0 $, we again start from the family of infinite-volume VMD states $\omega_{3k} :\cA_{\bZ}\to\bC$ defined 
through~\eqref{eq:omegak}, which have a void at $ 3k $. The state $\omega_0$ defines a GNS representation, which we again denote by  $(\pi_0,\cH_0,\Omega_0)$.
To proceed, we note that states for distinct $ k \in \mathbb{Z} $ can again be represented on the same Hilbert space.
\begin{proposition}
	The set of states $\{\omega_{3k}\}_{k\in\bZ}$ belong to the GNS representation $(\pi_0,\cH_0,\Omega_0)$.
\end{proposition}
As in the case $ \lambda = 0 $, the proof  proceeds by connecting $  \psi_{k}^{(n)} $ and $  \psi_{k+1}^{(n)} $ by a local operator $ S_k $ which, for $\lambda\neq 0$, is supported on the larger set  $[3(k-1),3(k+2)] $.
The action of this operator can be read by using the recursion relation~\eqref{eq:recursion} to expand both vectors to the right and left of the block $ [3k,3k+3] $. More precisely, we use as truncation points in the recursion relation the edge to the left of $ 3(k-1) $ (i.e. one monomer to the left of $3k$) as well as the edge to the right of $ 3(k+2) $ (i.e. one monomer to the right of $3k+3$). \\

Denoting by $\Omega_{3k}\in \cH_0$ a state such that $\omega_{3k}(A) = \braket{\Omega_{3k}}{\pi_0(A) \Omega_{3k}}$ the state 
\begin{equation}\label{decay_state}
\omega(A) = \braket{\Omega}{A\Omega} \;\; \text{where}\;\; \Omega = \sum_{k\in \bZ} c_k \Omega_{3k}\in \cH_0,
\end{equation}  
again defines a pure state on $\cA_{\bZ}$ for any normalized  $(c_k)\in\ell^2(\bZ)$. The first two equalities in~\eqref{eq:productrule} remain valid for any $ x\neq y $. However, unless $ x \leq 3k \leq y $ the expectation value $ \omega_{3k}(n_x n_y) $ does not factorize. 

To simplify the presentation and computations, we again concentrate on the case $ x = 0 $ and $ y = 3j $ with some $ j > 1 $. 
The following proposition lists all relevant expectations in terms of  $\mu_{\pm} = (1\pm \sqrt{1+4|\lambda|^2})/2$, the ratio $ \mu = \mu_-/\mu_+ $ and the difference $ \Delta\mu := \mu_+ - \mu_- = \sqrt{4|\lambda|^2+1} $, cf.~Lemma~\ref{lem:normalization}. 
\begin{proposition}\label{prop:computations}
	For all $ k \in \mathbb{Z} $ and 
	all $ j \in \mathbb{Z}  $:
	\begin{equation}
	\omega_{3k}(n_{3j}) = \frac{1- \mu^{|j-k|} }{\Delta\mu} \times \begin{cases} 1 & \mbox{if} \quad j < k, \\
	0 & \mbox{if} \quad   j= k, \\
	\frac{|\lambda|^2}{\mu_+} & \mbox{if} \quad   j >  k .
	\end{cases} 
	\end{equation}
	Moreover, if  $ j > 1$, we also have
	\begin{equation}\label{eq:productexp}
	\omega_{3k}(n_0 n_{3j}) = \begin{cases}  \frac{1}{(\Delta\mu)^2}   
	\left( 1- \mu^j \right)   \left( 1- \mu^{k-j} \right) & \mbox{if} \quad j < k, \\  
	\omega_{3k}(n_0) \,  \omega_{3k}(n_{3j})  &  \mbox{if} \quad 0 \leq k\leq  j, \\
	\frac{|\lambda|^4}{\mu_+^2 (\Delta\mu)^2 }  \left( 1- \mu^{j-1} \right)   \left( 1- \mu^{-k} \right)  & \mbox{if} \quad  k < 0 . \end{cases} 
	\end{equation}
\end{proposition}
The proof proceeds by elementary, explicit computations starting from the finite-volume states in~\eqref{eq:omegak} and using the recursion relation~\eqref{eq:recursion} as well as the explicit normalization of VMD states from Lemma~\ref{lem:normalization}. \\

From Proposition~\ref{prop:computations} we conclude that  the expectation of the product~\eqref{eq:productexp} approximately factorizes with an error that is exponential in $ j  $:
\begin{equation}
\left| \omega_{3k}(n_0 n_{3j})  -  \omega_{3k}(n_0) \,  \omega_{3k}(n_{3j}) \right| \leq \frac{ 4 \max\{ 1, |\lambda|^2 \} } { (\Delta\mu)^2 }  |\mu|^j  . 
\end{equation}
Up to this error, the truncated correlation function $ \omega(n_0 n_{3j})  - \omega(n_0) \omega(n_{3j}) $ is thus given by the absolutely convergent double series
\begin{align}
\sum_{k\in \mathbb{Z} } \sum_{l\in \mathbb{Z}} & |c_k|^2 |c_l|^2  \omega_{3k}(n_{3j}) \left(  \omega_{3k}(n_0)  -  \omega_{3l}(n_0) \right) \label{eq:noexpdecay1} .
\end{align}
This series still retains the main features as in the case of $ \lambda = 0 $, for which the explicit formula~\eqref{eq:noexpdecay0} holds. As a specific example, one may choose $ c_k = 0 $ if $ k < 0 $ and $ c_0 \neq 0 $ for which, setting $S_\mu := \sum_{l \geq 0}|c_l|^2\mu^l \leq 1$,
\begin{align}
\sum_{k\geq 0 } \sum_{l\geq 0} |c_k|^2 |c_l|^2 & \omega_{3k}(n_{3j}) \left(  \omega_{3k}(n_0)  -  \omega_{3l}(n_0) \right) \notag\\
& = \frac{1}{\Delta\mu} \sum_{k\geq 0} |c_k|^2\omega_{3k}(n_{3j})(S_\mu-\mu^k) \notag\\
& =  \frac{1}{(\Delta\mu)^2} \sum_{k\geq 0}|c_k|^2(1-\mu^{|j-k|})(S_\mu-\mu^k)\left(\delta_{k>j} +\frac{|\lambda|^2}{\mu_+}\delta_{k<j}\right) .\label{eq:noexpdecay2}
\end{align}
For small $ \lambda $, we have $|\mu| <<1$ and $S_\mu \approx |c_0|^2$. Hence, the leading term in~\eqref{eq:noexpdecay2} is given by (up to corrections which are exponential in $ j $):
\[
\frac{S_\mu}{(\Delta\mu)^2} \sum_{k > j } |c_k|^2.
\]
This sum tends to zero as $ j \to \infty $, but the rate can be tuned to an arbitrarily slow algebraic decay.

\appendix

\section{Estimates for $f(|\lambda|^2)$}\label{app:f_estimates}

The estimate for the gap given in Theorem \ref{thm:gap} involves a function $f(r)$ ($r=|\lambda|^2$), which is defined in Lemma \ref{lem:epsilon_calc} as the supremum
over $n$ of the family of functions $f_n$ given by
$$
f_n(r) =r\alpha_{n-2}(r)\alpha_n(r)\left[\frac{(1-\alpha_{n-1}(r)(1+r))^2}{1+2r} +\frac{r\alpha_{n-3}(r) (1-\alpha_{n-1}(r))^2}{1+r}\right] , \quad r\geq 0, n\geq 4,
$$
where $\alpha_n(r)$ is defined in \eq{def:alpha}. We are interested in finding the range of $r$ where $f(r) < 1/3$. We do this by considering a controlled approximation of $f$, $f^{(n)}, n\geq 4$, that is easy to calculate numerically with machine precision:
\be
f(r) = \lim_{n\to\infty}f^{(n)}(r), \mbox{ with } f^{(n)}(r) = \max_{4\leq m \leq n} f_m(r).
\label{f70}\ee
\begin{proposition}
For all $r\in[0,35]$, and $n\geq 73$, one has $|f(r) - f^{(n)}(r)|\leq .0052$.
\end{proposition}
\begin{proof}
The convergence $f^{(n)} \to f$ is not uniform but it is uniform on compact intervals. To simplify the estimates, we will assume $r\in [0,35]$ throughout the proof, 
and give an explicit estimate for $|f(r) - f^{(n)}(r)|$, $n\geq 73$, on that interval.

The only subtlety to consider is that the limit
\begin{equation}\label{limit_N_alpha}
\lim_{n\to \infty} \alpha_n(r) = \mu_+^{-1} = \frac{2}{1+\sqrt{1+4r}}
\end{equation}
is neither monotone in $n$ nor uniform in $r\in [0,\infty)$. An elementary calculation starting from the expressions for $\alpha_n(r)$, $\mu_\pm(r)$ and 
$$
\mu(r) = \frac{1-\sqrt{1+4r}}{1+\sqrt{1+4r}}
$$
shows that for all $n\geq n_0$ and $0\leq r\leq 35$ we have
$$
|\alpha_n(r) - \alpha_{n_0} (r)|\leq 8 |\mu(r)|^{n_0}.
$$
which implies $|\alpha_n(r) - \alpha_{n_0} (r)|\leq 10^{-4}$ for $n_0\geq 70$.

To bound the dependence of $f_n(r)$ on $n$, it is convenient to introduce the variables $\beta_n=\alpha_{n-1},\gamma_n = \alpha_{n-2}$, and
$\delta_n = \alpha_{n-3}$, and note that $f_n(r) = g(\alpha_n,\beta_n,\gamma_n,\delta_n, r)$ where
$$
g(\alpha,\beta,\gamma,\delta, r) = \alpha\gamma r \left[ \frac{(1-\beta(1+r))^2}{1+2r} + \frac{\delta r (1-\beta)^2}{1+r}\right].
$$
The variation of $f_n(r)$ for $n\geq 73$ can then be bounded by the variation of 
$g(\alpha,\beta,\gamma,\delta,r)$ for $\alpha,\beta,\gamma,\delta \in [ \mu_+^{-1} -10^{-4}, \mu_+^{-1} + 10^{-4}]$ and $r\in [0,35]$. It is elementary to 
obtain estimates on the derivatives of $g$ with respect to $\alpha,\beta,\gamma,\delta$, in the same range of variables. One finds $|\partial_\tau 
g(\alpha,\beta,\gamma,\delta,r)| \leq 13$, for $\tau \in \{\alpha,\beta,\gamma,\delta\}$. These estimates imply that the variation of $g$ as it depends 
on $n\geq 73$, is bounded by $4 \cdot 13 \cdot 10^{-4}$. This proves the proposition.
\end{proof}

Figure \ref{fig:fplot} shows the result of a numerical calculation of $f^{(73)}(r)$ for $r\in [0,35]$. From that calculation and the above proposition, it is clear that 
$f(|\lambda|^2) < 1/3$ for $|\lambda| \in [0,5.3]$, and that this interval is close to optimal.

\bigskip
\minisec{Acknowledgements }

{\small BN and SW thank Duncan Haldane for stimulating our interest in the spectral gap problem for the FQHE pseudo-potential models. SW acknowledges useful discussions and pointers to the literature by Frank Pollmann and Emil Bergholtz. Our work was facilitated by
opportunities to meet and discuss, notably the Oberwolfach Workshop on Many-Body Quantum Systems, the Quantum Information Theory program at the Instituto de Ciencias Matem\'aticas, Madrid, and the workshop on Random Schr\"odinger Operators and Related Topics in Florence. BN would like to specially thank the Munich Center for Quantum Science and Technology for its warm hospitality during an extended visit and the Simons Foundation for support through the Simons Visiting Professorship program administered by the Mathematisches Forschungsinstitut Oberwolfach. This work was supported by the DFG under EXC-2111--390814868 and by the National Science Foundation (USA) under Grant DMS-1813149.}

\bibliographystyle{abbrv}  
\bibliography{FQHE}

\bigskip
\bigskip
\noindent Bruno Nachtergaele\\
Department of Mathematics\\
and Center for Quantum Mathematics and Physics\\
University of California \\
Davis, CA  95616-8633, USA\\
\verb+bxn@math.ucdavis.edu+\\

\noindent Simone Warzel\\
Munich Center for Quantum Science and Technology, and\\
Zentrum Mathematik, TU M\"{u}nchen\\
85747 Garching, Germany\\
\verb+warzel@ma.tum.de+\\

\noindent Amanda Young\\
Munich Center for Quantum Science and Technology, and\\
Zentrum Mathematik, TU M\"{u}nchen\\
85747 Garching, Germany\\
\verb+young@ma.tum.de+\\

\end{document}